\documentclass[twoside,11pt]{article}
\usepackage{amsfonts, amsmath, graphicx, float, subcaption, enumerate, xcolor, %
  algorithm, algorithmic}
\usepackage[shortlabels]{enumitem}

\usepackage{jmlr2e}
\usepackage{natbib}
\graphicspath{{figures/}}

\newcommand{\0}{\mathbf{0}}

\renewcommand{\H}{\mathbf{H}}
\newcommand{\A}{\mathbf{A}}
\renewcommand{\b}{\mathbf{b}}
\newcommand{\B}{\mathbf{B}}
\newcommand{\I}{\mathbf{I}}

\newcommand{\M}{\mathbf{M}}

\newcommand{\s}{\mathbf{s}}

\renewcommand{\v}{\mathbf{v}}
\newcommand{\V}{\mathbf{V}}

\newcommand{\x}{\mathbf{x}}

\newcommand{\Z}{\mathbf{Z}}

\newcommand{\bbeta}{\boldsymbol{\beta}}
\newcommand{\cbeta}{\check{\boldsymbol{\beta}}}
\newcommand{\hbeta}{\hat{\boldsymbol{\beta}}}
\newcommand{\tbeta}{\tilde{\boldsymbol{\beta}}}
\newcommand{\bLambda}{\boldsymbol{\Lambda}}
\newcommand{\bkappa}{\boldsymbol{\kappa}}

\newcommand{\bSigma}{\boldsymbol{\Sigma}}
\newcommand{\bvsigma}{\boldsymbol{\varsigma}}

\newcommand{\eeta}{\boldsymbol{\eta}}%

\newcommand{\onen}{\frac{1}{n}}
\newcommand{\oneN}{\frac{1}{N}}

\newcommand{\op}{o_{P}(1)}

\newcommand{\opb}{o_{{P}|\hat\bbeta_0}(1)}
\newcommand{\Op}{O_{P}(1)}
\newcommand{\Opf}{O_{{P}|\Fn}(1)}

\newcommand{\Fn}{\mathcal{D}_N}

\newcommand{\sumn}{\sum_{i=1}^{n}}
\newcommand{\sumN}{\sum_{i=1}^{N}}

\newcommand{\tp}{^{\rm T}}
\newcommand{\tr}{\mathrm{tr}}

\renewcommand{\Pr}{\mathbb{P}}
\newcommand{\Exp}{\mathbb{E}}
\newcommand{\Nor}{\mathbb{N}}
\newcommand{\Var}{\mathbb{V}}
\newcommand{\mmse}{{\mathrm{Aopt}}}
\newcommand{\mvc}{{\mathrm{Lopt}}}
\newcommand{\os}{{\mathrm{OS}}}

\newtheorem{assumption}{\bf Assumption}

\newcommand{\mle}{{\textnormal{\tiny MLE}}}%
\newcommand{\wmle}{{\textnormal{\tiny wMLE}}}%

\ShortHeadings{More Efficient Estimation with Optimal Subsamples}{HaiYing Wang}
\firstpageno{1}

\usepackage{lastpage}
\jmlrheading{20}{2019}{1-\pageref{LastPage}}{9/18; Revised
7/19}{8/19}{18-596}{HaiYing Wang}
\ShortHeadings{More Efficient Estimation with Optimal Subsamples}{Wang}

\begin{document}

\title{More Efficient Estimation for Logistic Regression with Optimal Subsamples}

\author{\name HaiYing Wang \email haiying.wang@uconn.edu \\
       \addr Department of Statistics\\
       University of Connecticut\\
       Storrs, CT 06269, USA}

\editor{Tong Zhang}

\maketitle

\begin{abstract}
  In this paper, we propose improved estimation method for logistic regression based on subsamples taken according the optimal subsampling probabilities developed in \cite{WangZhuMa2017}. Both asymptotic results and numerical results show that the new estimator has a higher estimation efficiency. We also develop a new algorithm based on Poisson subsampling, which does not require to approximate the optimal subsampling probabilities all at once. This is computationally advantageous when available random-access memory is not enough to hold the full data. Interestingly, asymptotic distributions also show that Poisson subsampling produces a more efficient estimator if the sampling ratio, the ratio of the subsample size to the full data sample size, does not converge to zero. We also obtain the unconditional asymptotic distribution for the estimator based on Poisson subsampling. Pilot estimators are required to calculate subsampling probabilities and to correct biases in un-weighted estimators; interestingly, even if pilot estimators are inconsistent, the proposed method still produce consistent and asymptotically normal estimators. 

\end{abstract}


\section{Introduction}
Extraordinary amounts of data that are collected offer unparalleled opportunities for advancing complicated scientific problems. %
However, the incredible sizes of big data bring new challenges for data analysis. 
A major challenge of big data analysis lies with the thirst for computing resources. Faced with this, subsampling has been widely used to reduce the computational burden, in which intended calculations are carried out on a subsample that is drawn from the full data, see \cite{drineas2006fast1, drineas2006fast2,
  drineas2006fast3, mahoney2009cur,Drineas:11, mahoney2011randomized,
  halko2011finding, clarkson2013low, kleiner2014scalable,
  mcwilliams2014fast, yang2015randomized}, among others. %

A key to success of a subsampling method is to specify nonuniform sampling probabilities so that more informative data points are sampled with higher probabilities. For this purpose, normalized statistical leverage scores or its variants are often used as subsampling probabilities in the context of linear regression, and this approach is termed {\it algorithmic leveraging} \citep{PingMa2014-JMLR}. It has demonstrated remarkable performance in better using of a fixed amount of computing power \citep{AMT10,MSM14_SISC}. Statistical leverage scores only contain information in the covariates and do not take into account the information contained in the observed responses. 
\cite{WangZhuMa2017} derived optimal subsampling probabilities that minimize the asymptotic mean squared error (MSE) of the subsampling-based estimator in the context of logistic regression. The optimal subsampling probabilities directly depend on both the covariates and the responses to take more informative subsamples. \cite{WangZhuMa2017} used a inverse probability weighted estimator based on the optimal subsample, where more informative data points are assigned smaller weights in the objective function. Thus, we can improve the estimation efficiency based on the optimal subsample by using a better weighting scheme. %

In this paper, we propose more efficient estimators based on subsamples taken randomly according to the optimal subsampling probabilities. We will derive asymptotic distributions to show that asymptotic variance-covariance matrices of the new estimators are smaller, in Loewner ordering, than that of the weighted estimator in \cite{WangZhuMa2017}. We also consider to use Poisson subsampling. %
Asymptotic distributions show that Poisson subsampling is more efficient in parameter estimation when the subsample size is proportional to the full data sample size. It is also computationally beneficial to use Poisson subsampling because there is no need to calculate and use subsampling probabilities for all data points simultaneously.

Before presenting the framework of the paper, we give a brief review of the emerging field of subsampling-based methods.
For linear regression, \cite{drineas2006sampling}
developed a subsampling method and focused on finding influential data
units for the least squares (LS) estimates. 
\cite{Drineas:11} developed an algorithm by processing the data with
randomized Hadamard transform and then using uniform subsampling to
approximate LS estimates. \cite{Drineas:12} developed an algorithm to
approximate statistical leverage scores that are used for algorithmic leveraging. \cite{yang2015explicit} showed that using normalized
square roots of statistical leverage scores as subsampling
probabilities yields better approximation than using original statistical leverage scores, if they are very nonuniform. The
aforementioned studies focused on developing algorithms for fast
approximation of LS estimates. \cite{PingMa2014-JMLR} considered the statistical properties of algorithmic leveraging. They derived biases and variances of
leverage-based subsampling estimators in linear regression and
proposed a shrinkage algorithmic leveraging method to improve the
performance. \cite{raskutti2014statistical} considered both the
algorithmic and statistical aspects of solving large-scale LS problems using
random sketching. \cite{WangYangStufken2018} and \cite{wang2019divide} developed an information-based optimal subdata selection method to select subsample deterministically for ordinary LS in linear regression. 
The aforesaid results were obtained
exclusively within the context of linear
models. \cite{fithian2014local} proposed a computationally efficient
local case-control subsampling method for logistic regression with
large imbalanced data. \cite{han2016local} developed a local uncertainty sampling approach for multi-class logistic regression. 
Recently, \cite{WangZhuMa2017} developed an Optimal Subsampling Method under the A-optimality Criterion (OSMAC) for logistic regression; \cite{Yao2019} and \cite{ai2018optimal} extended this method to include multi-class logistic regression and generalized linear regression models, respectively. Although they derived optimal subsampling probabilities, they did not investigate whether a better weighting scheme can further improve the estimation efficiency. 

This paper focuses on logistic regression models, which are widely used for statistical inference in many disciplines, such as business, computer science, education, and genetics, among others \citep{hosmer2013applied}. Based on optimal subsamples taken according to OSMAC developed in \cite{WangZhuMa2017}, more efficient methods, in terms of both parameter estimation and numerical computation, will be proposed. %
The reminder of the paper is organized as follows. Model setups and notations are introduced in Section~\ref{sec:notations}. The OSMAC will also be briefly reviewed in this section. Section~\ref{sec:more-effic-infer} presents the more efficient estimator and its asymptotic properties. Section~\ref{sec:poisson-sampling} considers Poisson subsampling. Section~\ref{sec:pilot-estim-pract} discusses issues related to practical implementation and summaries the methods from Sections~\ref{sec:more-effic-infer} and \ref{sec:poisson-sampling} into two practical algorithms. Section~\ref{sec:uncond-distr} gives unconditional asymptotic distributions for the estimator from Poisson subsampling. Section~\ref{sec:misspecifications} discusses asymptotic distributions with pilot and model misspecifications. Section~\ref{sec:numerical-examples} evaluates the practical performance of the proposed methods using numerical experiments. Section~\ref{sec:summary} concludes, and the appendix contains proofs and technical details.

\section{Model setup and optimal subsampling}
\label{sec:notations}
Let $y$ $\in\{0,1\}$ be a binary response variable and $\x$ be a $d$ dimensional covariate. A logistic regression model describes the conditional probability of $y=1$ given $\x$, and it has the following form,
\begin{equation}\label{eq:1}
  \Pr(y=1|\x)=p(\x,\bbeta)
  =\frac{e^{\x\tp\bbeta}}{1+e^{\x\tp\bbeta}},
\end{equation}
where $\bbeta$ is a $d\times1$ vector of unknown regression coefficients belonging to a compact subset of $\mathbb{R}^d$.

With independent full data of size $N$ from Model~\eqref{eq:1}, say, $\Fn=\{(\x_1, y_1), ..., (\x_N, y_N)\}$, 
the unknown parameter $\bbeta$ is
often estimated by the maximum likelihood estimator (MLE), denoted as $\hbeta_{\mle}$. It is the maximizer of the log-likelihood function, namely,
\begin{equation*}
  \hbeta_{\mle}=\arg\max_{\bbeta}\ \ell_f(\bbeta)
  =\arg\max_{\bbeta}\sumN
  \big\{y_i\x_i\tp\bbeta-\log\big(1+e^{\bbeta\tp\x_i}\big)\big\}.
\end{equation*}
Since there is no general closed-form solution to the MLE, Newton's method or iteratively reweighted least squares method \citep{mccullagh1989generalized} is often adopted to find it numerically. This typically takes $O(\zeta Nd^2)$ time, where $\zeta$ is the number of iterations in the optimization procedure. %
For super-large data set, the computing time $O(\zeta Nd^2)$ may be too long to afford, and iterative computation is infeasible if the data volume is larger than the available random-access memory (RAM). To overcome this computational bottleneck for the application of logistic regression to massive data, \cite{WangZhuMa2017} developed the OSMAC under the subsampling framework. %

Let $\pi_1$, ..., $\pi_N$ be subsampling probabilities such that $\sumN\pi_i=1$. Using subsampling with replacement, draw a random subsample of size $n$ %
according to the probabilities $\{\pi_i\}_{i=1}^{N}$ from the full data. We use $^*$ to indicate quantities for a subsample, namely, denote the covariates,
responses, and subsampling probabilities in a subsample as
${\mathbf{x}}^{*}_i$, $y^{*}_i$, and $\pi_i^{*}$, respectively, for
$i=1, ..., n$. \cite{WangZhuMa2017} define the weighted subsample estimator $\hbeta_w^{\pi}$ to be %
\begin{equation*}
  \hbeta_w^{\pi}=\arg\max_{\bbeta}\ell_w^*(\bbeta)
  =\arg\max_{\bbeta}\sumn\frac{
    y_i^*\bbeta\tp\x_i^*-\log\big(1+e^{\bbeta\tp\x_i^*}\big)}{\pi_i^*}.
\end{equation*}

The key to success here is how to specify the values for $\pi_i$'s so that more informative data points are sampled with higher probabilities. \cite{WangZhuMa2017} derived optimal subsampling probabilities that minimize the asymptotic MSE of $\hbeta_w^{\pi}$. They first showed that $\hbeta_w^{\pi}$ is asymptotically normal. Specifically, for large $n$ and $N$, the conditional distribution of $\sqrt{n}(\hbeta_w^{\pi}-\hbeta_{\mle})$ given the full data $\Fn$ can be approximated by a normal distribution with mean $\0$ and variance-covariance matrix $\V_N=\M_N^{-1}\V_{Nc}\M_N^{-1}$, in which
\begin{equation*}
  \M_N=\oneN\sumN\phi_i(\hbeta_{\mle})\x_i\x_i\tp, \quad
  \V_{Nc}=\oneN\sumN\frac{|y_i-p(\x_i,\hbeta_{\mle})|^2\x_i\x_i\tp}{N\pi_i},
\end{equation*}
and $\phi_i(\bbeta)=p(\x_i,\bbeta)\{1-p(\x_i,\bbeta)\}$ with $p(\x_i,\bbeta)={e^{\x_i\tp\bbeta}}/{(1+e^{\x_i\tp\bbeta})}$. 
Based on this asymptotic distribution, they derive the following two optimal subsampling probabilities 
\begin{align}
  \pi_i^{\mmse}(\hbeta_{\mle})&=\label{eq:35}
  \frac{|y_i-p(\x_i,\hbeta_{\mle})|\|\M_N^{-1}\x_i\|}
 {\sum_{j=1}^N|y_j-p(\x_j,\hbeta_{\mle})|\|\M_N^{-1}\x_j\|},
  \quad i=1, ..., N;\\
  \pi_i^{\mvc}(\hbeta_{\mle})&=\label{eq:38}
  \frac{|y_i-p(\x_i,\hbeta_{\mle})|\|\x_i\|}
  {\sum_{j=1}^N|y_j-p(\x_j,\hbeta_{\mle})|\|\x_j\|},
  \quad i=1, ..., N.
\end{align}
Here, $\{\pi_i^{\mmse}(\hbeta_{\mle})\}_{i=1}^N$ minimize $\tr(\V_N)$, the trace of $\V_N$, and this is the A-optimality criterion in optimum experimental designs \citep{atkinson2007optimum};  $\{\pi_i^{\mvc}(\hbeta_{\mle})\}_{i=1}^N$ minimize $\tr(\V_{Nc})$, %
and this is a choice of the L-optimality criterion. %
These subsampling probabilities have a lot of nice properties and meaningful interpretations. More details can be found in Section 3 of \cite{WangZhuMa2017}. %

For ease of presentation, use the following general notation to denote subsampling probabilities
\begin{equation}\label{eq:2}
  \pi_i^{\os}(\bbeta)=
  \frac{|y_i-p(\x_i,\bbeta)|h(\x_i)}
  {\sum_{j=1}^N|y_j-p(\x_j,\bbeta)|h(\x_j)},
  \quad i=1, ..., N,
\end{equation}
where $h(\x)$ is a univariate function of $\x$. {We provide some intuitions on choosing $h(\x)$. Let $L$ be a matrix with $d$ columns. Choosing $h(\x)=\|L\M_N^{-1}\x\|$ minimizes the trace of $L\V_NL\tp$, which is the conditional asymptotic variance-covariance matrix of $L\hbeta_w^{\pi}$ (scaled by $n$) given the full data $\Fn$. Two special choices of $h(\x)$ correspond to $L=\I$ (the identity matrix) and $L=\M_N$. If $L=\I$, then $h(\x)=\|\M_N^{-1}\x\|$ and $\pi_i^{\os}(\bbeta)$ becomes  $\pi_i^{\mmse}(\bbeta)$; if $L=\M_N$, then $h(\x)=\|\x\|$ and $\pi_i^{\os}(\bbeta)$ becomes $\pi_i^{\mvc}(\bbeta)$. If one is interested in a specific component of $\bbeta$, say $\beta_j$, then $L$ can be chosen as a row vector with the $j$-th element being one and all other elements being zero. With this choice, $h(\x)=\|\M_{N,\ \centerdot j}^{-1}\x\|$ where $\M_{N,\ \centerdot j}^{-1}$ means the $j$-th row of $\M_{N}^{-1}$, and the asymptotic variance of $\hbeta_{w,j}^{\pi}$ is minimized. 
  If $h(\x)=1$, then $\pi_i^{\os}(\bbeta)$'s are proportional to the local case-control subsampling probabilities \citep{fithian2014local}.}

Note that $\{\pi_i^{\os}(\bbeta)\}_{i=1}^N$ depend on the unknown $\bbeta$, so a pilot estimate of $\bbeta$ is required to approximate them. 
Let $\hbeta_0$ be a pilot estimator from a pilot subsample taken from the full data, for which we will provide more details in Section~\ref{sec:pilot-estim-pract}. The original weighted OSMAC estimator is
\begin{equation}\label{eq:3}
  \hbeta_w=\arg\max_{\bbeta}\sumn
  \frac{y_i^*\bbeta\tp\x_i^*-\log\big(1+e^{\bbeta\tp\x_i^*}\big)}
  {\pi_i^{\os}(\hbeta_0)^*}.
\end{equation}

In \cite{WangZhuMa2017}, $\hbeta_w$ has exceptional performance because $\{\pi_i^{\os}(\hbeta_0)\}_{i=1}^N$ are able to include more informative data points in the subsample. {However, we can improve the weighting scheme adopted in \eqref{eq:3}.
  Intuitively, a larger $\pi_i^{\os}(\hbeta_0)$ means that the data point $(\x_i, y_i)$ contains more information about $\bbeta$, but it has a smaller weight in the objective function in \eqref{eq:3}. This reduces contributions of more informative data points to the objective function for parameter estimation.}

The weighted estimator in \eqref{eq:3} is used because $\{\pi_i^{\os}(\hbeta_0)\}_{i=1}^N$ depend on the responses $y_i$'s and an un-weighted estimator is biased. If the bias can be corrected, then the resultant estimator can be more efficient in parameter estimation, 
{because an un-weighted estimator often has a smaller variance-covariance matrix compared with an inverse probability  weighted estimator. Intuitively, if some data points with very small values of $\pi_i^{\os}(\hbeta_0)$ are selected in the subsample, then the target function in \eqref{eq:3} would be dominated by these data points. As a result, the variance-covariance matrix of the weighted estimator would be inflated by small values of $\pi_i^{\os}(\hbeta_0)$. Note that $\pi_i$'s appear in the denominator of $\V_{Nc}$ in the asymptotic variance-covariance matrix of the weighted estimator. A major goal of this paper is to develop un-weighted estimation procedures. Interestingly, for the subsampling probabilities in \eqref{eq:2}, the bright idea proposed in \cite{fithian2014local} can be used to correct the bias of the un-weighted estimator.}

\section{More efficient estimator %
}
\label{sec:more-effic-infer}
Let $\{(\x_1^*,y_1^*), ..., (\x_n^*,y_n^*)\}$ be a random subsample of size $n$ taken from the full data using sampling with replacement
according to the probabilities $\{\pi_i^{\os}(\hbeta_0)\}_{i=1}^{N}$ defined in \eqref{eq:2}. {Using this subsample, we present a more efficient estimation procedure based on un-weighted estimator with bias correction. 
Remember that a pilot estimate is required, and we use $\hbeta_0$ to denote it. Here, we focus the discussion on the new estimation procedure and assume that $\hbeta_0$ is obtained based on a pilot subsample of size $n_0$ and it is consistent. More details about this pilot estimator will be provided in Section~\ref{sec:pilot-estim-pract}, and the scenario that $\hbeta_0$ is inconsistent will be investigated in Section~\ref{sec:pilot-estim-missp}.} The following procedure describes how to obtain the un-weighted estimator with bias correction, denoted as $\hbeta_{uw}$.

Calculate the naive un-weighted estimator
\begin{equation}\label{eq:4}
  \tbeta_{uw}=\arg\max_{\bbeta}\ell_{uw}^*(\bbeta)
  =\arg\max_{\bbeta}\sumn
  \big\{\bbeta\tp\x_i^*y_i^*-\log\big(1+e^{\bbeta\tp\x_i^*}\big)\big\},
  \end{equation}
and then let
\begin{equation}\label{eq:5}
\hbeta_{uw}=\tbeta_{uw}+\hbeta_0.
\end{equation}

The naive un-weighted estimator $\tbeta_{uw}$ in \eqref{eq:4} is biased, and the bias is corrected in \eqref{eq:5} using $\hbeta_0$. We will show in the following that $\hbeta_{uw}$ is asymptotically unbiased. This, together with the fact that $\hbeta_0$ is consistent, shows the interesting fact that $\tbeta_{uw}$ converges to $\0$ in probability as $n_0$, $n$, and $N$ go to infinity. 

To investigate the asymptotic properties, we use $\bbeta_t$ to denote the true value of $\bbeta$, and summarize some regularity conditions in the following.

\begin{assumption}\label{as:1}
  The matrix $\Exp\{\phi(\bbeta_t)h(\x)\x\x\tp\}$ is finite and
  positive-definite.
\end{assumption}
\begin{assumption}\label{as:2}
  The covariate $\x$ and function $h(\cdot)$ satisfy that
  $\Exp\{\|\x\|^2h^2(\x)\}<\infty$,
  and $\Exp\{\|\x\|^2h(\x)\}<\infty$.
\end{assumption}
\begin{assumption}\label{as:3}
  As $n\rightarrow\infty$, $n\Exp\{h(\x)I(\|\x\|^2>n)\}\rightarrow0$, 
  where $I()$ is the indicator function. 
\end{assumption}
Assumption~\ref{as:1} is required to establish the asymptotic normality. This is a commonly used assumption, e.g., in \cite{fithian2014local,WangZhuMa2017}, among others. 
Assumptions~\ref{as:2} and \ref{as:3} impose moment conditions on the covariate distribution and the function $h(\x)$. When $h(\x)=1$, if $\Exp\|\x\|^2<\infty$, then both the two conditions in Assumption~\ref{as:2} and the condition in Assumption~\ref{as:3} hold. Thus, the assumptions  required in this paper are not stronger than those required by \cite{fithian2014local}. 
When $h(\x)=\|\x\|$, by H\"{o}lder's inequality,
\begin{align*}
  n\Exp\{h(\x)I(\|\x\|^2>n)\}
  &\le n(\Exp\|\x\|^3)^{1/3}\{\Exp I(\|\x\|^2>n)\}^{2/3}\\
  &=(\Exp\|\x\|^3)^{1/3}\{n^{3/2}\Pr(\|\x\|^3>n^{3/2})\}^{2/3}.
\end{align*}
{Note that $n^{3/2}I(\|\x\|^3>n^{3/2})\le \|\x\|^3$ and $I(\|\x\|^3>n^{3/2})\rightarrow0$ in probability. Thus, if $\Exp(\|\x\|^3)<\infty$, then
  $n^{3/2}\Pr(\|\x\|^3>n^{3/2})
  =\Exp\{n^{3/2}I(\|\x\|^3>n^{3/2})\}\rightarrow0$ \citep[see Theorem 1.3.6 of][]{Serfling1980}.} 
Therefore, if $\Exp(\|\x\|^3)<\infty$, Assumption~\ref{as:3} holds. This shows that $\Exp\|\x\|^4<\infty$ implies all the three conditions required in Assumptions~\ref{as:2} and \ref{as:3}. Note that \cite{WangZhuMa2017} requires that $\Exp(e^{\v\tp\x})<\infty$ for any $\v\in\mathbb{R}^d$ in order to establish the asymptotic properties when a pilot estimate is used to approximate optimal subsampling probabilities. Thus, the required conditions in this paper are weaker than those required in \cite{WangZhuMa2017}. Assumptions~\ref{as:1} and \ref{as:2} are required in all the theorems in this paper while Assumption~\ref{as:3} is only required in Theorems~\ref{thm:1}, \ref{thm:R3}, and \ref{thm:R1}. 

\begin{theorem}\label{thm:1}
  Under Assumptions \ref{as:1}-\ref{as:3}, conditional on
  $\Fn$, if $\hbeta_0$ is consistent, then as $n_0$, $n$, and $N$ go to infinity,
  \begin{equation}\label{eq:21}
    \sqrt{n}(\hbeta_{uw}-\hbeta_{\wmle})
    \longrightarrow \Nor\big(\0,\ \bSigma_{\bbeta_t}\big),
  \end{equation}
  in distribution; furthermore, if $n/N\rightarrow0$, then
  \begin{equation}\label{eq:13}
    \sqrt{n}(\hbeta_{uw}-\bbeta_t)
    \longrightarrow \Nor\big(\0,\ \bSigma_{\bbeta_t}\big)
  \end{equation}
  in distribution, where
  \begin{equation*}
    \bSigma_{\bbeta}
    =\bigg[\frac{\Exp\{\phi(\bbeta)h(\x)\x\x\tp\}}
    {4\Phi(\bbeta)}\bigg]^{-1}, 
    \quad
    \Phi(\bbeta)=\Exp\{\phi(\bbeta)h(\x)\},
    \quad
    \phi(\bbeta)=p(\x,\bbeta)\{1-p(\x,\bbeta)\},
  \end{equation*}
  and $\hbeta_{\wmle}$ is a weighted MLE based on the full data defined as
  \begin{equation}\label{eq:65}
    \hbeta_{\wmle}=%
    \arg\max_{\bbeta}\sumN|y_i-p(\x_i,\hbeta_0)|h(\x_i)
    \big[y_i\x_i\tp(\bbeta-\hbeta_0)
  -\log\{1+e^{\x_i\tp(\bbeta-\hbeta_0)}\}\big].
\end{equation}
Here $\hbeta_{\wmle}$ satisfies that
\begin{equation}\label{eq:6}
  \sqrt{N}(\hbeta_{\wmle}-\bbeta_t)
  \longrightarrow \Nor\big(\0,\ \bSigma_{\wmle}\big),
  \end{equation}
in distribution if $\hbeta_0$ is obtained from a uniform pilot subsample of size $n_0$ such that $n_0/\sqrt{N}=o(1)$ or if $\hbeta_0$ is independent of $\Fn$, where
\begin{align*}
  \bSigma_{\wmle}
  &=[\Exp\{\phi(\bbeta_t)h(\x)\x\x\tp\}]^{-1}
    \Exp\{\phi(\bbeta_t)h^2(\x)\x\x\}
    [\Exp\{\phi(\bbeta_t)h(\x)\x\x\tp\}]^{-1}.
\end{align*}
\end{theorem}

\begin{remark}\label{remark1}
  Theorem~\ref{thm:1} shows that the un-weighted estimator $\hbeta_{uw}$ is $\sqrt{n}$-consistent to $\hbeta_{\wmle}$, a weighted MLE based on the full data in conditional probability, while Theorem 5 of \cite{WangZhuMa2017} shows that the weighted estimator $\hbeta_w$ is $\sqrt{n}$-consistent to $\hbeta_{\mle}$, the un-weighted MLE based on the full data in conditional probability. Specifically, \eqref{eq:21} implies that %
  given $\Fn$ %
  in probability,
  \begin{equation}\label{eq:19}
    \hbeta_{uw}-\hbeta_{\wmle}=O_{{P}|\Fn}(n^{-1/2}).
  \end{equation}
  The $O_{{P}|\Fn}(n^{-1/2})$ expression in \eqref{eq:19} means that for any $\epsilon>0$, there exist a $\delta_\epsilon$ such that as $n, N\rightarrow\infty$,
\begin{align*}
  &\Pr\Big\{\sup_n\Pr(\|\hbeta_{uw}-\hbeta_{\wmle}\|
    >n^{-1/2}\delta_\epsilon|\Fn)
    \le\epsilon\Big\}\rightarrow1.
\end{align*}
Note that if a sequence is bounded in conditional probability, then it is bounded in unconditional probability, i.e., if $a_n=\Opf$, then $a_n=\Op$ \citep{xiong2008some, cheng2010bootstrap}. Therefore, \eqref{eq:19} implies that $\hbeta_{uw}-\hbeta_{\wmle}=O_P(n^{-1/2})$. Similarly, \eqref{eq:6} implies that $\hbeta_{\wmle}-\bbeta_t=O_P(N^{-1/2})$. Thus, $\hbeta_{uw}-\bbeta_t=O_P(n^{-1/2}+N^{-1/2})=O_P(n^{-1/2})$, showing the $\sqrt{n}$-consistency of $\hbeta_{uw}$ to the true parameter under the unconditional distribution. %
\end{remark}

\begin{remark}
For $\hbeta_{\wmle}$, if $\hbeta_0$ is fixed, say $\hbeta_0=\bbeta_0$, then the population log-likelihood for the objective function in \eqref{eq:65} is 
\begin{align*}
  \Exp\Big(a(\x,\bbeta,\bbeta_0)
  \Big[p(\x,\bbeta-\bbeta_0)\x\tp(\bbeta-\bbeta_0)
  -\log\{1+e^{\x\tp(\bbeta-\bbeta_0)}\}\Big]\Big),
\end{align*}
where 
$a(\x,\bbeta,\bbeta_0)=\big[p(\x,\bbeta)
  \{1-p(\x,\bbeta_0)\} + \{1-p(\x,\bbeta)\}
  p(\x,\bbeta_0)\big]h(\x)$. If $h(\x)=1$, then this population log-likelihood is identical to that for the local case-control subsampling estimator. For general $h(\x)$, since it does not rely on the response variable, we expect that $\hbeta_{\wmle}$ inherits the main properties of the the local case-control subsampling estimator, including those under model misspecification. Indeed this is the case, and more details for the scenarios of misspecifications will be presented in Section~\ref{sec:misspecifications}.
\end{remark}

Theorem~\ref{thm:1} shows that, asymptotically, the distribution of $\hbeta_{uw}$ given $\Fn$ %
is centered around $\hbeta_{\wmle}$ with variance-covariance matrix $n^{-1}\bSigma_{\bbeta_t}$, and the distribution of $\hbeta_{\wmle}$ is centered around $\bbeta_t$ with variance-covariance matrix $N^{-1}\bSigma_{\wmle}$. Thus, both $n^{-1}\bSigma_{\bbeta_t}$ and $N^{-1}\bSigma_{\wmle}$ should be considered in accessing the quality of $\hbeta_{uw}$ for estimating the true parameter $\bbeta_t$. However, in subsampling setting, it is expected that $n\ll N$; otherwise, the computational benefit is minimum. Thus, $n^{-1}\bSigma_{\bbeta_t}$ is the dominating term in quantifying the variation of $\hbeta_{uw}$. If $n/N\rightarrow0$, then the variation of $\hbeta_{\wmle}$ can be ignored as stated in \eqref{eq:13}. 

Now we compare the estimation efficiency of $\hbeta_{uw}$ with that of the weighted estimator $\hbeta_w$. With the optimal subsampling probabilities $\{\pi_i^{\os}(\hbeta_{\mle})\}_{i=1}^N$, the asymptotic variance-covariance matrix (scaled by $n$), $\V_N$, for the weighted estimator $\hbeta_w$ has a form of $\V_N^\os=\M_N^{-1}\V_{Nc}^\os\M_N^{-1}$, where
\begin{equation*}
 \V_{Nc}^\os=\bigg\{\oneN\sumN|y_i-p(\x_i,\hbeta_{\mle})|h(\x_i)\bigg\} 
 \bigg\{\oneN\sumN\frac{|y_i-p(\x_i,\hbeta_{\mle})|\x_i\x_i\tp}{h(\x_i)}\bigg\}.
\end{equation*}
Note that the full data MLE $\hbeta_{\mle}$ is consistent under Assumptions~\ref{as:1}-\ref{as:2}. If $\Exp\{\|\x\|^2/h(\x)\}<\infty$, then from Lemma~\ref{lem1} in the appendix and the law of large numbers, $\V_N^\os$ converges in probability to
$\V^\os=\M^{-1}\V_{c}^\os\M^{-1}$, where
\begin{equation*}
  \M=\Exp\{\phi(\bbeta_t)\x\x\tp\}
  \quad\text{and}\quad
  \V_{c}^\os=4\Phi(\bbeta_t)
  \Exp\bigg\{\frac{\phi(\bbeta_t)\x\x\tp}{h(\x)}\bigg\}. 
\end{equation*}

Note that the asymptotic distribution of $\hbeta_w$ given $\Fn$ is centered around $\hbeta_{\mle}$. It can be shown that under Assumptions~\ref{as:1}-\ref{as:2}, 
\begin{equation*}
  \sqrt{N}(\hbeta_{\mle}-\bbeta_t)
  \longrightarrow \Nor\big(\0, \M^{-1}\big),
\end{equation*}
in distribution. 
Thus, both $n^{-1}\V^\os$ and $N^{-1}\M^{-1}$ should be considered in accessing the quality of $\hbeta_w$ for estimating the true parameter $\bbeta_t$. However, similar to the case for $\hbeta_{uw}$, $N^{-1}\M^{-1}$ is small compared with $n^{-1}\V^\os$ if $n\ll N$, and it is negligible if $n/N\rightarrow0$. Therefore, the relative performance between $\hbeta_{uw}$ and $\hbeta_w$ are mainly determined by the relative magnitude between $\V^\os$ and $\bSigma_{\bbeta_t}$. We have the following result comparing $\V^\os$ and $\bSigma_{\bbeta_t}$. 

\begin{proposition}\label{prop1} 
  If $\M$, $\V_{c}^\os$, and $\bSigma_{\bbeta_t}$ are finite and positive definite matrices, then
  \begin{equation}\label{eq:12}
    \bSigma_{\bbeta_t}\le \V^\os.
  \end{equation}
  Here, the inequality is in the Loewner ordering, i.e., for positive semi-definite matrices $\A$ and $\B$, $\A\ge\B$ if and only if $\A-\B$ is positive semi-definite. If $h(\x)=1$, then the equality in \eqref{eq:12} holds.
  Furthermore, note that the asymptotic variance-covariance matrix (scaled by $n$) for uniform subsampling estimator is $\M^{-1}$. If $\bbeta_t\neq\0$ and $h(\x)=\|L\M^{-1}\x\|$ for some matrix $L$, then
\begin{equation}\label{eq:60}
  \tr(L\bSigma_{\bbeta_t}L\tp) \le \tr(L\V^\os L\tp)
  \le \Exp\{\phi(\bbeta_t)\}\tr(L\M^{-1}L\tp)
  <\tr(L\M^{-1}L\tp).
\end{equation}
\end{proposition}
\begin{remark}
This proposition shows that 
$\hbeta_{uw}$ is typically more efficient than $\hbeta_w$ in estimating $\bbeta_t$. The numerical results in Section~\ref{sec:numerical-examples} also confirm this. {Assume that $n/N\rightarrow\rho$. For the un-weighted estimator, the variation of $\sqrt{N}(\hbeta_{uw}-\hbeta_{\wmle})$ is measured by $\rho^{-1}\bSigma_{\bbeta_t}$ and the variation of $\sqrt{N}(\hbeta_{\wmle}-\bbeta_t)$ is measured by $\bSigma_{\wmle}$, while for the weighted estimator the variation of $\sqrt{N}(\hbeta_w-\hbeta_{\mle})$ is measured by $\rho^{-1}\V^\os$ and the variation of $\sqrt{N}(\hbeta_{\mle}-\bbeta_t)$ is measured by $\M^{-1}$. 
  Note that $\bSigma_{\bbeta_t}$, $\bSigma_{\wmle}$, $\V^\os$, and $\M^{-1}$ are all fixed constant matrices that do not depend on $\rho$,  $\bSigma_{\bbeta_t}\le \V^\os$, and $\bSigma_{\wmle}=\bSigma_{\mle}$ if $\bSigma_{\bbeta_t}=\V^\os$. Thus, if $\rho$ is small enough, 
$\hbeta_{uw}$ is more efficient than $\hbeta_w$ in estimating $\bbeta_t$, and we do not need to require that $n/N\rightarrow0$.
}
\end{remark}
Since the equality in \eqref{eq:12} holds if $h(\x)=1$, this indicates that for subsample obtained from local case-control subsampling with replacement, the weighted and un-weighted estimators have the same conditional asymptotic distribution.

\section{Poisson subsampling}
\label{sec:poisson-sampling}
For the more efficient estimator $\hbeta_{uw}$ in Section~\ref{sec:more-effic-infer} as well as the weighted estimator $\hbeta_w$, the subsampling procedure used is sampling with replacement, which is faster to compute than sampling without replacement for a fixed sample size. In addition, the resultant subsample are independent and identically distributed (i.i.d.) conditional on the full data. However, to implement sampling with replacement, subsampling probabilities $\{\pi_i^{\os}(\hbeta_0)\}_{i=1}^{N}$ need to be used all at once, and a large amount of random numbers need to be generated all at once. This may reduce the computational efficiency, and it may require a large RAM to implement the method. Furthermore, since a data point may be included for multiple times in the subsample, the resultant estimator may not be the most efficient.

To enhance the computation and estimation efficiency of the subsample estimator, we consider Poisson subsampling, which is also fast to compute and the resultant subsample can be independent without conditioning on the full data. Note that for subsampling with replacement, a resultant subsample is generally not independent, although it is i.i.d conditional on the full data. %
As another advantage with Poisson subsampling, there is no need to calculate subsampling probabilities all at once, nor to generate a large amount of random numbers all at once. Furthermore, a data point cannot be included in the subsample for more than one time. 
A limitation of Poisson subsampling is that the subsample size is always random. %
Due to this, {we use $n^*$ to denote the actual subsample size, and abuse the notation in this section to use $n$ to denote the expected subsample size, i.e., $\Exp(n^*)=n$.}

Note that $\{\pi_i^{\os}(\bbeta)\}_{i=1}^N$ depend on the full data through the term in the denominator, $\sumN|y_i-p(\x_i,\bbeta)|h(\x_i)$. Write  
$\Psi_N(\bbeta)=N^{-1}\sumN|y_i-p(\x_i,\bbeta)|h(\x_i)$, and denote its limit as $\Psi(\bbeta)=\Exp\{|y-p(\x,\bbeta)|h(\x)\}$. Note that $\Psi(\bbeta_t)=2\Phi(\bbeta_t)$. The pilot subsample can be used to obtain an estimator of $\Psi(\bbeta_t)$ to approximate $\Psi_N(\bbeta)$.  %
Let $\hat\Psi_0$ be a pilot estimator of $\Psi(\bbeta_t)$. Here, we focus on the Poisson subsampling procedure and assume that such $\hat\Psi_0$ is available and consistent. %
We will provide more details on $\hat\Psi_0$ in Section~\ref{sec:pilot-estim-pract} and Section~\ref{sec:misspecifications}. 

With $\hbeta_0$ and $\hat\Psi_0$ available, the Poisson subsampling procedure is described as the following. 
For $i=1, ..., N$, calculate $\pi_i^p=|y_i-p(\x_i,\hbeta_0)|h(\x_i)/(N\hat\Psi_0)$, generate $u_i\sim U(0,1)$, and include $(\x_i,y_i,\pi_i^p)$ in the subsample if $u_i\le n\pi_i^p$. For the obtained subsample, say $\{(\x_1^*,y_1^*,\pi_1^{p*}), ..., (\x_{n^*}^*,y_{n^*}^*,\pi_{n^*}^{p*})\}$, calculate
  \begin{equation}\label{eq:14}
    \tilde\bbeta_p=\arg\max_{\bbeta}\ell_{p}^*(\bbeta)
    =\arg\max_{\bbeta}\sum_{i=1}^{n^*}
    (n\pi_i^{p*}\vee1)
    \big\{\bbeta\tp\x_i^*y_i^*+\log(1+e^{\bbeta\tp\x_i^*})\big\},
  \end{equation}
  and let $\hbeta_p=\tbeta_p+\hbeta_0$. Note that here the actual subsample size $n^*$ is random. 

Poisson subsampling does not require to calculate $\pi_i^p$'s all at once; each $\pi_i^p$ is calculated for each individual data point when scanning through the full data. Thus, one pass through the data finishes the sampling. For the estimation step, if $\pi_i^p$ is large so that $n\pi_i^p>1$, then this more informative data point will be given a larger weight, $n\pi_i^p$, in the objective function in \eqref{eq:14}. The following theorem describes asymptotic properties of $\hbeta_p$. 

\begin{theorem}\label{thm:2}
Under Assumptions \ref{as:1}-\ref{as:2} and assume that $\hbeta_0$ is consistent, conditional on
$\Fn$, as $n_0$, $n$, and $N$ go to infinity, if $n/N\rightarrow0$, then
  \begin{equation*}
    \sqrt{n}(\hbeta_p-\bbeta_t)
    \longrightarrow \Nor(0,\ \bSigma_{\bbeta_t}),
  \end{equation*}
  in distribution; if $n/N\rightarrow\rho\in(0,1)$, then
  \begin{equation}\label{eq:64}
    \sqrt{n}(\hbeta_p-\hbeta_{\wmle})
    \longrightarrow \Nor(0,\ \bSigma_{\bbeta_t}\bLambda_{\rho}\bSigma_{\bbeta_t}),
  \end{equation}
  in distribution, where
  \begin{align*}
    \bLambda_{\rho}
    &=\frac{\Exp\big[|\psi(\bbeta_t)|h(\x)
    \{\Psi(\bbeta_t)-\rho|\psi(\bbeta_t)|h(\x)\}_+\x\x\tp\big]}
      {4\Psi^2(\bbeta_t)}
  \end{align*}
  with {$\psi(\bbeta)=y-p(\x,\bbeta)$ and $\Psi(\bbeta)=\Exp\{|y-p(\x,\bbeta)|h(\x)\}$}, 
  and $()_+$ means the positive part of the quantity, i.e., $a_+=aI(a>0)$. %
\end{theorem}
\begin{remark}
  Similar to the case of Theorem~\ref{thm:1}, \eqref{eq:64} implies that given $\Fn$ in probability, $\hbeta_p-\hbeta_{\wmle}=O_{{P}|\Fn}(n^{-1/2})$, which implies that $\hbeta_p-\hbeta_{\wmle}=O_P(n^{-1/2})$ unconditionally because if a sequence is stochastically bounded in conditional probability, then it is also stochastically bounded in unconditional probability \citep{xiong2008some, cheng2010bootstrap}. Since $\hbeta_{\wmle}-\bbeta_t=O_P(N^{-1/2})$, we have $\hbeta_p-\bbeta_t=O_P(n^{-1/2}+N^{-1/2})=O_P(n^{-1/2})$, showing that $\hbeta_p$ is $\sqrt{n}$-consistent to $\bbeta_t$ unconditionally on the full data.
\end{remark}

Theorem~\ref{thm:2} shows that with Poisson subsampling, the asymptotic variance-covariance matrices may differ for different sampling ratios $n/N$. In addition, comparing Theorems~\ref{thm:1} and~\ref{thm:2}, we know that $\hbeta_{uw}$ and $\hbeta_p$ have the same asymptotic distribution if $n/N\rightarrow0$. This is intuitive because if the sampling ratio $n/N$ is small, sampling with replacement has close performance to sampling without replacement. However, if the sampling ratio $n/N$ does not converge to zero, then $\hbeta_{uw}$ and $\hbeta_p$ have the same asymptotic mean but different asymptotic variance-covariance matrices. The following result compares the two asymptotic variance-covariance matrices. 

\begin{proposition}\label{prop2}
  If $\rho>0$ and $\bSigma_{\bbeta_t}$ is a finite and positive definite matrix, then
  \begin{align*}
    \bSigma_{\bbeta_t}\bLambda_{\rho}\bSigma_{\bbeta_t}<\bSigma_{\bbeta_t},
  \end{align*}
  under the Loewner ordering.
\end{proposition}
This proposition shows that Poisson subsampling is more efficient than sampling with replacement.  %

\section{Pilot estimate and practical implementation}
\label{sec:pilot-estim-pract}

Since $\{\pi_i^{\os}(\bbeta)\}_{i=1}^N$ depend on the unknown $\bbeta$, a pilot estimate of $\bbeta$ is required to approximate them. The pilot estimate can be obtained by taking a pilot subsample using uniform subsampling or case-control subsampling. For uniform subsampling, all subsampling probabilities are equal, while for case-control subsampling, the subsampling probability for the cases ($y_i=1$) is different from that for the controls ($y_i=0$). Let the subsampling probabilities used to take the pilot subsample be
\begin{equation}\label{eq:34}
  \pi_{0i}=\frac{c_0(1-y_i)+c_1y_i}{N},
\end{equation}
where $c_0$ and $c_1$ are two constants that can be used to balance the numbers of 0's and 1's in the responses for the pilot subsample. If $c_0=c_1=1$, then $\pi_{0i}=N^{-1}$ corresponds to the uniform subsampling. This choice is recommended due to its simplicity if the proportion of 1's is close to 0.5 \citep{WangZhuMa2017}. If $c_0\neq c_1$, then $\pi_{0i}$'s are the case-control subsampling probabilities. This choice is recommended for imbalanced full data. Often, some prior information about the marginal probability $\Pr(y=1)$ is available. If $p_{pr}$ is the prior marginal probability, we can choose $c_0=\{2(1-p_{pr})\}^{-1}$ and $c_1=(2p_{pr})^{-1}$. 
The pilot estimate $\hbeta_0$ can be obtained using the pilot subsample. For uniform subsampling, weighted and un-weighted estimators are the same. For case-control subsampling, we use un-weighted estimators with bias correction for both sampling with replacement and Poisson subsampling.

To obtain a final estimator, \cite{WangZhuMa2017} pooled the pilot subsample with the second stage subsample taken using approximated optimal subsampling probabilities. While this does not make a difference asymptotically since $n_0$ is typically a small term compared with $n$, i.e., $n_0=o(n)$, using the pilot subsample helps to improve the finite sample performance in practical applications. However, pooling the raw samples may not be the most computationally efficient way of utilizing the pilot subsample. Since $\hbeta_0$ is already calculated, we can use it directly to improve the second stage estimator using the aggregation procedure in the divide-and-conquer method \citep{LinXie2011,schifano2016online}. This avoids iterative calculations on the pilot subsample for the second time.

For subsampling with replacement, when the full data cannot be loaded into available RAM, special considerations have to be given in practical implementation. If the full data is larger than available RAM while subsampling probabilities $\{\pi_i^{\os}(\hbeta_0)\}_{i=1}^{N}$ can still be loaded in available RAM, one can calculate $\{\pi_i^{\os}(\hbeta_0)\}_{i=1}^{N}$ by reading the data from hard drive line-by-line or block-by-block, generate row indexes for a subsample, and then scan the data line-by-line or block-by-block to take the subsample. A detailed procedure is provided in Section~\ref{sec:subs-with-repl} of the appendix.

For Poisson subsampling, the pilot subsample can also be used to construct $\hat\Psi_0$ to approximate $\Psi_N(\bbeta)$. We use the following expression to obtain $\hat\Psi_0$.
\begin{align}\label{eq:18}
  \hat\Psi_0
  &=\oneN\sum_{i=1}^{n_0^*}
    \frac{|y_i^{*0}-p(\x_i^{*0},\hbeta_0)|h(\x_i^{*0})}{(n_0\pi_{0i}^*)\wedge1},
\end{align}
 where $(\x_i^{*0},y_i^{*0})$'s are observations in the pilot subsample. 
 If $h(\x)=\|L\M_N^{-1}\x\|$ for some $L$, then the pilot subsample is used to approximate $\M_N$ through
  \begin{equation*}
    \hat\M_0=\oneN\sum_{i=1}^{n_0^*}
  \frac{\phi_i^{*_0}(\hbeta_0)\x_i^{*_0}(\x_i^{*_0})\tp}
  {(n_0\pi_{0i}^*)\wedge1},
\end{equation*}
where $\phi_i^{*_0}(\bbeta)=p(\x_i^{*_0},\bbeta)\{1-p(\x_i^{*_0},\bbeta)\}$. 
It can be verified that $\hat\Psi_0$ and $\hat\M_0$ converge in probability to $\Psi(\bbeta_t)$ and $\M$, respectively.

Taking into account all aforementioned issues in this section, including how to obtain pilot estimates, how to combine them with the second stage estimates, as well as how to process data file line-by-line, we summarize practical implementation procedures in Algorithm~\ref{alg:2} for sampling with replacement and in Algorithm~\ref{alg:3} for Poisson subsampling.

\begin{algorithm}[htp]%
  \caption{More efficient estimation based on subsampling with replacement}
  \label{alg:2}
  {\bf Step 1: obtain the pilot $\hbeta_0$}
  \begin{enumerate}[(1), labelindent=\parindent]
  \item Take pilot subsample $(\x_i^{*_0},y_i^{*_0})$, $i=1, ..., n_0$ using sampling with replacement according to subsampling probabilities $\{\pi_{0i}\}_{i=1}^N$ in \eqref{eq:34}.
  \item Calculate\\[-0.4cm]
\begin{equation*}
  \tbeta_0 =\arg\max_{\bbeta}\ell_{uw}^{*_0}(\bbeta)
  =\arg\max_{\bbeta}\sum_{i=1}^{n_0}
  \big\{\bbeta\tp\x_i^{*_0}y_i^{*_0}
  -\log\big(1+e^{\bbeta\tp\x_i^{*_0}}\big)\big\},
\end{equation*}
and let $\hbeta_0=\tbeta_0+\mathbf{b}$, where
$\mathbf{b}=\{\log({c_0}/{c_1}), 0, ..., 0\}\tp$.
\end{enumerate}

{\bf Step 2: obtain the more efficient estimator $\hbeta_{uw}$}
\begin{enumerate}[(1), labelindent=\parindent]
\item Calculate $\{\pi_i^{\os}(\hbeta_0)\}_{i=1}^{N}$ defined in equation~\eqref{eq:2}; 
  take subsample $(\x_i^*,y_i^*)$, $i=1, ..., n$ according to sampling probabilities $\{\pi_i^{\os}(\hbeta_0)\}_{i=1}^{N}$ using sampling with replacement.

\item Calculate\\[-0.4cm]
\begin{equation*}
  \tbeta_{uw}=\arg\max_{\bbeta}\ell_{uw}^*(\bbeta)
    =\arg\max_{\bbeta}\sumn
    \big\{\bbeta\tp\x_i^*y_i^*-\log\big(1+e^{\bbeta\tp\x_i}\big)\big\},
  \end{equation*}
and let $\hbeta_{uw}=\tbeta_{uw}+\hbeta_0$. 
\end{enumerate}

{\bf Step 3: combine the two estimators $\hbeta_0$ and $\hbeta_{uw}$}
\begin{enumerate}[{ }]
\item Calculate%
\begin{align*}
  \cbeta_{uw}
  =\{\ddot\ell_{uw}^{*_0}(\tbeta_0)+\ddot\ell_{uw}^*(\tbeta_{uw})\}^{-1}
  \{\ddot\ell_{uw}^{*_0}(\tbeta_0)\hbeta_0
  +\ddot\ell_{uw}^*(\tbeta_{uw})\hbeta_{uw}\},
\end{align*}
where $\ddot\ell_{uw}^{*_0}(\tbeta_0)
=\sum_{i=1}^{n_0}\phi_i^{*_0}(\tbeta_0)\x_i^{*_0}(\x_i^{*_0})\tp$, $\ddot\ell_{uw}^*(\tbeta_{uw})=\sumn\phi_i^*(\tbeta_{uw})\x_i^*(\x_i^*)\tp$, and $\phi_i^*(\bbeta)=p(\x_i^*,\bbeta)\{1-p(\x_i^*,\bbeta)\}$. .

The variance-covariance matrix of $\cbeta_{uw}$ can be estimated by %
\begin{align}
  \hat\Var(\cbeta_{uw})
  &=\{\ddot\ell_{uw}^{*_0}(\tbeta_0)+\ddot\ell_{uw}^*(\tbeta_{uw})\}^{-1}
 \bigg[\sum_{i=1}^{n_0}\{\psi_i^{*_0}(\tbeta_0)\}^2\x_i^{*_0}(\x_i^{*_0})\tp\notag\\
 &\hspace{3cm} +\sumn\{\psi_i^*(\tbeta_{uw})\}^2\x_i^*(\x_i^*)\tp\bigg]
   \{\ddot\ell_{uw}^{*_0}(\tbeta_0)+\ddot\ell_{uw}^*(\tbeta_{uw})\}^{-1},
   \label{eq:36}
\end{align}
where $\psi_i^*(\bbeta)=y_i^*-p(\x_i^*,\bbeta)$.
\end{enumerate}%
\end{algorithm}

\begin{algorithm}[htp]%
\caption{More efficient estimation based on Poisson subsampling}
  \label{alg:3}
{\bf Step 1: obtain the pilots $\hbeta_0$ and $\hat\Psi_0$}
\begin{enumerate}[(1), labelindent=\parindent]
      \item For $i=1, ..., N$, calculate $\pi_{0i}=\frac{c_0(1-y_i)+c_1y_i}{N}$, generate $u_{0i}\sim U(0,1)$, and add $(\x_i,y_i, \pi_{i1})$ in the subsample if $u_{0i}\le n_0\pi_{0i}$.
      \item For the obtained subsample, say $(\x_i^{*_0},y_i^{*_0},\pi_{0i}^{*_0})$, $i=1, ..., n_0^*$, calculate
  \begin{equation*}
    \tbeta_0=\arg\max_{\bbeta}\ell_p^{*_0}(\bbeta)
    =\arg\max_{\bbeta}\sum_{i=1}^{n_0^*}(n\pi_{0i}^{*_0}\vee1)
    \big\{\bbeta\tp\x_i^{*_0}y_i^{*_0}
    +\log\big(1+e^{\bbeta\tp\x_i^{*_0}}\big)\big\},
  \end{equation*}
  let $\hbeta_0=\tbeta_0+\b$, and then calculate $\hat\Psi_0$ in equation~\eqref{eq:18}.
\end{enumerate}

{\bf Step 2: obtain the more efficient estimator $\hbeta_p$}
\begin{enumerate}[(1), labelindent=\parindent]
\item For $i=1, ..., N$, calculate $\pi_i^p=\frac{|y_i-p(\x_i,\hbeta_0)|h(\x_i)}{N\hat\Psi_0}$,  generate $u_i\sim U(0,1)$, and if $u_i\le n\pi_i^p$ add $(\x_i,y_i,\pi_i^p)$ in
    the subsample. 
\item For the obtained subsample, say $\{(\x_1^*,y_1^*,\pi_1^{p*}), ..., (\x_{n^*}^*,y_{n^*}^*,\pi_{n^*}^{p*})\}$, calculate
  \begin{equation*}
    \tilde\bbeta_p=\arg\max_{\bbeta}\ell_p^*(\bbeta)
    =\arg\max_{\bbeta}\sum_{i=1}^{n^*}(n\pi_i^{p*}\vee1)
    \big\{\bbeta\tp\x_i^*y_i^*+\log(1+e^{\bbeta\tp\x_i^*})\big\},
  \end{equation*}
  and let $\hbeta_p=\tbeta_p+\hbeta_0$.
\end{enumerate}

{\bf Step 3: combine the two estimators $\hbeta_0$ and $\hbeta_p$}
\begin{enumerate}[{ }]
\item Calculate
\begin{align*}
  \cbeta_p
  =\{\ddot\ell_p^{*_0}(\tbeta_0)+\ddot\ell_p^*(\tbeta_p)\}^{-1}
  \{\ddot\ell_p^{*_0}(\tbeta_0)\hbeta_0
  +\ddot\ell_p^*(\tbeta_p)\hbeta_p\},
\end{align*}
where $\ddot\ell_p^{*_0}(\tbeta_0)
=\sum_{i=1}^{n_0^*}\phi_i^{*_0}(\tbeta_0)\x_i^{*_0}(\x_i^{*_0})\tp$ and $\ddot\ell_p^*(\tbeta_p)=\sum_{i=1}^{n^*}\phi_i^*(\tbeta_p)\x_i^*(\x_i^*)\tp$.

The variance-covariance matrix of $\cbeta_p$ can be estimated by %
\begin{align}
  \hat\Var(\cbeta_p)
  &=\{\ddot\ell_p^{*_0}(\tbeta_0)+\ddot\ell_p^*(\tbeta_p)\}^{-1}
\bigg[\sum_{i=1}^{n_0^*}\{\psi_i^{*_0}(\tbeta_0)\}^2\x_i^{*_0}(\x_i^{*_0})\tp\notag\\
  &\hspace{4cm}+\sum_{i=1}^{n^*}\{\psi_i^*(\tbeta_p)\}^2\x_i^*(\x_i^*)\tp\bigg]
    \{\ddot\ell_p^{*_0}(\tbeta_0)+\ddot\ell_p^*(\tbeta_p)\}^{-1}.
    \label{eq:37}
\end{align}
\end{enumerate}%
\end{algorithm}

\begin{remark}
  In Algorithm~\ref{alg:2} and Algorithm~\ref{alg:3}, if $n_0=o(n)$, then the results for $\hbeta_{uw}$ in Theorem~\ref{thm:1} hold for $\cbeta_{uw}$ and the results for $\hbeta_p$ in Theorem~\ref{thm:2} hold for $\cbeta_p$ as well. This is because $\{\ddot\ell_{uw}^{*_0}(\tbeta_0)+\ddot\ell_{uw}^*(\tbeta_{uw})\}^{-1}
  \ddot\ell_{uw}^{*_0}(\tbeta_0)\sqrt{n}(\hbeta_0-\bbeta_t)
  =O_p(\sqrt{n_0}/\sqrt{n})=\op$ and $\{\ddot\ell_{uw}^{*_0}(\tbeta_0)+\ddot\ell_{uw}^*(\tbeta_{uw})\}^{-1}
  \ddot\ell_{uw}(\tbeta_{uw})\rightarrow1$ in probability. The reason for $\cbeta_p$ is similar. 
\end{remark}

\begin{remark}
  In Algorithm~\ref{alg:2} and Algorithm~\ref{alg:3}, 
  to combine the two stage estimates using the second derivative of the objective functions, the inconsistent estimators $\tbeta_0$, and $\tbeta_{uw}$ or $\tbeta_p$ should be used, because their limits correspond to the terms in the asymptotic variance-covariance matrices of the more efficient estimators. This is an advantage of the proposed estimators for implementation using existing software that fit logistic regression. One can use the inverse of the estimated variance-covariance matrix from the software output to replace the second derivative of the objective function. 
\end{remark}

\begin{remark}
  The variance-covariance estimators $\hat\Var(\cbeta_{uw})$ in \eqref{eq:36} and $\hat\Var(\cbeta_p)$ in \eqref{eq:37} can be replaced by the following simplified estimators,
\begin{align*}
  \hat\Var_s(\cbeta_{uw})=
  \{\ddot\ell_{uw}^{*_0}(\tbeta_0)+\ddot\ell_{uw}^*(\tbeta_{uw})\}^{-1}
  \qquad\text{ and }\qquad
  \hat\Var_s(\cbeta_p)=
  \{\ddot\ell_p^{*_0}(\tbeta_0)+\ddot\ell_p^*(\tbeta_p)\}^{-1},
\end{align*}
respectively. If the subsampling ratio $n/N$ is much smaller than one, then $\hat\Var_s(\cbeta_{uw})$ and $\hat\Var_s(\cbeta_p)$ perform very similarly to $\hat\Var(\cbeta_{uw})$ and $\hat\Var(\cbeta_p)$, respectively.
\end{remark}

\begin{remark}
  The time complexity of Algorithm~\ref{alg:2} is the same as that of Algorithm 2 in \cite{WangZhuMa2017}. The major computing time is to calculate $\{\pi_i^{\os}(\hbeta_0)\}_{i=1}^{N}$ in Step 2, but it does not require iterative calculations on the full data. Once $\{\pi_i^{\os}(\hbeta_0)\}_{i=1}^{N}$ are available, %
  the calculations of $\hbeta_{uw}$ and $\cbeta_{uw}$ are fast because they are done on the subsamples only. 
  To calculate $\{\pi_i^{\os}(\hbeta_0)\}_{i=1}^{N}$, the required time varies. For $\pi_i^{\mvc}$, the required time is $O(Nd)$; for $\pi_i^{\mmse}$, the required time is $O(Nd^2)$. Thus, the time complexity of Algorithm~\ref{alg:2} with $\pi_i^{\mvc}$ is $O(Nd)$ and the time complexity with $\pi_i^{\mmse}$ is $O(Nd^2)$, if the sampling ratio $n/N$ is much smaller than one.
\end{remark}

\section{Unconditional distribution}
\label{sec:uncond-distr}
Asymptotic distributional results in Sections~\ref{sec:more-effic-infer} and \ref{sec:poisson-sampling}, as well as in \cite{WangZhuMa2017}, are about conditional distributions, i.e., they are about conditional distributions of subsample-based estimators given the full data. We investigate the unconditional distribution of $\hbeta_p$ in this section. 

\begin{theorem}\label{thm:3}
  Under Assumptions \ref{as:1} and \ref{as:2}, if the pilot estimators are obtained {from a uniform subsample of sample size $n_0=o(\sqrt{N})$} and $\Exp\{h^3(\x)\|\x\|^3\}$, $\Exp\{h^3(\x)\|\x\|^2\}$, $\Exp\{h(\x)\|\x\|^3\}$, and $\Exp\{h^2(\x)\}$ are finite, or if $\hbeta_0$ and $\hat\Psi_0$ are independent of the data $\Fn$, then %
  as $n_0$, $n$, and $N$ go to infinity such that $n/N\rightarrow\rho\in[0,1)$, we have
  \begin{equation}\label{eq:62}
    \sqrt{n}(\hbeta_p-\bbeta_t)
    \longrightarrow \Nor(0,\ \bSigma_{\bbeta_t}\bLambda_{u}\bSigma_{\bbeta_t}),
  \end{equation}
  in distribution, where
\begin{align*}
  \bLambda_{u}
  =\frac{\Exp[|\psi(\bbeta_t)|\{
    \rho|\psi(\bbeta_t)|h(\x)\vee\Psi(\bbeta_t)\}
  h(\x)\x\x\tp]}{4\Psi^2(\bbeta_t)}.
\end{align*}
\end{theorem}

\begin{remark}
  If the pilot estimators $\hbeta_0$ and $\hat\Psi_0$ are obtained through the full data $\Fn$, stronger moment conditions are required. Note that $h(\x)$ is often a function of the norm of $\x$, such as in $\pi_i^{\mvc}$, $\pi_i^{\mmse}$, and the local case-control subsampling. In general, if $h(\x)=\|\A\x\|^a$ for some matrix $\A$ and constant $a\ge0$, then the four additional moment conditions reduce to one requirement of $\Exp\{h^3(\x)\|\x\|^3\}<\infty$.
\end{remark}

\begin{remark}
  If $\rho|\psi(\bbeta_t)|h(\x)\le\Psi(\bbeta_t)$ almost surely, then $\bLambda_{u}$ reduced to $\bSigma_{\bbeta_t}^{-1}$ and as a result $\bSigma_{\bbeta_t}\bLambda_{u}\bSigma_{\bbeta_t}$ reduces to $\bSigma_{\bbeta_t}$. Furthermore, if the subsampling probabilities are propositional to the local case-control subsampling probabilities, i.e., $h(\x)=1$, then $\bSigma_{\bbeta_t}\bLambda_{u}\bSigma_{\bbeta_t}$ reduces to $4\Exp\{\phi(\bbeta)\}\M^{-1}$. For the uniform Poisson subsampling estimator, the unconditional asymptotic variance-covariance matrix (scaled by $n$) is $\M^{-1}$. 
  From \eqref{eq:60}, with the same expected subsample size, the proposed method has a higher estimation efficiency than subsampling proportional to the local case-control subsampling probabilities, which is more efficient than the uniform Poisson subsampling approach.
\end{remark}

\begin{remark}
  \cite{fithian2014local}'s investigation corresponds to the case of $h(\x)=1$ and $\rho=2\Exp\{\phi(\bbeta_t)\}$. For this scenario in Theorem~\ref{thm:3}, the asymptotic variance-covariance matrix of $\hbeta_p$ reduces to $2N^{-1}\M^{-1}$, which is the same as obtained in \cite{fithian2014local}. This result is particularly neat in the fact that this asymptotic variance-covariance matrix is proportional to that from the full data MLE with a multiplier of 2. The result in Theorem~\ref{thm:3} is more general. It shows that if $h(\x)=1$, then as long as $\rho|\psi(\bbeta_t)|\le2\Exp\{\phi(\bbeta_t)\}$ (which is satisfied if $\rho=2\Exp\{\phi(\bbeta_t)\}$), the asymptotic variance-covariance matrix of $\hbeta_p$ can be written as
\begin{equation*}
  \frac{4\Exp\{\phi(\bbeta_t)\}}{\rho N}\M^{-1},
\end{equation*}
which is proportional to that of the full data MLE with a multiplier of $4\rho^{-1}\Exp\{\phi(\bbeta_t)\}$. We need to emphasize that this simple representation holds only when $\rho|\psi(\bbeta_t)|\le2\Exp\{\phi(\bbeta_t)\}$ almost surely. If the subsampling ratio $\rho$ gets closer to one, the asymptotic variance-covariance matrix in \eqref{eq:62} may not be simplified. 
\end{remark}

{From Theorems~\ref{thm:2} and~\ref{thm:3}, the conditional asymptotic distribution and unconditional asymptotic distribution of $\hbeta_p$ are the same if $n/N\rightarrow0$.} This is intuitive, because if the sampling ratio $n/N$ is small, the variation of $\hbeta_p$ due to the variation of the full data is small compared with the variation due to the variation of the subsampling. However, if the sampling ratio $n/N$ does not converge to zero, then the conditional asymptotic distribution and unconditional asymptotic distribution of $\hbeta_p$ are quite different. First, we notice that under the unconditional distribution, $\hbeta_p$ is asymptotically unbiased to $\bbeta_t$, while under the conditional distribution, $\hbeta_p$ is asymptotically biased with the bias being $\hbeta_{\wmle}-\bbeta_t=O_P(N^{-1/2})$. Second, since the variation of $\hbeta_p$ due to the variation of the full data is not negligible, we  expect that the asymptotic variance-covariance matrix for the unconditional distribution to be larger than that for the conditional distribution. Indeed this is true, and we present it in the following proposition. 

\begin{proposition}\label{prop3}
  If $\rho>0$ and $\bSigma_{\bbeta_t}$ is a finite and positive definite matrix, then
  \begin{align}\label{eq:39}
    \bSigma_{\bbeta_t}\bLambda_{u}\bSigma_{\bbeta_t}
    \ge\bSigma_{\bbeta_t}
    >\bSigma_{\bbeta_t}\bLambda_{\rho}\bSigma_{\bbeta_t},
  \end{align}
  under the Loewner ordering. Furthermore, if $\Pr\{\rho|\psi(\bbeta_t)|h(\x)>\Psi(\bbeta_t)\}>0$, then the ``$\ge$'' sign in~\eqref{eq:39} can be replaced by ``$>$'', the strict great sign. 
\end{proposition}

\cite{fithian2014local} obtained unconditional distribution of local case-control estimator by assuming that the pilot estimate is independent of the data. Our Theorem~\ref{thm:3} includes this scenario, and the required assumptions are the same as those required in \cite{fithian2014local}. In practice, a consistent pilot estimator that is independent of the data may not be available and a pilot subsample from the full data is required to construct it. For this scenario, a pilot estimator is dependent on the data, and we need a stronger moment condition to establish the asymptotic normality. For local case-control subsampling, $h(\x)=1$, and the additional moment requirement is that $\Exp(\|\x\|^3)<\infty$.

\section{Misspecifications}
\label{sec:misspecifications}
In this section, we discuss the effect when the pilot estimates are misspecified or when the model is misspecified. Pilot estimates misspecification often occurs when they are from other data sources or when they are calculated based on convenient subsamples, e.g., using the first $n_0$ observations in the full data to calculate them. In these cases, it is reasonable to assume that the pilot estimates are independent of $\Fn$ and we use this assumption in this section.

\subsection{Pilot estimates misspecification}
\label{sec:pilot-estim-missp}
Here, we assume that the model is correctly specified but the pilot estimates $\hbeta_0$ and $\hat\Psi_0$ converge to limits that are different from the true parameters for the current data. Interestingly, in this case, the proposed estimators are still consistent and no specific convergence rate is required for $\hbeta_0$ or $\hat\Psi_0$.

The following theorem describes the asymptotic distribution of $\hat\bbeta_{uw}$, the estimator based on subsampling with replacement. Note that $\hat\Psi_0$ is not required by $\hat\bbeta_{uw}$. 

\begin{theorem}\label{thm:R3}
  When the logistic regression model in \eqref{eq:1} is correctly specified and the pilot estimator $\hbeta_0$ that is independent of $\Fn$ is inconsistent, i.e., $\hbeta_0\rightarrow\bbeta_0$ in probability for some $\bbeta_0$ that is different from $\bbeta_t$, then under Assumptions \ref{as:1}-\ref{as:3},  conditional on $\Fn$, as $n$, and $N$ go to infinity,
\begin{align*}
  \sqrt{n}(\hat\bbeta_{uw}-\hbeta_{\wmle}) \longrightarrow
  \Nor\big\{\0,\ \Psi(\bbeta_0)\bvsigma_{a}^{-1}\big\},
\end{align*}
  in distribution; furthermore, if $n/N\rightarrow0$, then
\begin{align*}
  \sqrt{n}(\hat\bbeta_{uw}-\bbeta_t) \longrightarrow
  \Nor\big\{\0,\ \Psi(\bbeta_0)\bvsigma_{a}^{-1}\big\},
\end{align*}
  in distribution, where $\Psi(\bbeta_0)=\Exp\{|\psi(\bbeta_0)|h(\x)\}$ and
  \begin{align*}
    \bvsigma_{a}
    &=\Exp[\{1-p(\x,\bbeta_t)\}p(\x,\bbeta_0)
     p(\x,\bbeta_t-\bbeta_0)h(\x)\x\x\tp].
  \end{align*}
  Here $\hbeta_{\wmle}$ satisfies that, %
\begin{equation*}
  \sqrt{N}(\hbeta_{\wmle}-\bbeta_t)
  \longrightarrow \Nor\big(\0,\ 
  \bvsigma_{a}^{-1}\bvsigma_{b}\bvsigma_{a}^{-1}\big),
  \end{equation*}
in distribution, where
  \begin{align*}
    \bvsigma_{b}
    =\Exp\big\{\phi(\bbeta_0)\phi(\bbeta_t-\bbeta_0)
    h^2(\x)\x\x\tp\big\}.
  \end{align*}
\end{theorem}

\begin{remark}
  If $\bbeta_0=\bbeta_t$, then direct calculations show that  $\Psi(\bbeta_0)\bvsigma_{a}^{-1}=\bSigma_{\bbeta_t}$ and $\bvsigma_{a}^{-1}\bvsigma_{b}\bvsigma_{a}^{-1}=\bSigma_{\wmle}$, that is, the results in Theorem~\ref{thm:R3} reduce to the same results in Theorem~\ref{thm:1}.   
\end{remark}
\begin{remark}\label{prop:R1}
If the pilot estimator $\hbeta_0$ is very wrong such that $\bbeta_t\tp\x\x\tp\bbeta_0<0$, i.e., $p(\x,\bbeta_t)>0.5>p(\x,\bbeta_0)$ or $p(\x,\bbeta_t)<0.5<p(\x,\bbeta_0)$, then it can be shown that
\begin{align*}
  \Psi(\bbeta_0)\bvsigma_{a}^{-1}>\bSigma_{\bbeta_t}.
\end{align*}
Detailed proof for this result is presented in Section~\ref{sec:proofs-pilot-missp} of the appendix. 
\end{remark}

The following theorem describes the asymptotic distribution of $\hat\bbeta_p$, the estimator based on Poisson subsampling. Note that $\hat\bbeta_p$ requires both $\hbeta_0$ and $\hat\Psi_0$.

\begin{theorem}\label{thm:R4}
  Assume that the logistic regression model is correctly specified, and the pilot estimators $\hbeta_0$ and $\hat\Psi_0$ are independent of $\Fn$ and they are inconsistent, i.e., $\hbeta_0\rightarrow\bbeta_0$ and $\hat\Psi_0\rightarrow\Psi_0$ in probability for some $\bbeta_0$ and $\Psi_0$, respectively. Under Assumptions \ref{as:1}-\ref{as:2}, conditional on $\Fn$, as $n$ and $N$ go to infinity, if $n/N\rightarrow0$, then
\begin{align*}
  \sqrt{n}(\hat\bbeta_p-\bbeta_t)
  \longrightarrow\Nor\big(\0,\ \Psi_0\bvsigma_{a}^{-1}\big),
\end{align*}
in distribution; if $n/N\rightarrow\rho$, then
\begin{align*}
  \sqrt{n}(\hat\bbeta_p-\hbeta_{\wmle})
  \longrightarrow\Nor\big(\0,\
  \Psi_0\bvsigma_{a}^{-1}\bvsigma_{c}\bvsigma_{a}^{-1}\big),
\end{align*}
  in distribution,
 where %
\begin{align*}
  \bvsigma_{c}
  &=\Exp\Big[|\psi(\bbeta_0)|
    \big\{1-\rho\Psi_0^{-1}|\psi(\bbeta_0)|h(\x)\big\}_+
    \psi^2(\bbeta_t-\bbeta_0)h(\x)\x\x\tp\Big].
\end{align*}
\end{theorem}

\begin{remark}
  If $\bbeta_0=\bbeta_t$ and $\Psi_0=\Psi(\bbeta_t)$, then direct calculations show that $\Psi_0^{-1}\bvsigma_{c}=\bLambda_{\rho}$, and thus the results in Theorem~\ref{thm:R4} reduce to the same results in Theorem~\ref{thm:2}.
\end{remark}

\begin{remark}
We have a result similar to that in Proposition~\ref{prop2}. By direct calculation, we know that
\begin{align*}
\bvsigma_{c}
  &<\Exp\Big[|\psi(\bbeta_0)|
    \psi^2(\bbeta_t-\bbeta_0)h(\x)\x\x\tp\Big]
  =\bvsigma_{a},
\end{align*}
under the Loewner ordering, 
which indicates that 
\begin{align*}
  \Psi_0\bvsigma_{a}^{-1}\bvsigma_{c}\bvsigma_{a}^{-1}
  <\Psi_0\bvsigma_{a}^{-1}.
\end{align*}
Thus, when the pilot estimators are misspecified, Poisson subsampling still has a higher estimation efficiency compared with subsampling with replacement if $\Psi_0=%
\Exp\{|\psi(\bbeta_0)|h(\x)\}$, which is the case if $\hat\Psi_0$ is constructed from a pilot subsample. 
\end{remark}

\subsection{Model misspecification}
\label{sec:model-missp}
In this section, we consider the case when the logistic regression model is misspecified, namely, the model in \eqref{eq:1} is not correct. Instead, we assume that the true probability of $y=1$ given $\x$ is
\begin{equation*}
  \Pr(y=1|\x)=p_t(\x),%
\end{equation*}
for some unknown function $p_t(\x)$. %
When the logistic regression model is misspecified, we need to define the meaning of consistency because there is no true $\bbeta$ any more. In this case, consistency often means that the estimator converges to a limit that minimizes expected loss with respect to a specified loss function. Here, if we denote the limit as $\bbeta_l$ and define it to be the minimizer of 
\begin{equation*}
  \Exp\Big\{-p_t(\x)h(\x)\x\tp\bbeta
  +h(\x)\log\big(1+e^{\bbeta\tp\x}\big)\Big\},
\end{equation*}
then $\bbeta_l$ satisfies 
\begin{align*}
  \Exp\big[\big\{p_t(\x)-p(\x,\bbeta_l)\big\}h(\x)\x\big]=\0,
\end{align*}
where $p(\x,\bbeta)={e^{\x\tp\bbeta}}{(1+e^{\x\tp\bbeta})^{-1}}$.

Now we investigate the asymptotic properties of the proposed estimators under model misspecification. In this case, we need to assume that the pilot estimators are consistent which is also required in the local case-control subsampling method. In addition, to investigate the asymptotic normality, we also need an additional assumption on the convergence rate of the pilot estimator $\hbeta_0$.

The following theorem describes the asymptotic behavior of the estimator $\hat\bbeta_{uw}$ based on subsampling with replacement. 

\begin{theorem}\label{thm:R1}
  Assume that the pilot sample is independent of $\Fn$ and the pilot estimator $\hbeta_0$ satisfies that
  $\sqrt{n_0}(\hbeta_0-\bbeta_l)\rightarrow N(\0,\bSigma_0)$ in distribution. 
  Under Assumptions \ref{as:1}-\ref{as:3}, if $n_0/N\rightarrow \rho_0$ and $n/N\rightarrow \rho$ with $\rho_0, \rho\in(0,1)$, then conditional on
  $\Fn$, as $n_0$, $n$, and $N$ go to infinity,
\begin{align}\label{eq:49}
  \sqrt{n}(\hat\bbeta_{uw}-\hbeta_{\wmle})
  &\longrightarrow \Nor\big(\0,\ \omega\bkappa_{a}^{-1}\big)
\end{align}
in distribution, where
\begin{align*}
  \bkappa_{a}
  &%
  =\frac{1}{4}\Exp\big[\{p_t(\x)-2p_t(\x)p(\x,\bbeta_l)
    +p(\x,\bbeta_l)\}h(\x)\x\x\tp\big],\\
  \omega
  &%
  =\Exp\big[\{p_t(\x)-2p_t(\x)p(\x,\bbeta_l)
    +p(\x,\bbeta_l)\}h(\x)\big],
\end{align*}
and $\hbeta_{\wmle}$ satisfies that
\begin{align}\label{eq:47}
  \sqrt{N}(\hbeta_{\wmle}-\bbeta_l)
  \longrightarrow \Nor\Big\{\0,\
  \bkappa_{a}^{-1}\big(\bkappa_{b}+\rho_0^{-1}
    \bkappa_{c}\bSigma_0\bkappa_{c}\big)\bkappa_{a}^{-1}\Big\},
\end{align}
in distribution, with
\begin{align*}
  \bkappa_{b}
  &%
  =\frac{1}{4}\Exp\big[\{p_t(\x)-2p_t(\x)p(\x,\bbeta_l)
    +p^2(\x,\bbeta_l)\}h^2(\x)\x\x\tp\big], \text{ and }\\
  \bkappa_{c}
  &%
=\frac{1}{4}\Exp\big[\{1-2p(\x,\bbeta_l)\}
    \{p_t(\x)-p(\x,\bbeta_l)\}h(\x)\x\x\tp\big].
\end{align*}
\end{theorem}
\begin{remark}
  If the model is correctly specified, i.e., $p_t(\x)=p(\x,\bbeta_t)$, then $\omega\bkappa_{a}^{-1}=\bSigma_{\bbeta_t}$,  $\bkappa_{b}=\frac{1}{4}\Exp\big\{\phi(\bbeta_t)
  h^2(\x)\x\x\tp\big\}$ and $\bkappa_{c}=\0$, and therefore the results in Theorem \ref{thm:R1} reduce to the same expressions as those in Theorem \ref{thm:1}. 
\end{remark}

From Theorem~\ref{thm:R1}, with model misspecification, it is critical to have a good pilot estimator $\hbeta_0$. Note that the pilot sample size is typically much smaller than the full data sample size, so $\rho_0$ can be close to zero. From \eqref{eq:47}, we see that the asymptotic variance-covariance matrix of $\hbeta_{\wmle}$ can be inflated by a small pilot sample size.

The following theorem presents asymptotic results for the estimator based on Poisson subsampling. 

\begin{theorem}\label{thm:R2}
  Assume that the pilot sample is independent of $\Fn$ and the pilot estimators satisfy that
  $\sqrt{n_0}(\hbeta_0-\bbeta_l)\rightarrow N(\0,\bSigma_0)$ in distribution and $\hat\Psi_0\rightarrow\omega$ in probability. 
  Under Assumptions \ref{as:1}-\ref{as:2}, if $n_0/N\rightarrow \rho_0$ and $n/N\rightarrow \rho$ with $\rho_0, \rho\in(0,1)$, then conditional on
  $\Fn$, as $n_0$, $n$, and $N$ go to infinity,
\begin{align*}
  \sqrt{n}(\hat\bbeta_p-\hbeta_{\wmle})
  &\longrightarrow \Nor\big(\0,\ \omega\bkappa_{a}^{-1}\bkappa_{d}\bkappa_{a}^{-1}\big).
\end{align*}
  in distribution, where $\omega$, $\bkappa_{a}$ and $\bkappa_{b}$ are defined in Theorem~\ref{thm:R1}, and 
\begin{align*}
  \bkappa_{d}
  &=\frac{1}{4}\Exp\Big[|\psi(\bbeta_l)|
    \{1-\rho\omega^{-1}|\psi(\bbeta_l)|h(\x)\}_+
    h(\x)\x\x\tp\Big].%
\end{align*}
\end{theorem}

\begin{remark}
  Similarly to Proposition~\ref{prop2}, we have that
  \begin{align*}
    \bkappa_{a}^{-1}\bkappa_{d}\bkappa_{a}^{-1}<\bkappa_{a}^{-1}
  \end{align*}
  under the Loewner ordering, indicating that the estimator based on Poisson subsampling has a smaller conditional variance-covariance matrix.
\end{remark}

For Poisson subsampling, compared with $\hbeta_0$, we require a much weaker assumption on $\hat\Psi_0$; we only need it to converge without specifying certain convergence rate. The reason is that the effect of $\hat\Psi_0$ on all the subsampling probabilities are the same and it mainly controls the expected subsample size, while $\hbeta_0$ affects individual subsampling probabilities differently corresponding to different values of $\x_i$ and $y_i$. 

\section{Numerical evaluations}
\label{sec:numerical-examples}
We evaluate the performance of the more efficient estimators in terms of both estimation efficiency and computational efficiency in this section.
 
\subsection{Estimation efficiency}
\label{sec:estim-effic}
In this section, we use numerical experiments based on simulated and real data sets to evaluate the estimators proposed in this paper. For simulation, to compare with the original OSMAC estimator, we use exactly the same setup used in Section 5.1 of \cite{WangZhuMa2017}. Specifically, the full data sample size $N=10,000$ and the true value of $\bbeta$, $\bbeta_t$, is a $7\times1$ vector of 0.5. The following 6 distributions of $\x$ are considered: multivariate normal distribution with mean zero (mzNormal), multivariate normal distribution
  with nonzero mean (nzNormal), multivariate normal distribution
  with mean zero and unequal variances (ueNormal), mixture of two multivariate normal distributions with different means (mixNormal), multivariate $t$ distribution
  with degrees of freedom 3 ($T_3$), and exponential distribution (EXP). Detailed explanations of these
  distributions can be found in Section 5.1 of \cite{WangZhuMa2017}.

To evaluate the estimation performance of the new estimators compared with the original weighted OSMAC estimator, we define the estimation efficiency of $\cbeta_{\text{new}}$ relative to $\cbeta_w$ as
  \begin{equation*}
    \text{Relative Efficiency}=
    \frac{\text{MSE}(\cbeta_w)}{\text{MSE}(\cbeta_{\text{new}})},
  \end{equation*}
  where $\cbeta_{\text{new}}=\cbeta_{uw}$ for the subsampling with
  replacement estimator described in Algorithm~\ref{alg:2} and
  $\cbeta_{\text{new}}=\cbeta_p$ for Poisson subsampling estimator
  described in Algorithm~\ref{alg:3}. We calculate empirical MSEs from
  $S=1000$ subsamples using
  \begin{equation}\label{eq:63}
    \text{MSE}(\cbeta)=\frac{1}{S}\sum_{s=1}^S
    \|\cbeta^{(s)}-\bbeta_t\|^2,
  \end{equation}
  where $\cbeta^{(s)}$ is the estimate from the $s$-th subsample. We
  fixed the first step sample size $n_0=200$ and choose $n$ to be 100,
  200, 400, 600, 800, and 1000. This is the same setup used in
  \cite{WangZhuMa2017}.

  Figure~\ref{fig:1} presents the relative efficiency of $\cbeta_{uw}$ and $\cbeta_p$ based on two different choices of $\pi_i^{\os}$: $\pi_i^{\mmse}$ and $\pi_i^{\mvc}$. It is seen that in general $\cbeta_{uw}$ and $\cbeta_p$ are more efficient than $\cbeta_w$. Among the six cases, the only case that $\cbeta_w$ can be more efficient is when $\x$ has a $T_3$ distribution and $\pi^{\mvc}$ is used, but the difference is not very significant. %
  For all other cases, $\cbeta_{uw}$ and $\cbeta_p$ are more efficient. For example, when $\x$ has the nzNormal distribution, $\cbeta_p$ can be 250\% as efficient as $\cbeta_w$ if $\pi^{\mmse}$ is used. Between $\cbeta_{uw}$ and $\cbeta_p$, $\cbeta_p$ is more efficient than $\cbeta_{uw}$ for all cases. We also calculate the empirical unconditional MSE by generating the full data in each repetition of the simulation. The results are similar and thus are omitted.

To evaluate the performance of the proposed method with different choices of the subsampling probabilities for subsampling with replacement and Poisson subsampling, Figure~\ref{fig:4} plots empirical MSEs of using $\pi^{\mmse}$, $\pi^{\mvc}$, $\pi^{\mathrm{lcc}}$ (local case-control), and the uniform subsampling probability. In general, $\pi^{\mmse}$ with Poisson subsampling has the smallest empirical MSEs while uniform subsampling with replacement has the worst estimation efficiency. This agrees with our theoretical results: 1) $\pi^{\mmse}$ minimizes the asymptotic MSE of the parameter estimator which corresponds to the empirical MSE defined in~\eqref{eq:63} for the experiments, while $\pi^{\mvc}$ minimizes the asymptotic MSE of a transformed parameter estimator, and 2) Poisson subsampling has a higher estimation efficiency compared with subsampling with replacement.  

  To assess the performance of $\hat\Var(\cbeta_{uw})$ in \eqref{eq:36} and $\hat\Var(\cbeta_p)$ in \eqref{eq:37}, we use  $\tr\{\hat\Var(\cbeta_{uw})\}$ and $\tr\{\hat\Var(\cbeta_p)\}$ to estimate the MSEs of $\cbeta_{uw}$ and $\cbeta_p$, and compare the average estimated MSEs with the unconditional empirical MSEs. We focus on the unconditional MSE because conditional inference may not be appropriate if $n/N$ is not small enough. 
  Figure~\ref{fig:2} presents the results for using 
  $\pi^{\mvc}$. Results for using $\pi^{\mmse}$ are similar and thus are omitted. Note that our purpose here is to evaluate the quality of $\hat\Var(\cbeta_{uw})$ in \eqref{eq:36} and $\hat\Var(\cbeta_p)$ in \eqref{eq:37}, so in this figure we plot the original empirical and estimated MSEs without scaling then using the MSEs of $\cbeta_w$. Here, a closer value between the estimated MSE and the empirical MSE indicates a better performance of $\hat\Var(\cbeta_{uw})$ or $\hat\Var(\cbeta_p)$. From Figure~\ref{fig:2}, the estimated MSEs are
  very close to the empirical MSEs, except for the case of nzNormal
  covariate for subsampling with replacement. In this case, the responses are imbalanced with about 95\% being 1's. For this scenario, the variance-covariance estimator for $\cbeta_w$ proposed in \cite{WangZhuMa2017} also has a similar problem of underestimation. For Poisson subsampling, the problem of underestimation from $\hat\Var(\cbeta_p)$ is not significant. 

We also apply the more efficient estimation methods to a supersymmetric (SUSY) benchmark data set \citep{Baldi2014} available from the Machine Learning Repository \citep{Dua:2017}. The data set contains a binary response variable indicating whether a process produces new supersymmetric particles or not and 18 covariates that are kinematic features about the process. The full sample size is $N=5,000,000$ and the data file is about 2.4 gigabytes. About 54.24\% of the responses in the full data are from the background process. We use the more efficient estimation methods with subsample size $n$ to estimate parameters in logistic regression.

Figures~\ref{fig:3} gives the relative efficiency of $\cbeta_{uw}$ and $\cbeta_p$ to $\cbeta_w$ for both $\pi_i^{\mvc}$ and $\pi_i^{\mmse}$. It is seen that when $\pi_i^{\mmse}$ are used, $\cbeta_{uw}$ and $\cbeta_p$ always outperform $\cbeta_w$. When $\pi_i^{\mvc}$ are used, $\cbeta_{uw}$ and $\cbeta_p$ may not be as efficient as $\cbeta_w$, but they become more efficient when the second stage sample size $n$ gets larger. It is also seen that $\cbeta_p$ dominates $\cbeta_{uw}$ and $\pi^{\mmse}$ dominates $\pi^{\mvc}$ in estimation efficiency. 

\newpage
\begin{figure}[H]
  \centering
  \begin{subfigure}{0.49\textwidth}
    \includegraphics[width=\textwidth,page=1]{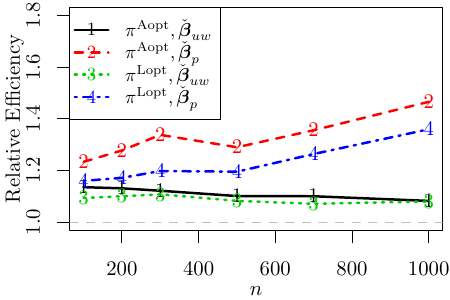}
    \caption{mzNormal}
  \end{subfigure}
  \begin{subfigure}{0.49\textwidth}
    \includegraphics[width=\textwidth,page=2]{relativeMSE.pdf}
    \caption{nzNormal}
  \end{subfigure}\\[5mm]
  \begin{subfigure}{0.49\textwidth}
    \includegraphics[width=\textwidth,page=3]{relativeMSE.pdf}
    \caption{ueNormal}
  \end{subfigure}
  \begin{subfigure}{0.49\textwidth}
    \includegraphics[width=\textwidth,page=4]{relativeMSE.pdf}
    \caption{mixNormal}
  \end{subfigure}\\[5mm]
  \begin{subfigure}{0.49\textwidth}
    \includegraphics[width=\textwidth,page=5]{relativeMSE.pdf}
    \caption{$T_3$}
  \end{subfigure}
  \begin{subfigure}{0.49\textwidth}
    \includegraphics[width=\textwidth,page=6]{relativeMSE.pdf}
    \caption{EXP}
  \end{subfigure}%
  \caption{Relative efficiency for different second step subsample size $n$ with the first step subsample size being fixed at $n_0=200$. A relative efficiency larger than one means the associate method is more efficient than the original OSMAC estimator. %
  }
  \label{fig:1}
\end{figure}

\begin{figure}[H]
  \centering
  \begin{subfigure}{0.49\textwidth}
    \includegraphics[width=\textwidth,page=1]{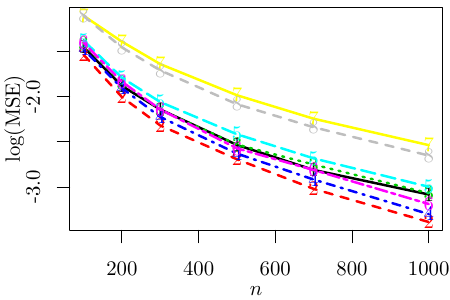}
    \caption{mzNormal}
  \end{subfigure}
  \begin{subfigure}{0.49\textwidth}
    \includegraphics[width=\textwidth,page=2]{MSE.pdf}
    \caption{nzNormal}
  \end{subfigure}\\[5mm]
  \begin{subfigure}{0.49\textwidth}
    \includegraphics[width=\textwidth,page=3]{MSE.pdf}
    \caption{ueNormal}
  \end{subfigure}
  \begin{subfigure}{0.49\textwidth}
    \includegraphics[width=\textwidth,page=4]{MSE.pdf}
    \caption{mixNormal}
  \end{subfigure}\\[5mm]
  \begin{subfigure}{0.49\textwidth}
    \includegraphics[width=\textwidth,page=5]{MSE.pdf}
    \caption{$T_3$}
  \end{subfigure}
  \begin{subfigure}{0.49\textwidth}
    \includegraphics[width=\textwidth,page=6]{MSE.pdf}
    \caption{EXP}
  \end{subfigure}%
  \caption{MSE for different subsampling probabilities with second step subsample size $n$ and a fixed first step subsample size $n_0=200$. Logarithm is taken on MSEs for better presentation.
  }
  \label{fig:4}
\end{figure}

\begin{figure}[H]
  \centering
  \begin{subfigure}{0.49\textwidth}
    \includegraphics[width=\textwidth,page=1]{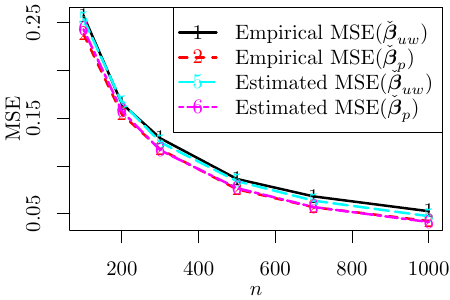}
    \caption{mzNormal}
  \end{subfigure}
  \begin{subfigure}{0.49\textwidth}
    \includegraphics[width=\textwidth,page=2]{amseunc.pdf}
    \caption{nzNormal}
  \end{subfigure}\\[5mm]
  \begin{subfigure}{0.49\textwidth}
    \includegraphics[width=\textwidth,page=3]{amseunc.pdf}
    \caption{ueNormal}
  \end{subfigure}
  \begin{subfigure}{0.49\textwidth}
    \includegraphics[width=\textwidth,page=4]{amseunc.pdf}
    \caption{mixNormal}
  \end{subfigure}\\[5mm]
  \begin{subfigure}{0.49\textwidth}
    \includegraphics[width=\textwidth,page=5]{amseunc.pdf}
    \caption{$T_3$}
  \end{subfigure}
  \begin{subfigure}{0.49\textwidth}
    \includegraphics[width=\textwidth,page=6]{amseunc.pdf}
    \caption{EXP}
  \end{subfigure}\\[5mm]
  \caption{Empirical and estimated MSEs for different second step subsample size $n$ based on $\pi^{\mvc}$ with the first step subsample size being fixed at $n_0=200$.}
  \label{fig:2}
\end{figure}

\begin{figure}[htp]
  \centering
    \includegraphics[width=0.5\textwidth]{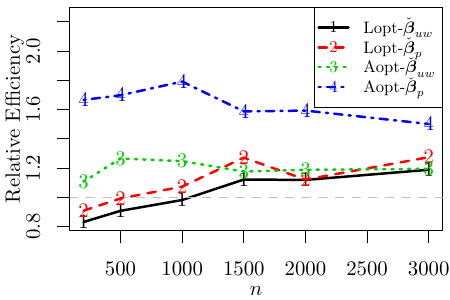}
  \caption{Relative efficiency for the SUSY data set with $n_0=200$ and different second step subsample size $n$. The gray horizontal dashed line is the reference line when relative efficiency is one.}
  \label{fig:3}
\end{figure}

\subsection{Computational efficiency}
\label{sec:comp-effic}
We consider the computational efficiency of the more efficient estimation methods in this section. Note that they have the same order of computational time complexity, so they should have similar computational efficiency as the weighted estimator. For Poisson subsampling, there is no need to calculate $\{\pi_i^p\}_{i=1}^N$ all at once and random numbers can be generated on the go, so it requires less RAM and may require less CPU times as well. To confirm this, we record the computing time of implementing each of them for the case  when $\x$ is mzNormal. All methods are implemented in the R programming language \citep{R}, and computations are carried out on a desktop running Ubuntu Linux 16.04 with an Intel I7 processor and 16GB RAM. Only one logical CPU is used for the calculation. We set the value of $d$ to $d=50$, the values of $N$ to be $N=10^4, 10^5, 10^6$ and $10^7$, and the subsample sizes to be $n_0=200$ and $n=1000$.

Table~\ref{tab:1} gives the required CPU times (in seconds) to obtain $\cbeta_w$, $\cbeta_{uw}$, and $\cbeta_p$, using $\pi_i^{\mvc}$ and $\pi_i^{\mmse}$. The computing times for using the full data
(Full) %
are also given for comparisons. It is seen that $\cbeta_{uw}$ and $\cbeta_p$ are a little faster than $\cbeta_w$ but the advantages are not very significant. {The reason is that the original OSMAC estimator $\cbeta_{uw}$ pools the pilot subsample with the second stage subsample and performs iterative calculations on the pilot subsample twice, while the proposed method combines the pilot estimator with the second stage estimator which only requires iterative calculations on the pilot subsample once. Since the pilot subsample size is small, the difference is not significant.} Note that these times are obtained when all the calculations are done in the RAM, and only the CPU times for implementing each method are counted while the time to generate the data is not counted.

\begin{table}[H]%
  \caption{CPU seconds when the full data are generated and kept in the RAM. Here $n_0=200$, $n=1000$, and the full data size $N$ varies; the covariates are from a $d=50$ dimensional multivariate normal distribution.}
\centering
\begin{tabular}{lrrrr}\hline
  Method & \multicolumn{4}{c}{$N$} \\
  \cline{2-5}
            & $10^4$ & $10^5$ & $10^6$ & $10^7$ \\ \hline
$\pi_i^{\mvc}$, $\cbeta_w$ & 0.14 & 0.13 & 0.45 & 5.24 \\ 
$\pi_i^{\mvc}$, $\cbeta_{uw}$ & 0.08 & 0.11 & 0.41 & 3.71 \\ 
$\pi_i^{\mvc}$, $\cbeta_p$ & 0.08 & 0.11 & 0.43 & 3.88 \\ 
$\pi_i^{\mmse}$, $\cbeta_w$ & 0.13 & 0.32 & 3.31 & 35.15 \\ 
$\pi_i^{\mmse}$, $\cbeta_{uw}$ & 0.12 & 0.31 & 3.29 & 34.98 \\ 
$\pi_i^{\mmse}$, $\cbeta_p$ & 0.12 & 0.31 & 3.29 & 35.06 \\ 
Full & 0.15 & 1.62 & 15.05 & 247.89 \\ 
  \hline
\end{tabular}
\label{tab:1}
\end{table}

For big data problem, it is common that the full data are larger than the size of the available RAM, and full data can not be loaded into the RAM. For this scenario, one has to load the data into RAM line-by-line or block-by-block. Note that communication between CPU and hard drive is much slower than communication between CPU and RAM. Thus, this will dramatically increase the computing time. To mimic this situation, we 
store the full data on hard drive and use \verb|readlines()| function to process data 1000 rows each time. We also use a smaller computer with 8GB RAM to implement the method. For the case when $N=10^7$, the full data is about 9.1GB which is larger than the available RAM. 

The computing times when data are scanned from hard drive are reported Table~\ref{tab:2}. Here the computing times can be over thousand times longer than those when data are loaded into RAM. Note that these computing times can be reduced dramatically if we use some other programming language like C++ \citep{stroustrup1986c} or Julia \citep{Julia}. However, for fair comparisons, we use the same programming language R here. Furthermore, our main purpose here is to demonstrate the computational advantage of subsampling so the real focus is on the relative performance among different methods. From Table~\ref{tab:2}, it is seen that using $\pi^{\mmse}$ does not cost much more time than using $\pi^{\mvc}$. The reason for this observation is that the major computing time is spent in data processing and the computing times used in calculating the subsampling probabilities are short. We also notice that Poisson subsampling is more computational efficient than subsampling with replacement since it calculates subsampling probabilities and generates random numbers on the go and requires one time less to scan the full data. Poisson subsampling only used about 2\% of the time required by implementing the full data approach. 

\begin{table}[H]%
  \caption{CPU seconds when the full data are scanned from hard drive. Here $n_0=200$, $n=1000$, and the full data size $N$ varies; the covariates are from a $d=50$ dimensional multivariate normal distribution.}
\centering
\begin{tabular}{lrrrr}\hline
  Method & \multicolumn{4}{c}{$N$} \\
  \cline{2-5}
         & $10^4$ & $10^5$ & $10^6$ & $10^7$ \\ \hline
$\pi_i^{\mvc}$, $\cbeta_w$  & 4.26   & 41.60   & 441.46   & 4374.94   \\ 
$\pi_i^{\mvc}$, $\cbeta_{uw}$  & 4.13   & 41.42   & 413.09   & 4384.99   \\ 
$\pi_i^{\mvc}$, $\cbeta_p$  & 2.77   & 27.58   & 272.32   & 2699.13   \\ 
$\pi_i^{\mmse}$, $\cbeta_w$ & 4.43   & 41.75   & 434.96   & 4393.38   \\ 
$\pi_i^{\mmse}$, $\cbeta_{uw}$ & 4.10   & 41.83   & 417.55   & 4369.04   \\ 
$\pi_i^{\mmse}$, $\cbeta_p$ & 2.88   & 27.93   & 273.24   & 2719.51   \\ 
Full        & 139.46 & 1411.78 & 14829.63 & 138134.69 \\   \hline
\end{tabular}
\label{tab:2}
\end{table}

\section{Summary}
\label{sec:summary}
In this paper, we have proposed a new un-weighted estimator for logistic regression based on an OSMAC subsample. We have derived conditional asymptotic distribution of the new estimator which has a smaller variance-covariance matrix compared with the weighted estimator.

We have also investigated the asymptotic properties if Poisson subsampling is used, and showed that the resultant estimator has the same conditional asymptotic distribution if the subsampling ratio converges to zero. %
However, if the subsampling ratio converges to a positive constant, %
the estimator based on Poisson subsampling has a smaller variance-covariance matrix.

In addition, we have derive the unconditional asymptotic distribution for the proposed estimator based on Poisson subsampling. Interestingly, if 
the subsampling ratio converges to zero, %
the unconditional asymptotic distribution is the same as the conditional asymptotic distribution, indicating that the variation of the full data can be ignored. If the subsampling ratio does not converge to zero, %
the unconditional asymptotic distribution has a larger variance-covariance matrix. Our results also include the local case-control sampling method. %
With a stronger moment condition that the third moment of the covariate is finite, we do not require the pilot estimate to be independent of the data.

Furthermore, %
we have proved consistency and asymptotic normality for the proposed estimators under two types of misspecifications: one is that pilot estimators %
are inconsistent, and the other is that the logistic regression model is misspecified.

\acks{The author would like to thank the reviewers for their insightful comments and suggestions, which greatly helped improve the paper and strengthen the proposed methods. %
  This work was partially supported by the NSF grant DMS-1812013 and a UConn REP grant.}

\appendix

\section{Subsampling with replacement from hard drive}
\label{sec:subs-with-repl}
If the full data can be loaded into available RAM, subsampling probabilities can be calculated in RAM and subsampling with replacement can be implemented directly. Otherwise, special considerations have to be given in practical implementation. If the full data is larger than available RAM while subsampling probabilities $\{\pi_i^{\os}(\hbeta_0)\}_{i=1}^{N}$ can still be loaded in available RAM, one can calculate $\{\pi_i^{\os}(\hbeta_0)\}_{i=1}^{N}$ by scanning the data from hard drive line-by-line or block-by-block, generate row indexes for a subsample, and then scan the data line-by-line or block-by-block to take the subsample. To be specific,
one can draw a subsample, say $\{idx_1, ..., idx_{n}\}$, from
$\{1, ..., N\}$, sort the indexes to have $\{idx_{(1)}, ..., idx_{(n)}\}$, and then use the Algorithm~\ref{alg:1} to scan the data
line-by-line or block-by-block in order to obtain the subsample.

\begin{algorithm}[H]%
  \caption{Obtain the subsample with the given indexes by scanning through the full data}
  \label{alg:1}
  \begin{algorithmic}
    \STATE {\bfseries Input:}
    data file, subsample indexes $\{idx_{(1)}, ..., idx_{(n)}\}$.
    \STATE $i \leftarrow 1$
    \STATE $j \leftarrow 1$
    \WHILE {$i \le N$ \AND $j \le n$}
    \STATE readline(data file)
    \IF {$i==idx_{(j)}$}
    \STATE include the $i$-th data point into the subsample
    \WHILE {$i==idx_{(j)}$}
    \STATE $j \leftarrow j+1$
    \ENDWHILE
    \ENDIF
    \STATE $i\leftarrow i+1$
    \ENDWHILE
    \end{algorithmic}
\end{algorithm}
Clearly, Algorithm~\ref{alg:1} takes no more than linear time to run.   Here, we assume that a generic function
  \verb|readline()| reads a single line (or multiple lines) from the
  data file and stop at the beginning of the next line (or next block)
  in the data
  file. No calculation is performed on a data line if it is not included in the subsample.
  Such functionality is provided by most programming languages. For example, Julia \citep{Julia} and Python \citep{vanRossum1995} has a function \verb|readline()| that read a file line-by-line; R \citep{R} has a function \verb|readLines()| that read one or multiple lines; C \citep{C} and C++ \citep{stroustrup1986c} has a function \verb|getline()| to read one line at a time. 

\section{Proofs and technical details}
\label{sec:proofs-techn-deta}
In this appendix, we provide proofs for the results in the paper.  Technical details related to sampling with replacement in Section~\ref{sec:more-effic-infer} are presented in Section~\ref{sec:proofs-subs-with}; technical details related to Poisson subsampling in Section~\ref{sec:poisson-sampling} are presented in Section~\ref{sec:proofs-poiss-subs}; technical details related to unconditional results in Section~\ref{sec:uncond-distr} are presented in Section~\ref{sec:proofs-uncond-distr}; and technical details related to model misspecification in Section~\ref{sec:misspecifications} are presented in Section~\ref{sec:proofs-pilot-missp}.

\subsection{Proofs for subsampling with replacement}
\label{sec:proofs-subs-with}
In this section we prove the results in
Section~\ref{sec:more-effic-infer}. %
For ease of presentation, we use notation $\lambda$ to denote the log-likelihood shifted by $\hbeta_0$. For the subsample, $\lambda_{uw}^*(\bbeta)=\ell_{uw}^*(\bbeta-\hbeta_0)$. Denote the first and second derivatives of $\lambda_{uw}^*(\bbeta)$ as $\dot\lambda_{uw}^*(\bbeta)=\partial\lambda_{uw}^*(\bbeta)/\partial\bbeta$ and $\ddot\lambda_{uw}^*(\bbeta)=\partial^2\lambda_{uw}^*(\bbeta)/(\partial\bbeta\partial\bbeta\tp)$.

Note that from \cite{xiong2008some,cheng2010bootstrap}, the fact that a sequence converges to 0 in conditional probability is equivalent to the fact that it converges to 0 in unconditional probability. This can also be proved directly by using the fact the probability measure is bounded by 1. Thus, in the following, we will use $\op$ to denote a sequence converging to 0 in probability without stating whether the underlying probability measure is  conditional or unconditional.

{We first present some lemmas that will be used to prove Theorem~\ref{thm:1}, and provide their proofs in Sections  \ref{sec:proof-lemma-reflem1} - \ref{sec:proof-lemma-reflem4}.}

\begin{lemma}\label{lem1}
  Let $\v_1, ..., \v_N$ be i.i.d. random vectors with the same distribution of $\v$. Let $g_{1N}$ be a bounded function %
  and $g_2$ be a fixed function that does not depend on $N$. If $g_{1N}(\v)=\op$ %
  and $\Exp|g_2(\v)|<\infty$, then
  \begin{equation*}
    \oneN\sumN g_{1N}(\v_i)g_2(\v_i)=\op.
  \end{equation*}
\end{lemma}

\begin{lemma}\label{lem2}
  Let $\eeta_i=|\psi_i(\hbeta_0)|\psi_i(\bbeta_t-\hbeta_0)h(\x_i)\x_i$, where $\psi_i(\bbeta)=y_i-p(\x_i,\bbeta)$. Under Assumptions~\ref{as:1} and \ref{as:2}, conditional on the consistent $\hbeta_0$, if $n_0/\sqrt{N}\rightarrow0$, then
  \begin{align*}
  \sqrt{N}(\hbeta_{\wmle}-\bbeta_t)
    =&\frac{\bSigma_{\bbeta_t}}
       {2\Exp\{\phi(\bbeta_t)h(\x)\}}
     \frac{1}{\sqrt{N}}\sumN\eeta_i+\op,
\end{align*}
which converges in distribution to a normal distribution with mean $\0$ and variance-covariance matrix $[\Exp\{\phi(\bbeta_t)h(\x)\x\x\tp\}]^{-1}
    \Exp\{\phi(\bbeta_t)h^2(\x)\x\x\}
    [\Exp\{\phi(\bbeta_t)h(\x)\x\x\tp\}]^{-1}$, as $n_0$ and $N$ go to infinity. 
\end{lemma}

\begin{lemma}\label{lem3}
  Let
  \begin{equation*}
  \dot\lambda_{uw}^*(\bbeta_t)
  =\sumn\big\{y_i^*-p_i^*(\bbeta_t-\hbeta_0)\big\}\x_i^*.
\end{equation*}
  Under Assumptions~\ref{as:1} and \ref{as:2}, conditional on $\Fn$ and the consistent $\hbeta_0$, as $n_0$, $n$, and $N$ go to infinity,
\begin{align*}
  \frac{\dot\lambda_{uw}^*(\bbeta_t)}{\sqrt{n}}
  -\frac{\sqrt{n}\sumN\eeta_i}{N\Psi_N(\hbeta_0)}
  \longrightarrow \Nor\big(\0,\ \bSigma_{\bbeta_t}^{-1}\big),
\end{align*}
in distribution.
\end{lemma}

\begin{lemma}\label{lem4}
  Under Assumptions~\ref{as:1}-\ref{as:3}, as $n_0$, $n$, and $N$ go to infinity, for any $\s_n\rightarrow0$ in probability, 
  \begin{align*}
  &\onen\sumn\phi_i^*(\bbeta_t-\hbeta_0+\s_n)\|\x_i^*\|^2
    -\sumN\pi_i(\hbeta_0)\phi_i(\bbeta_t-\hbeta_0)\|\x_i\|^2
    =\op.
\end{align*}
\end{lemma}

\begin{proof}{\bf of Theorem~\ref{thm:1}.}
  The estimator $\hbeta_{uw}$ is the maximizer of 
\begin{align*}
  \lambda_{uw}^*(\bbeta)=\sumn\Big[(\bbeta-\hbeta_0)\tp\x_i^*y_i^*
  -\log\big\{1+e^{(\bbeta-\hbeta_0)\tp\x_i^*}\big\}\Big],
\end{align*}
so $\sqrt{n}(\hbeta_{uw}-\bbeta_t)$ is the maximizer of $\gamma(\s)=\lambda_{uw}^*(\bbeta_t+\s/\sqrt{n})-\lambda_{uw}^*(\bbeta_t)$. By Taylor's expansion,
\begin{align*}
  \gamma(\s)
  &=\frac{1}{\sqrt{n}}\s\tp\dot\lambda_{uw}^*(\bbeta_t)
    +\frac{1}{2n}\sumn
    \phi_i^*(\bbeta_t-\hbeta_0+\Acute\s/\sqrt{n})
    (\s\tp\x_i^*)^2,
\end{align*}
where $\phi_i^*(\bbeta)=p_i^*(\bbeta)\{1-p_i^*(\bbeta)\}$, and $\Acute\s$ lies between $\0$ and $\s$.

From Lemma~\ref{lem4},
\begin{align*}
  &\onen\sumn\phi_i^*(\bbeta_t-\hbeta_0+\Acute\s/\sqrt{n})
    \x_i^*(\x_i^*)\tp
    -\sumN\pi_i(\hbeta_0)\phi_i(\bbeta_t-\hbeta_0)\x_i\x_i\tp=\op.
\end{align*}
From Lemma~\ref{lem1} and the law of large numbers,
\begin{align*}
  \sumN\pi_i(\hbeta_0)\phi_i(\bbeta_t-\hbeta_0)\x_i\x_i\tp
  =\frac{\oneN\sumN|\psi_i(\hbeta_0)|h(\x_i)
    \phi_i(\bbeta_t-\hbeta_0)\x_i\x_i\tp}{\Psi_N(\hbeta_0)}
  =\bSigma_{\bbeta_t}^{-1}+\op.
\end{align*}
Combining the above two equations, we have that $n^{-1}\sumn\phi_i^*(\bbeta_t-\hbeta_0+\Acute\s/\sqrt{n})\x_i^*(\x_i^*)\tp$ converges in probability to $\bSigma_{\bbeta_t}^{-1}$, a positive definite matrix. In addition, from Lemma~\ref{lem2} and Lemma~\ref{lem3}, $\dot\lambda_{uw}^*(\bbeta_t)/\sqrt{n}$ is stochastically bounded.
Thus, from the Basic Corollary in page 2 of \cite{hjort2011asymptotics}, the maximizer of $\gamma(\s)$, $\sqrt{n}(\hbeta_{uw}-\bbeta_t)$, satisfies
\begin{align*}
  \sqrt{n}(\hat\bbeta_{uw}-\bbeta_t)
  =\bSigma_{\bbeta_t}\frac{1}{\sqrt{n}}
  \dot\lambda_{uw}^*(\bbeta_t)+\op%
\end{align*}
given $\Fn$ and $\hbeta_0$. Thus, 
\begin{align}\label{eq:10}
  \sqrt{n}(\hat\bbeta_{uw}-\hbeta_{\wmle})
  =\bSigma_{\bbeta_t}
  \bigg\{\frac{1}{\sqrt{n}}\dot\lambda_{uw}^*(\bbeta_t)
  -\bSigma_{\bbeta_t}^{-1}\sqrt{n}(\hbeta_{\wmle}-\bbeta_t)\bigg\}+\op.
\end{align}
From Lemma~\ref{lem2},
  \begin{align}\label{eq:11}
  \bSigma_{\bbeta_t}^{-1}\sqrt{n}(\hbeta_{\wmle}-\bbeta_t)
  =\frac{\sqrt{n}\sumN\eeta_i}{2N\Exp\{\phi(\bbeta_t)h(\x)\}}
  =&\frac{\sqrt{n}\sumN\eeta_i}{N\Psi_N(\hbeta_0)}+\op,
\end{align}

Combining equations \eqref{eq:10} and \eqref{eq:11}, Lemma~\ref{lem3},  Slutsky's theorem, and the fact that a conditional probability is bounded by one, Theorem 1 follows.
\end{proof}

\subsubsection{Proof of Proposition~\ref{prop1}}
\begin{proof}{\bf of Proposition~\ref{prop1}.}
  To prove that $\bSigma_{\bbeta_t}\le \V^\os=\M^{-1}\V_{c}^\os\M^{-1}$, we just need to show that
  \begin{align*}
    \bSigma_{\hbeta_{\mle}}^{-1}\ge\M(\V_c^\os)^{-1}\M.
  \end{align*}
  From the strong law of large numbers,
\begin{align*}
  \M&=\oneN\sumN\phi_i(\bbeta_t)\x_i\x_i\tp+o(1),\\
  \V_{c}^\os&=4\Phi(\bbeta_t)
    \oneN\sumN\frac{\phi_i(\bbeta_t)\x_i\x_i\tp}{h(\x_i)}+o(1),\\
  \bSigma_{\bbeta_t}&=4\Phi(\bbeta_t)
  \bigg\{\oneN\sumN\phi_i(\bbeta_t)h(\x_i)\x_i\x_i\tp\bigg\}^{-1}+o(1),
\end{align*}
almost surely. Thus, we only need to verify that
\begin{equation*}%
  \sumN\phi_i(\bbeta_t)h(\x_i)\x_i\x_i\tp
  \ge\bigg\{\sumN\phi_i(\bbeta_t)\x_i\x_i\tp\bigg\}
  \bigg\{\sumN\frac{\phi_i(\bbeta_t)\x_i\x_i\tp}{h(\x_i)}\bigg\}^{-1}
  \bigg\{\sumN\phi_i(\bbeta_t)\x_i\x_i\tp\bigg\}.
\end{equation*}

Denote $\Z=\big\{\sqrt{\phi_1(\bbeta_t)}\x_1, ...,
\sqrt{\phi_N(\bbeta_t)}\x_N\big\}\tp$,
and $\H=\mathrm{diag}\{h(\x_1), ..., h(\x_N)\}$. The above inequality %
can be written as
\begin{align*}
  \Z\tp\H\Z \ge& \Z\tp\Z (\Z\tp\H^{-1}\Z)^{-1} \Z\tp\Z,
\end{align*}
which is true if 
\begin{equation}\label{eq:89}
  \H \ge \Z(\Z\tp\H^{-1}\Z)^{-1}\Z\tp
\end{equation}
Note that  $(\H^{-1/2}\Z)\{(\H^{-1/2}\Z)(\H^{-1/2}\Z)\tp\}^{-1}(\H^{-1/2}\Z)\tp$ is the projection matrix of $\H^{-1/2}\Z$, so it is, under the Loewner ordering, smaller than or equal to the identity matrix $\I_N$, namely,
\begin{equation*}
  \I_N\ge\H^{-1/2}\Z(\Z\H^{-1}\Z\tp)^{-1}\Z\tp\H^{-1/2},
\end{equation*}
which implies \eqref{eq:89}. If $h(\x)=1$, the equality can be verified directly.

The first inequality in \eqref{eq:60} can be verified directly using the result in \eqref{eq:12}. For the second inequality in \eqref{eq:60}, noting that $h(\x)=\|L\M^{-1}\x\|$, by the Cauchy–Schwarz inequality, we have
\begin{align*}
  \tr(L\V^\os L\tp)
  &=\tr(L\M^{-1}\V_{c}^\os\M^{-1} L\tp)\\
  &=4\Phi(\bbeta_t)\tr\big[L\M^{-1}
    \Exp\big\{\phi(\bbeta_t)h^{-1}(\x)
    \x\x\tp\big\}\M^{-1}L\tp\big]\\
  &=4\Exp\{\phi(\bbeta_t)h(\x)\}
  \Exp\big\{\phi(\bbeta_t)h^{-1}(\x)\|L\M^{-1}\x\|^2\big\}\\
  &=4\big[\Exp\big\{\phi(\bbeta_t)\|L\M^{-1}\x\|\big\}\big]^2\\
  &\le4\Exp\{\phi(\bbeta_t)\}
    \Exp\big\{\phi(\bbeta_t)\|L\M^{-1}\x\|^2\big\}\\
  &=4\Exp\{\phi(\bbeta_t)\}
    \tr\big[L\M^{-1}\Exp\big\{\phi(\bbeta_t)\x\x\tp\big\}
    \M^{-1}L\tp\big]\\
  &=4\Exp\{\phi(\bbeta_t)\}\tr(L\M^{-1}L\tp)\\
  &<\tr(L\M^{-1}L\tp),
\end{align*}
which finishes the proof.
\end{proof}

\subsubsection{Proof of Lemma \ref{lem1}}
\label{sec:proof-lemma-reflem1}
\begin{proof}{\bf of Lemma~\ref{lem1}.}
  Let $B$ be a bound for $g_{1N}$ i.e., $|g_{1N}|\le B$. For any $\epsilon>0$, by Markov's inequality,
  \begin{align*}
    &\Pr\bigg\{\Big|\oneN\sumN g_{1N}(\v_i)g_2(\v_i)\Big|>\epsilon\bigg\}
    \le\frac{\Exp|g_{1N}(\v)g_2(\v)|}{\epsilon}\\
    &=\frac{\Exp[|g_{1N}(\v)||g_2(\v)|I\{|g_2(\v)|\le K\}]}{\epsilon}
    +\frac{\Exp[|g_{1N}(\v)||g_2(\v)|I\{|g_2(\v)|> K\}]}{\epsilon}\\
    &\le\frac{K}{\epsilon}\Exp|g_{1N}(\v)|
    +\frac{B}{\epsilon}\Exp\{|g_2(\v)|I(|g_2(\v)|>K)\}.
  \end{align*}
  For any $\zeta>0$, we can choose a $K$ large enough such that $\Exp\{|g_2(\v)|I(|g_2(\v)|\le K)\}<\zeta\epsilon/(2B)$, since $\Exp|g_2(\v)|<\infty$. The facts that $g_{1N}(\v_i)\le B$ and $g_{1N}(\v_i)=\op$ imply that $\Exp|g_{1N}(\v)|=o(1)$. Thus, there is a $N_\zeta$ such that $\Exp|g_{1N}(\v)|<\zeta\epsilon/(2K)$ when $N>N_\zeta$. Therefore, for any $\zeta>0$, $\Pr\{|N^{-1}\sumN g_{1N}(\v_i)g_2(\v_i)|>\epsilon\}<\zeta$ for sufficiently large $N$. This finishes the proof. 
\end{proof}

\subsubsection{Proof of Lemma \ref{lem2}}
\label{sec:proof-lemma-reflem2}
\begin{proof}{\bf of Lemma~\ref{lem2}.}
  Since $\hbeta_{\wmle}$ is the maximizer of
  \begin{align}
  \lambda_{\wmle}(\bbeta)\notag=\sumN|y_i-p(\x_i,\hbeta_0)|h(\x_i)
    \big[y_i\x_i\tp(\bbeta-\hbeta_0)
  -\log\{1+e^{\x_i\tp(\bbeta-\hbeta_0)}\}\big],
\end{align}
{$\sqrt{N}(\hbeta_{\wmle}-\bbeta_t)$ is the maximizer of} $\gamma_{\wmle}(\s)=\lambda_{\wmle}(\bbeta_t+\s/\sqrt{N})-\lambda_{\wmle}(\bbeta_t)$. By Taylor's expansion,
\begin{align*}
  \gamma_{\wmle}(\s)
  &=\frac{1}{\sqrt{N}}\s\tp\dot\lambda_{\wmle}(\bbeta_t)
    -\frac{1}{2N}\sumN|y_i-p(\x_i,\hbeta_0)|h(\x_i)
    \phi_i(\bbeta_t-\hbeta_0+\Acute\s/\sqrt{N})(\s\tp\x_i)^2
\end{align*}
where
\begin{align}\label{eq:8}
  \dot\lambda_{\wmle}(\bbeta_t)=\sumN\eeta_i
  =\sumN |\psi_i(\hbeta_0)|\psi_i(\bbeta_t-\hbeta_0)h(\x_i)\x_i,
\end{align}
and $\Acute\s$ lies between $\0$ and $\s$.

{
Since that $n_0/\sqrt{N}\rightarrow0$ and $\|\eeta_i\|$ is bounded by $\|h(\x_i)\x_i\|$, we can ignore the data points that are used in obtaining the pilot $\hbeta_0$. The following discussion focus on $\eeta_i$'s for which $(\x_i,y_i)$'s are not included in the pilot subsample.} Note that for such $\eeta_i$'s, $\Exp(\eeta_i|\hbeta_0)=\0$ because
\begin{align}
  \Exp(\eeta_i|\hbeta_0,\x_i)
  =&\Exp\Big[|\psi_i(\hbeta_0)|
    \big\{y_i-p(\x_i,\bbeta_t-\hbeta_0)\big\}h(\x_i)\x_i
    \Big|\hbeta_0,\x_i\Big]\notag\\
  =&-p(\x_i,\hbeta_0)
    p(\x_i,\bbeta_t-\hbeta_0)\{1-p(\x_i,\bbeta_t)\}h(\x_i)\x_i\notag\\
    &+\{1-p(\x_i,\hbeta_0)\}
    \{1-p(\x_i,\bbeta_t-\hbeta_0)\}p(\x_i,\bbeta_t)h(\x_i)\x_i\notag\\
  =&-\frac{e^{\x_i\tp\hbeta_0}}{1+e^{\x_i\tp\hbeta_0}}
    \frac{e^{\x_i\tp\bbeta_t-\x_i\tp\hbeta_0}}
    {1+e^{\x_i\tp\bbeta_t-\x_i\tp\hbeta_0}}
    \frac{1}{1+e^{\x_i\tp\bbeta_t}}h(\x_i)\x_i\notag\\
  &+\frac{1}{1+e^{\x_i\tp\hbeta_0}}
    \frac{1}{1+e^{\x_i\tp\bbeta_t-\x_i\tp\hbeta_0}}
    \frac{e^{\x_i\tp\bbeta_t}}{1+e^{\x_i\tp\bbeta_t}}h(\x_i)\x_i=0.\label{eq:22}
\end{align}
This also gives that
\begin{align*}
  \Var(\eeta_i|\hbeta_0)
  =\Exp\big\{\Var(\eeta_i|\x_i,\hbeta_0)\big|\hbeta_0\big\}
  +\Var\big\{\Exp(\eeta_i|\x_i,\hbeta_0)\big|\hbeta_0\big\}
  =\Exp\big\{\Var(\eeta_i|\x_i,\hbeta_0)\big|\hbeta_0\big\}.
\end{align*}
Now, since
\begin{align*}
  \Var(\eeta_i|\x_i,\hbeta_0)
  =&\Exp\Big[|y_i-p(\x_i,\hbeta_0)|^2
    \big\{y_i-p(\x_i,\bbeta_t-\hbeta_0)\big\}^2h^2(\x_i)\x_i\x_i\tp
    \Big|\hbeta_0,\x_i\Big]\\
  =&p(\x_i,\bbeta_t)\{1-p(\x_i,\hbeta_0)\}^2
     \big\{1-p(\x_i,\bbeta_t-\hbeta_0)\big\}^2h^2(\x_i)\x_i\x_i\tp\\
  &+\{1-p(\x_i,\bbeta_t)\}\{p(\x_i,\hbeta_0)\}^2
    \big\{p(\x_i,\bbeta_t-\hbeta_0)\big\}^2h^2(\x_i)\x_i\x_i\tp\\
  =&\phi_i(\hbeta_0)\phi_i(\hbeta_0-\bbeta_t)h^2(\x_i)\x_i\x_i\tp,
\end{align*}
we have
\begin{align*}
  \Var(\eeta_i|\hbeta_0)=\Exp\Big\{
  \phi(\hbeta_0)\phi(\hbeta_0-\bbeta_t)h^2(\x_i)\x_i\x_i\tp
  \Big|\hbeta_0\Big\}.
\end{align*}
Let $\|\|$ denote the Frobenius norm if applied on a martix, i.e., for a matrix $A$, $\|A\|^2=\tr(AA\tp)$, and denote 
$\Var(\eeta_i|\bbeta_t) = \Exp\{\phi(\bbeta_t)\phi(\bbeta_t-\bbeta_t)h^2(\x_i)\x_i\x_i\tp\}
=0.25\Exp\{\phi(\bbeta_t)h^2(\x_i)\x_i\x_i\tp\}$. 
Notice that $\big|\phi(\hbeta_0)\phi(\hbeta_0-\bbeta_t)
       -0.25\phi_i(\bbeta_t) \big|h^2(\x_i)\|\x_i\|^2$ converges to 0 in probability and it is bounded by $h^2(\x_i)\|\x_i\|^2$, an integrable random variable under Assumption~\ref{as:2}. Thus,
\begin{align*}
  &\Exp\big\|\Var(\eeta_i|\hbeta_0)-\Var(\eeta_i|\bbeta_t)\big\|
  \le\Exp\big\{\big|\{\phi(\hbeta_0)\phi(\hbeta_0-\bbeta_t)
    -0.25\phi_i(\bbeta_t) \big|h^2(\x_i)\|\x_i\|^2\big\}=o(1).
\end{align*}
 This implies that
\begin{align*}
  &\Var(\eeta_i|\hbeta_0)=\Var(\eeta_i|\bbeta_t)+\op
    =0.25\Exp\big\{\phi(\bbeta_t)h^2(\x_i)\x_i\x_i\big\}+\op.
\end{align*}

For $\eeta_i$'s that $(\x_i,y_i)$'s are not included in the pilot subsample, conditional on $\hbeta_0$, $\eeta_i$'s are i.i.d. with mean $\0$ and variance $\Var(\eeta_i|\hbeta_0)$. Since for any $\epsilon>0$, 
\begin{align*}
  \oneN\sumN\Exp\Big\{
    \|\eeta_i\|^2I(\|\eeta_i\|>\sqrt{N}\epsilon)\Big|\hbeta_0\Big\}
  \le&\oneN\sumN\Exp\Big\{\|h(\x_i)\x_i\|^2
       I(\|h(\x_i)\x_i\|>\sqrt{N}\epsilon)\Big|\hbeta_0\Big\}\\
  =&\Exp\{\|h(\x)\x\|^2
       I(\|h(\x)\x\|>\sqrt{N}\epsilon)\}\rightarrow0,
\end{align*}
the Lindeberg-Feller central limit theorem \citep[Section $^*$2.8 of][]{Vaart:98} applies conditional on $\hbeta_0$. Thus, we have, conditional on $\hbeta_0$, 
\begin{align*}
  \frac{\dot\lambda_{\wmle}(\bbeta_t)}{\sqrt{N}}
    \longrightarrow \Nor\bigg[\0,\
  \frac{\Exp\big\{\phi(\bbeta_t)h^2(\x)\x\x\big\}}{4}
  \bigg],
\end{align*}
in distribution. 
From Lemma~\ref{lem1}, conditional on $\hbeta_0$,  
\begin{align*}
 &\oneN\sumN|y_i-p(\x_i,\hbeta_0)|h(\x_i)
    \phi_i(\bbeta_t-\hbeta_0+\Acute\s/\sqrt{N})\x_i\x_i\tp\\
    &=\frac{1}{4}\Exp\{|\psi(\bbeta_t)|h(\x)\x\x\tp\}+\opb
    =\frac{1}{2}\Exp\{\phi(\bbeta_t)h(\x)\x\x\tp\}+\op.
\end{align*}
Thus, from the Basic Corollary in page 2 of \cite{hjort2011asymptotics}, the maximizer of $\gamma_{\wmle}(\s)$, $\sqrt{N}(\hbeta_{\wmle}-\bbeta_t)$, satisfies
\begin{align}\label{eq:7}
  \sqrt{N}(\hbeta_{\wmle}-\bbeta_t)
  =&2[\Exp\{\phi(\bbeta_t)h(\x)\x\x\tp\}]^{-1}
    \frac{1}{\sqrt{N}}\dot\lambda_{\wmle}(\bbeta_t)+\op.
\end{align}
Note that
\begin{equation}\label{eq:9}
  [\Exp\{\phi(\bbeta_t)h(\x)\x\x\tp\}]^{-1}
  =\frac{\bSigma_{\bbeta_t}}{4\Phi(\bbeta_t)}.
\end{equation}
Combining equations~\eqref{eq:8}, \eqref{eq:7}, and \eqref{eq:9}, we have
\begin{align*}
  \sqrt{N}(\hbeta_{\wmle}-\bbeta_t)
  =&\frac{\bSigma_{\bbeta_t}}{2\Phi(\bbeta_t)}
    \frac{1}{\sqrt{N}}\dot\lambda_{\wmle}(\bbeta_t)+\op.
\end{align*}
An application of Slutsky's theorem yields the result for the asymptotic normality. 

\end{proof}

\subsubsection{Proof of Lemma \ref{lem3}} %

\begin{proof}{\bf of Lemma~\ref{lem3}.}
  Note that given $\Fn$ and $\hbeta_0$, $\{y_i^*-p_i^*(\bbeta_t-\hbeta_0)\big\}\x_i^*$ are i.i.d. random vectors. We now exam their mean and variance, and check the Lindeberg-Feller condition \citep[Section $^*$2.8 of][]{Vaart:98} under the conditional distribution given $\Fn$ and $\hbeta_0$. 
For the expectation, we have, 
\begin{align*}
  &\Exp\big[\big\{y^*-p(\x_i^*,\bbeta_t-\hbeta_0)\big\}\x^*
    \big|\Fn,\hbeta_0\big]
  =\sumN\pi_i(\hbeta_0)\psi_i(\bbeta_t-\hbeta_0)\x_i
  =\frac{\sumN\eeta_i}{N\Psi_N(\hbeta_0)},
\end{align*}
where $\Psi_N(\bbeta)=N^{-1}\sumN|y_i-p(\x_i,\bbeta)|h(\x_i)$. 
From Lemma~\ref{lem2} and its proof, $\sumN\eeta_i=O_P(\sqrt{N})$ conditional on $\hbeta_0$ in probability, i.e., for any $\epsilon>0$, there exits a $K$ such that $\Pr\{\Pr(\sumN\eeta_i/\sqrt{N}>K\big|\hbeta_0)<\epsilon\}\rightarrow1$ as $n_0, N\rightarrow\infty$. %
From \cite{xiong2008some}, we know that $\sumN\eeta_i=O_P(\sqrt{N})$ unconditionally. Thus, for the expectation, we have
\begin{align*}
  \Delta
  =&\Exp\big[\big\{y^*-p(\x_i^*,\bbeta_t-\hbeta_0)\big\}\x^*\big|\Fn,\hbeta_0\big]
  =O_P(1/\sqrt{N}).
\end{align*}

For the variance,
\begin{align*}
  &\Var\big[\big\{y^*-p(\x_i^*,\bbeta_t-\hbeta_0)\big\}\x^*
    |\Fn,\hbeta_0\big]\\
  =&\sumN\pi_i(\hbeta_0)\{y_i-p(\x_i,\bbeta_t-\hbeta_0)\}^2
     \x_i\x_i\tp-\Delta^2\\
  =&\frac{\oneN\sumN|\psi_i(\hbeta_0)|
    \psi_i^2(\bbeta_t-\hbeta_0)h(\x_i)\x_i\x_i\tp}
    {\Psi_N(\hbeta_0)}-O_P(1/N)\\
  =&\frac{\oneN\sumN|\psi_i(\bbeta_t)|
    (y_i-0.5)^2h(\x_i)\x_i\x_i\tp}
    {\Psi_N(\bbeta_t)}+\op\\
  =&\frac{1}{4}\frac{\Exp\{|\psi(\bbeta_t)|h(\x)\x\x\tp\}}
     {\Psi(\bbeta_t)}+\op
  =\bSigma_{\bbeta_t}^{-1}+\op
\end{align*}
where the third equality is from Lemma~\ref{lem1} and the fact that $\Exp\{h(\x)\|\x\|^2\}<\infty$, and the forth equality is from the law of large numbers. 

Now we check the Lindeberg-Feller condition \citep[Section $^*$2.8 of][]{Vaart:98} under the conditional distribution. Denote $\dot\lambda_{ri}^*=\{y_i^*-p_i^*(\bbeta_t-\hbeta_0)\}\x_i^*$.
\begin{align*}
  &\Exp\onen\sumn\big\{\|\dot\lambda_{ri}^*\|^2
    I(\|\dot\lambda_{ri}^*\|>\sqrt{n}\epsilon)
    \big|\Fn,\hbeta_0\big\}\\
  \le&\Exp\big\{\|\x^*\|^2I(\|\x\|>\sqrt{n}\epsilon)
    \big|\Fn,\hbeta_0\big\}\\
  =&\sumN\pi_i(\hbeta_0)\big\{\|\x_i\|^2
    I(\|\x_i\|>\sqrt{n}\epsilon)\big\}\\
  \le&\frac{\oneN\sumN\big\{h(\x_i)\|\x_i\|^2
    I(\|\x_i\|>\sqrt{n}\epsilon)\big\}}{\Psi_N(\hbeta_0)}\\
  \le&\frac{\oneN\sumN\big\{h(\x_i)\|\x_i\|^2
       I(\|\x_i\|>\sqrt{n}\epsilon)\big\}}{\Psi_N(\hbeta_0)}=\op,
\end{align*}
by Lemma~\ref{lem1}  and the fact that $\Exp\{h(\x)\|\x\|^2\}<\infty$. Thus, applying the Lindeberg-Feller central limit theorem \citep[Section $^*$2.8 of][]{Vaart:98} finishes the proof.
\end{proof}

\subsubsection{Proof of Lemma \ref{lem4}}
\label{sec:proof-lemma-reflem4}
\begin{proof}{\bf of Lemma~\ref{lem4}.}
We begin with the following partition,
\begin{align*}
  &\onen\sumn\phi_i^*(\bbeta_t-\hbeta_0+\s_n)\|\x_i^*\|^2\\
  =&\onen\sumn\phi_i^*(\bbeta_t-\hbeta_0+\s_n)\|\x_i^*\|^2
     I(\|\x_i^*\|^2\le n)
     +\onen\sumn\phi_i^*(\bbeta_t-\hbeta_0+\s_n)\|\x_i^*\|^2
     I(\|\x_i^*\|^2> n)\notag\\
  \equiv&\Delta_1+\Delta_2.
\end{align*}
  The second term $\Delta_2$ is $\op$ because it is non-negative and
\begin{align*}
  \Exp(\Delta_2|\Fn,\hbeta_0)
  &=\sumN\pi_i(\hbeta_0)\phi_i(\bbeta_t-\hbeta_0+\s_n)
    \|\x_i\|^2I(\|\x_i\|^2> n)\\
  &\le\frac{\sumN|\psi_i(\hbeta_0)|h(\x_i)
    \|\x_i\|^2I(\|\x_i\|^2> n)}
    {\sumN|\psi_i(\hbeta_0)|h(\x_i)}\\
  &\le\frac{\oneN\sumN h(\x_i)\|\x_i\|^2I(\|\x_i\|^2> n)}
    {\Psi_N(\hbeta_0)}
  =\op
\end{align*}
as $n,N\rightarrow\infty$, where the last step is from Lemma~\ref{lem1}. 

Similarly, we can show that
\begin{align*}
  \Exp(\Delta_1|\Fn,\hbeta_0)
  -\sumN\pi_i(\hbeta_0)\phi_i(\bbeta_t-\hbeta_0)\|\x_i\|^2
  =\op. 
\end{align*}
Thus, we only need to show that $\Delta_1-\Exp(\Delta_1|\Fn,\hbeta_0)=\op$. For this, we show that the conditional variance of $\Delta_1$ goes to 0 in probability. Notice that
\begin{align*}
  &\Var(\Delta_1|\Fn,\hbeta_0)\\
  =&\onen\Var\big\{
     \phi^*(\bbeta_t-\hbeta_0)\|\x^*\|^2
     I(\|\x^*\|^2\le n)\big|\Fn,\hbeta_0\big\}\\
  \le&\frac{1}{16n}\Exp\big\{\|\x^*\|^4
     I(\|\x^*\|^2\le n)\big|\Fn,\hbeta_0\big\}\\
  =&\frac{1}{16n}\sumn\Exp\big\{\|\x^*\|^4
     I(i-1<\|\x^*\|^2\le i)\big|\Fn,\hbeta_0\big\}\\
  \le&\frac{1}{16n}\sumn i^2\Exp\big\{I(i-1<\|\x^*\|^2\le i)
       \big|\Fn,\hbeta_0\big\}\\
  \le&\frac{1}{16n}\sumn i^2
       \big\{\Pr(\|\x^*\|^2>i-1\big|\Fn,\hbeta_0)
  -\Pr(\|\x^*\|^2>i\big|\Fn,\hbeta_0)\big\}\\
  =&\frac{1}{16n}\bigg\{\Pr(\|\x^*\|^2>0\big|\Fn,\hbeta_0)
     -n^2\Pr(\|\x^*\|^2>n\big|\Fn,\hbeta_0)
     +\sum_{i=1}^{n-1}(2i+1)
     \Pr(\|\x^*\|^2>i\big|\Fn,\hbeta_0)\bigg\}\notag\\
  \le&\frac{1}{16n}\bigg\{1+\sum_{i=1}^{n-1}
       3i\Pr(\|\x^*\|^2>i\big|\Fn,\hbeta_0)\bigg\}
\end{align*}
This is $\op$ because
\begin{align*}
  &\onen\sumn i\Pr(\|\x^*\|^2>i|\Fn,\hbeta_0)
  =\onen\sumn i\sum_{j=1}^n\pi_j(\hbeta_0)I(\|\x_j\|^2>i)\\
  &=\onen\sum_{j=1}^n\sumn i\pi_j(\hbeta_0)I(\|\x_j\|^2>i)\\
  &=\frac{\oneN\sum_{j=1}^n\onen\sumn
    i|\psi_j(\hbeta_0)|h(\x_j)I(\|\x_j\|^2>i)}
    {\Psi_N(\hbeta_0)}\\
  &\le\frac{\oneN\sum_{j=1}^n\onen\sumn
    ih(\x_j)I(\|\x_j\|^2>i)}
    {\Psi_N(\hbeta_0)},
\end{align*}
and the numerator is non-negative and has an expectation
\begin{align*}
    &\oneN\sum_{j=1}^n\onen\sumn i\Exp\{h(\x)I(\|\x\|^2>i)\}
\end{align*}
which is $o(1)$ since $i\Exp\{h(\x)I(\|\x\|^2>i)\}=o(1)$ as $i\rightarrow\infty$.
\end{proof}

\subsection{Proofs for Poisson subsampling}
\label{sec:proofs-poiss-subs}
In this section we prove the results in
Section~\ref{sec:poisson-sampling} about Poisson subsampling.

Define $\delta_i^{\hbeta_0}=I\{u_i\le n\pi_i^p(\hbeta_0)\}$, and use notation $\lambda_p$ to denote the log-likelihood shifted by $\hbeta_0$, i.e., $\lambda_p(\bbeta)=\ell_p^*(\bbeta-\hbeta_0)$. Using these notations, the estimator $\hbeta_p$ is the maximizer of
\begin{align}\label{eq:40}
  \lambda_p(\bbeta)=\sumN\delta_i^{\hbeta_0}\{n\pi_i^p(\hbeta_0)\vee1\}
  \big[(\bbeta-\hbeta_0)\tp\x_iy_i
  -\log\{1+e^{(\bbeta-\hbeta_0)\tp\x_i}\}\big],
\end{align}

Denote the first and second derivatives of $\lambda_p(\bbeta)$ as $\dot\lambda_p(\bbeta)=\partial\lambda_p(\bbeta)/\partial\bbeta$ and $\ddot\lambda_p(\bbeta)=\partial^2\lambda_p(\bbeta)/(\partial\bbeta\partial\bbeta\tp)$. 
Two lemmas similar to Lemmas~\ref{lem3} and \ref{lem4} are derived below which will be used to prove Theorem \ref{thm:2}. {We will prove these two lemmas in Sections \ref{sec:proof-lemma-reflem5} and \ref{sec:proof-lemma-reflem6}}.

\begin{lemma}\label{lem5}
  Let
  \begin{equation*}
  \dot\lambda_p(\bbeta_t)
  =\sumN\delta_i^{\hbeta_0}\{n\pi_i^p(\hbeta_0)\vee1\}
  \{y_i-p(\x_i,\bbeta_t-\hbeta_0)\}\x_i.
\end{equation*}
  Under Assumptions~\ref{as:1} and \ref{as:2}, conditional on $\Fn$, the consistent estimator $\hbeta_0$, and $\hat\Psi_0$, if $n=o(N)$, then
\begin{align*}
  \frac{\dot\lambda_p(\bbeta_t)}{\sqrt{n}}
  -\frac{\sqrt{n}\sumN\eeta_i}{N\Psi_N(\hbeta_0)}
  \longrightarrow \Nor\big(\0,\ \bSigma_{\bbeta_t}^{-1}\big),
\end{align*}
in distribution; if $n/N\rightarrow\rho\in(0,1)$, then
\begin{align*}
  \frac{\dot\lambda_p(\bbeta_t)}{\sqrt{n}}
  -\frac{\sqrt{n}\sumN\eeta_i}{N\Psi_N(\hbeta_0)}
  \longrightarrow \Nor\big(\0,\ \bLambda_{\rho}\big),
\end{align*}
in distribution. 
\end{lemma}

\begin{lemma}\label{lem6}
  Under Assumptions~\ref{as:1} and \ref{as:2}, as $n_0$, $n$, and $N$ go to infinity, for any $\s_n\rightarrow0$ in probability, 
  \begin{equation*}
    \onen\sumN\delta_i^{\hbeta_0}\{n\pi_i^p(\hbeta_0)\vee1\}
    \phi_i(\bbeta_t-\hbeta_0+\s_n)\|\x_i\|^2
    -\sumN\pi_i^p(\hbeta_0)\phi_i(\bbeta_t-\hbeta_0)\|\x_i\|^2
    =\op.
  \end{equation*}
\end{lemma}

\begin{proof}{\bf of Theorem \ref{thm:2}.}
The estimator $\hbeta_p$ is the maximizer of \eqref{eq:40}, 
so $\sqrt{n}(\hbeta_p-\bbeta_t)$ is the maximizer of $\gamma_p(\s)=\lambda_p(\bbeta_t+\s/\sqrt{n})-\lambda_p(\bbeta_t)$. By Taylor's expansion,
\begin{align*}
  \gamma_p(\s)
  &=\frac{1}{\sqrt{n}}\s\tp\dot\lambda_p(\bbeta_t)
    +\frac{1}{2n}\sumN\delta_i^{\hbeta_0}\{n\pi_i^p(\hbeta_0)\vee1\}
    \phi_i(\bbeta_t-\hbeta_0+\Acute\s/\sqrt{n})(\s\tp\x_i)^2
\end{align*}
where $\phi_i(\bbeta)=p(\x_i,\bbeta)\{1-p(\x_i,\bbeta)\}$, and $\Acute\s$ lies between $\0$ and $\s$.

From Lemmas~\ref{lem5} and \ref{lem6}, conditional on $\Fn$, and $\hbeta_0$,  
\begin{align*}
 &\onen\sumN\delta_i^{\hbeta_0}\{n\pi_i^p(\hbeta_0)\vee1\}
    \phi_i(\bbeta_t-\hbeta_0+\Acute\s/\sqrt{n})\x_i\x_i\tp
    =\bSigma_{\bbeta_t}^{-1}+\op.
\end{align*}
In addition, from Lemma~\ref{lem5}, conditional on $\Fn$, $\hbeta_0$, and $\hat\Psi_0$, $\dot\lambda_p(\bbeta_t)/\sqrt{n}$ converges in distribution to a normal limit. Thus, from the Basic Corollary in page 2 of \cite{hjort2011asymptotics}, the maximizer of $\gamma_p(\s)$, $\sqrt{n}(\hbeta_p-\bbeta_t)$, satisfies
\begin{align*}
  \sqrt{n}(\hat\bbeta_p-\bbeta_t)
  =\bSigma_{\bbeta_t}\frac{1}{\sqrt{n}}\dot\lambda_p(\bbeta_t)+\op
\end{align*}
given $\Fn$, $\hbeta_0$, and $\hat\Psi_0$. Combining this with {Lemma~\ref{lem5}, Slutsky's theorem, and the fact that a conditional probability is bounded,} Theorem~\ref{thm:2} follows. 
\end{proof}

\subsubsection{Proof of Proposition~\ref{prop2}}
\begin{proof}{\bf of Proposition~\ref{prop2}.}
  To prove that $\bSigma_{\bbeta_t}\bLambda_{\rho}\bSigma_{\bbeta_t}<\bSigma_{\bbeta_t}$, we just need to show that $\bLambda_{\rho}<\bSigma_{\bbeta_t}^{-1}$. This is true because 
\begin{align*}
  \bLambda_{\rho}
  &=\frac{\Exp\big[|\psi(\bbeta_t)|
  \{\Psi(\bbeta_t)-\rho|\psi(\bbeta_t)|h(\x)\}_+h(\x)\x\x\tp\big]}
    {4\Psi^2(\bbeta_t)}\\
  &<\frac{\Exp\big\{|\psi(\bbeta_t)|
  \Psi(\bbeta_t)h(\x)\x\x\tp\big\}}{4\Psi^2(\bbeta_t)}
  =\frac{\Exp\big\{|\psi(\bbeta_t)|
  h(\x)\x\x\tp\big\}}{4\Psi(\bbeta_t)}
  =\frac{\Exp\big\{\phi(\bbeta_t)h(\x)\x\x\tp\big\}}{4\Phi(\bbeta_t)}
  =\bSigma_{\bbeta_t}^{-1}.
\end{align*}
\end{proof}

\subsubsection{Proof of Lemma \ref{lem5}}
\label{sec:proof-lemma-reflem5}
\begin{proof}{\bf of Lemma~\ref{lem5}.}
Note that, $\delta_i^{\hbeta_0}=I\{u_i\le n\pi_i^p(\hbeta_0)\}$, where $u_i$ are i.i.d. with the standard uniform distribution. Thus, given $\Fn$, $\hbeta_0$, and $\hat\Psi_0$, $\dot\lambda_p(\bbeta_t)$ is a sum of $N$ independent random vectors. We now exam the mean and variance of $\dot\lambda_p(\bbeta_t)$. Recall that $\eeta_i=|\psi_i(\hbeta_0)|\psi_i(\bbeta_t-\hbeta_0)h(\x_i)\x_i$, and $\psi_i(\bbeta)=y_i-p(\x_i,\bbeta)$. 
For the mean, we have, 
\begin{align*}
  &\frac{1}{\sqrt{n}}\Exp
  \big\{\dot\lambda_p(\bbeta_t)|\Fn,\hbeta_0,\hat\Psi_0\big\}\\
  &=\frac{1}{\sqrt{n}}\sumN
    \{n\pi_i^p(\hbeta_0)\wedge1\}\{n\pi_i^p(\hbeta_0)\vee1\}
    \psi_i(\bbeta_t-\hbeta_0)\x_i\\
  &=\frac{1}{\sqrt{n}}\sumN
    n\pi_i^p(\hbeta_0)\psi_i(\bbeta_t-\hbeta_0)\x_i
    =\frac{\sqrt{n}}{\sqrt{N}}
    \frac{\sumN\eeta_i}{\hat\Psi_0\sqrt{N}}=O_P(\sqrt{n/N}),
\end{align*}
where the last equality is from Lemma~\ref{lem2}. 

For the variance,
\begin{align}
  \onen&\Var\big\{
    \dot\lambda_p(\bbeta_t)|\Fn,\hbeta_0,\hat\Psi_0\big\}\notag\\
  =&\onen\sumN[\{n\pi_i^p(\hbeta_0)\wedge1\}
     -\{n\pi_i^p(\hbeta_0)\wedge1\}^2]\{n\pi_i^p(\hbeta_0)\vee1\}^2
    \psi_i^2(\bbeta_t-\hbeta_0)\x_i\x_i\tp\notag\\
  =&\sumN\pi_i^p(\hbeta_0)\{n\pi_i^p(\hbeta_0)\vee1\}
     \psi_i^2(\bbeta_t-\hbeta_0)\x_i\x_i\tp
  -n\sumN\{\pi_i^p(\hbeta_0)\}^2
    \psi_i^2(\bbeta_t-\hbeta_0)\x_i\x_i\tp\notag\\
  =&\frac{\oneN\sumN
     |\psi_i(\hbeta_0)|\{n\pi_i^p(\hbeta_0)\vee1\}
    \psi_i^2(\bbeta_t-\hbeta_0)h(\x_i)\x_i\x_i\tp}
    {\hat\Psi_0}\notag\\
  &-\frac{n}{N}\frac{\oneN\sumN\psi_i^2(\hbeta_0)
    \psi_i^2(\bbeta_t-\hbeta_0)h^2(\x_i)\x_i\x_i\tp}
    {\hat\Psi_0^2}\notag\\
 \equiv&\Delta_3-\Delta_4\label{eq:15}
\end{align}
Note that $\Exp\{h(\x)\|\x\|^2\}<\infty$, $\Exp\{h^2(\x)\|\x\|^2\}<\infty$, and $|\psi_i(\cdot)|$ are bounded. Thus, from Lemma~\ref{lem1}, if $n/N\rightarrow\rho$, 
\begin{align}
 \Delta_4\rightarrow\label{eq:16} \rho\frac{\Exp\{\psi^2(\bbeta_t)h^2(\x)\x\x\tp\}}{4\Psi^2(\bbeta_t)},
\end{align}
in probability. 

For the term $\Delta_3$ in~\eqref{eq:15}, it is equal to
\begin{align*}
  \Delta_3
  &=\frac{1}{\hat\Psi_0^2}\oneN\sumN|\psi_i(\hbeta_0)|
    \Big\{\frac{n|\psi_i(\hbeta_0)|h(\x_i)}{N}
    \vee\hat\Psi_0\Big\}
    \psi_i^2(\bbeta_t-\hbeta_0)h(\x_i)\x_i\x_i\tp\notag\\
  &=\frac{1}{\hat\Psi_0^2}\frac{n}{N^2}\sumN\psi_i^2(\hbeta_0)
     \psi_i^2(\bbeta_t-\hbeta_0)h^2(\x_i)\x_i\x_i\tp
     I\Big\{\frac{n|\psi_i(\hbeta_0)|h(\x_i)}{N}
     >\hat\Psi_0\Big\}\notag\\
  &\quad+\frac{1}{\hat\Psi_0}\oneN\sumN|\psi_i(\hbeta_0)|
    \psi_i^2(\bbeta_t-\hbeta_0)h(\x_i)\x_i\x_i\tp
    I\Big\{\frac{n|\psi_i(\hbeta_0)|h(\x_i)}{N}
    \le\hat\Psi_0\Big\}.
\end{align*}
Since $\Exp\{h(\x)\|\x\|^2\}<\infty$, $\Exp\{h^2(\x)\|\x\|^2\}<\infty$, and $|\psi_i(\cdot)|$ are bounded, from Lemma~\ref{lem1}, if $n/N\rightarrow\rho$, as $n_0$, $n$, and $N$ go to infinity,
\begin{align}
  \Delta_3\rightarrow
  &\frac{\rho\Exp\big[\psi^2(\bbeta_t)h^2(\x)\x\x\tp
    I\big\{\rho|\psi(\bbeta_t)|h(\x)
    \ge\Psi(\bbeta_t)\big\}\big]}{4\Psi^2(\bbeta_t)}\notag\\
  &+\frac{\Exp\big[|\psi(\bbeta_t)|h(\x)\x\x\tp
    I\big\{\rho|\psi(\bbeta_t)|h(\x)
    \le\Psi(\bbeta_t)\}\big\}\big]}{4\Psi(\bbeta_t)}\notag\\
  =&\frac{\Exp\Big(|\psi(\bbeta_t)|h(\x)\x\x\tp
     \big[\{\rho|\psi(\bbeta_t)|h(\x)\}\vee
     \Psi(\bbeta_t)\}\big]\Big)}{4\Psi^2(\bbeta_t)},\label{eq:17}
\end{align}
 in probability. 
From, \eqref{eq:15}, \eqref{eq:16}, and \eqref{eq:17}, if $n/N\rightarrow\rho$,
\begin{align*}
  \onen
  &\Var\big\{\dot\lambda_p(\bbeta_t)|\Fn,\hbeta_0,\hat\Psi_0\big\}
  =\frac{\Exp\big[|\psi(\bbeta_t)|h(\x)\x\x\tp
    \{\Psi(\bbeta_t)-\rho|\psi(\bbeta_t)|h(\x)\}_+\big]}
    {4\Psi^2(\bbeta_t)}+\op.
\end{align*}
Specifically, when $\rho=0$,
\begin{align*}
  \onen
  &\Var\big\{\dot\lambda_p(\bbeta_t)|\Fn,\hbeta_0,\hat\Psi_0\big\}
  =\bSigma_{\bbeta_t}^{-1}+\op.
\end{align*}

Now we check the Lindeberg-Feller condition \citep[Section $^*$2.8 of][]{Vaart:98} under the condition distribution. Denote $\dot\lambda_{pi}=\delta_i^{\hbeta_0}\{n\pi_i^p(\hbeta_0)\vee1\}
    \psi_i(\bbeta_t-\hbeta_0)\x_i$. 
For any $\epsilon>0$
\begin{align*}
  &\onen\sumN\Exp\Big\{\big\|\dot\lambda_{pi}\big\|^2
    I(\big\|\dot\lambda_{pi}\big\|>\sqrt{n}\epsilon)
    \Big|\Fn,\hbeta_0,\hat\Psi_0\Big\}\\
  &\le\onen\sumN\Exp\Big[
    \big\|\delta_i^{\hbeta_0}\{n\pi_i^p(\hbeta_0)\vee1\}\x_i\big\|^2
    I(\big\|\delta_i^{\hbeta_0}\{n\pi_i^p(\hbeta_0)\vee1\}
    \x_i\big\|>\sqrt{n}\epsilon)\Big|\Fn,\hbeta_0,\hat\Psi_0\Big]\\
  &=\sumN\pi_i^p(\hbeta_0)\{n\pi_i^p(\hbeta_0)\vee1\}\|\x_i\|^2
    I(\{n\pi_i^p(\hbeta_0)\vee1\}\|\x_i\|>\sqrt{n}\epsilon)\\
  &\le\frac{|\psi_i(\hbeta_0)|h(\x_i)\{
    n/N|\psi_i(\hbeta_0)|h(\x_i)+\hat\Psi_0\}\|\x_i\|^2
    I(\{n\pi_i^p(\hbeta_0)+1\}\|\x_i\|>\sqrt{n}\epsilon)}
    {\hat\Psi_0^2}\\
  &\le\frac{\oneN\sumN h^2(\x_i)\|\x_i\|^2
    I(\{h(\x_i)/\hat\Psi_0+1\}\|\x_i\|
    >\sqrt{n}\epsilon)}{\hat\Psi_0^2}\\
  &+\frac{\oneN\sumN h(\x_i)\|\x_i\|^2
    I(\{h(\x_i)/\hat\Psi_0+1\}\|\x_i\|>\sqrt{n}\epsilon)}{\hat\Psi_0}
  =\op,
\end{align*}
where the last equality is from Lemma~\ref{lem1}. Thus, applying the Lindeberg-Feller central limit theorem \citep[Section $^*$2.8 of][]{Vaart:98} finishes the proof.
\end{proof}

\subsubsection{Proof of Lemma \ref{lem6}}
\label{sec:proof-lemma-reflem6}
\begin{proof}{\bf of Lemma~\ref{lem6}.}
  Note that, from Lemma~\ref{lem1},
  \begin{equation*}
    \sumN\pi_i^p(\hbeta_0)\phi_i(\bbeta_t-\hbeta_0)\x_i\x_i\tp
    =\frac{1}{\hat\Psi_0N}\sumN
    |\psi_i(\hbeta_0)|\phi_i(\bbeta_t-\hbeta_0)h(\x_i)\x_i\x_i\tp
    =\bSigma_{\bbeta_t}^{-1}+\op;
  \end{equation*}
  and from the strong law of large numbers
  \begin{equation*}
    \onen\sumN\delta_i^{\bbeta_t}\{n\pi_i^p(\bbeta_t)\vee1\}
    \phi_i(\bbeta_t-\bbeta_t)\x_i\x_i\tp
    =\bSigma_{\bbeta_t}^{-1}+\op,
  \end{equation*}
where $\delta_i^{\bbeta_t}=I\{u_i\le n\pi_i^p(\bbeta_t)\}$. Thus, if we show that
\begin{align*}
  \Delta_5\equiv&\onen\sumN\Big|\delta_i^{\hbeta_0}
    \{n\pi_i^p(\hbeta_0)\vee1\}
    \phi_i(\bbeta_t-\hbeta_0+\s_n)
    -\delta_i^{\bbeta_t}\{n\pi_i^p(\bbeta_t)\vee1\}
    \phi_i(\bbeta_t-\bbeta_t)\Big|\|\x_i\|^2=\op,%
\end{align*}
then the result in Lemma~\ref{lem6} follows. Noting that $\Delta_5$ is nonnegative, we prove $\Delta_5=\op$ %
by showing that $\Exp(\Delta_5|\Fn,\hbeta_0,\hat\Psi_0)=\op$. Note that given $\Fn$, $\hbeta_0$, and $\hat\Psi_0$, the only random terms in $\Delta_5$ are $\delta_i^{\hbeta_0}=I\{u_i\le n\pi_i^p(\hbeta_0)\}$ and $\delta_i^{\bbeta_t}=I\{u_i\le n\pi_i^p(\bbeta_t)\}$. We have that
\begin{align*}
  &\Exp(\Delta_5|\Fn,\hbeta_0,\hat\Psi_0)\notag\\
  \le&
  \onen\sumN
    \{n\pi_i^p(\hbeta_0)\wedge n\pi_i^p(\bbeta_t)\wedge1\}\notag\\
  &\qquad\qquad\times\Big|\{n\pi_i^p(\hbeta_0)\vee1\}
    \phi_i(\bbeta_t-\hbeta_0+\s_n)-\{n\pi_i^p(\bbeta_t)\vee1\}
    \phi_i(\bbeta_t-\bbeta_t)\Big|\|\x_i\|^2\notag\\
 &+\onen\sumN|n\pi_i^p(\hbeta_0)-n\pi_i^p(\bbeta_t)|
    \Big|n\pi_i^p(\hbeta_0)+n\pi_i^p(\bbeta_t)+2\Big|\|\x_i\|^2\notag\\
  \equiv&\Delta_6+\Delta_7.
\end{align*}
Note that $n\pi_i^p(\hbeta_0)\wedge n\pi_i^p(\bbeta_t)\wedge1\le n\pi_i^p(\hbeta_0)$. Thus $\Delta_6$ is bounded by
\begin{align*}
  &\frac{1}{\hat\Psi_0}\oneN\sumN\Big|\{n\pi_i^p(\hbeta_0)\vee1\}
    \phi_i(\bbeta_t-\hbeta_0+\s_n)-\{n\pi_i^p(\bbeta_t)\vee1\}
    \phi_i(\bbeta_t-\bbeta_t)\Big|h(\x_i)\|\x_i\|^2,
\end{align*}
which is $\op$ by Lemma~\ref{lem1} if $|\{n\pi_i^p(\hbeta_0)\vee1\}
 -\{n\pi_i^p(\bbeta_t)\vee1\}|=\op$. This is true because
 \begin{align*}
  |\{n\pi_i^p(\hbeta_0)\vee1\} -\{n\pi_i^p(\bbeta_t)\vee1\}|
 \le&n|\pi_i^p(\hbeta_0)-\pi_i^p(\bbeta_t)|\\
 \le&\frac{nh(\x_i)}{N}\bigg|\frac{|\psi_i(\hbeta_0)|}{\hat\Psi_0}
      -\frac{|\psi_i(\bbeta_t)|}{\Psi_N(\bbeta_t)}\bigg|
      =\op.
 \end{align*}
The term $\Delta_7$ is bounded by
\begin{align*}
  \oneN\sumN\bigg|
    \frac{|\psi_i(\hbeta_0)|}{\hat\Psi_0}
    -\frac{|\psi_i(\bbeta_t)}{\Psi_N(\bbeta_t)}\bigg|
    \bigg|\frac{|\psi_i(\hbeta_0)|}{\hat\Psi_0}
  +\frac{|\psi_i(\bbeta_t)|}{\Psi_N(\bbeta_t)}
  +\frac{2}{h(\x_i)}\bigg|
     h^2(\x_i)\|\x_i\|^2=\op,
\end{align*}
where the last equality is from Lemma~\ref{lem1} and the fact that $\Exp\{h^2(\x)\|\x\|^2\}<\infty$. 
\end{proof}

\subsection{Proofs for unconditional distribution}
\label{sec:proofs-uncond-distr}
In this section we prove Theorem~\ref{thm:3} in
Section~\ref{sec:uncond-distr}. A lemma similar to Lemma~\ref{lem5} is presented below and will be proved later in this section. Lemma~\ref{lem6} can be used in the proof of Theorem \ref{thm:3} because for the problem considered in this paper, convergence to zero in probability is equivalent to convergence to zero in probability under the conditional probability measure \citep{xiong2008some}.

For the pilot subsample taken according to the subsampling probabilities $\pi_{0i}$ in~\eqref{eq:34}, we define $\delta_i^{(1)}=I\{u_{0i}\le\frac{c_0(1-y_i)+c_1y_i}{N}\}$, where $u_{0i}$ are i.i.d. standard uniform random variables. With this notation, the estimator $\hat\Psi_0$ defined in~\eqref{eq:18} can be written as
\begin{align}\label{eq:41}
  \hat\Psi_0&=\oneN\sum_{i=1}^{N}
    \frac{\delta_i^{(1)}|y_i-p(\x_i,\hbeta_0)|h(\x_i)}
    {n\pi_{0i}\wedge1}.
\end{align}

\begin{lemma}\label{lem7}
  Let $\hbeta_0$ and $\hat\Psi_0$ be constructed according to Step 1 of Algorithm~\ref{alg:3}, respectively. For
  \begin{equation*}
  \dot\lambda_p(\bbeta_t)
  =\sumN\delta_i^{\hbeta_0}\{n\pi_i^p(\hbeta_0)\vee1\}
  \{y_i-p(\x_i,\bbeta_t-\hbeta_0)\}\x_i,
\end{equation*}
  under the same assumptions of Theorem~\ref{thm:3}, if $n=o(N)$, then
\begin{align*}
  \frac{\dot\lambda_p(\bbeta_t)}{\sqrt{n}}
  \longrightarrow \Nor\big(\0,\ \bSigma_{\bbeta_t}^{-1}\big),
\end{align*}
in distribution; if $n/N\rightarrow\rho\in(0,1)$, then
\begin{align*}
  \frac{\dot\lambda_p(\bbeta_t)}{\sqrt{n}}
  \longrightarrow \Nor\big(\0,\ \bLambda_{u}\big),
\end{align*}
in distribution. 
\end{lemma}

\begin{proof}{\bf of Theorem \ref{thm:3}.}
  The proof of this theorem is similar to that of Theorem~\ref{thm:2}. The key difference is that Lemma~\ref{lem7} is about asymptotic distribution unconditionally. 
  
  The estimator $\hbeta_p$ is the maximizer of 
\begin{align*}
  \lambda_p(\bbeta)=\sumN\delta_i^{\hbeta_0}\{n\pi_i^p(\hbeta_0)\vee1\}
  \big[(\bbeta-\hbeta_0)\tp\x_iy_i
  -\log\{1+e^{(\bbeta-\hbeta_0)\tp\x_i}\}\big],
\end{align*}
so $\sqrt{n}(\hbeta_p-\bbeta_t)$ is the maximizer of $\gamma_p(\s)=\lambda_p(\bbeta_t+\s/\sqrt{n})-\lambda_p(\bbeta_t)$. By Taylor's expansion,
\begin{align*}
  \gamma_p(\s)
  &=\frac{1}{\sqrt{n}}\s\tp\dot\lambda_p(\bbeta_t)
    +\frac{1}{2n}\sumN\delta_i^{\hbeta_0}\{n\pi_i^p(\hbeta_0)\vee1\}
    \phi_i(\bbeta_t-\hbeta_0+\Acute\s/\sqrt{n})(\s\tp\x_i)^2
\end{align*}
where $\Acute\s$ lies between $\0$ and $\s$.

From {Lemma~\ref{lem6},}
\begin{align*}
 &\onen\sumN\delta_i^{\hbeta_0}\{n\pi_i^p(\hbeta_0)\vee1\}
    \phi_i(\bbeta_t-\hbeta_0+\Acute\s/\sqrt{n})\x_i(\x_i)\tp
    =\bSigma_{\bbeta_t}^{-1}+\op.
\end{align*}
In addition, from Lemma~\ref{lem7}, $\dot\lambda_p(\bbeta_t)/\sqrt{n}$ converges in distribution to a normal limit. Thus, from the Basic Corollary in page 2 of \cite{hjort2011asymptotics}, the maximizer of $\gamma_p(\s)$, $\sqrt{n}(\hbeta_p-\bbeta_t)$, satisfies
\begin{align*}
  \sqrt{n}(\hat\bbeta_p-\bbeta_t)
  =\bSigma_{\bbeta_t}\frac{1}{\sqrt{n}}\dot\lambda_p(\bbeta_t)+\op.
\end{align*}
Combining this with Lemma~\ref{lem7} and Slutsky's theorem, Theorem~\ref{thm:3} follows. 
\end{proof}

\subsubsection{Proof of Proposition~\ref{prop3}}
\begin{proof}{\bf of Proposition~\ref{prop3}}
  To prove \eqref{eq:39}, we just need to show that $\bLambda_{u}\ge\bSigma_{\bbeta_t}^{-1}>\bLambda_{\rho}$. From Proposition~\ref{prop2}, we know that $\bSigma_{\bbeta_t}^{-1}>\bLambda_{\rho}$. To show that $\bLambda_{u}\ge\bSigma_{\bbeta_t}^{-1}$, we notice that
\begin{align*}
\bLambda_{u}
  &=\frac{\Exp[\phi(\bbeta_t)\{
    \rho\phi(\bbeta_t)h(\x)\vee\Phi(\bbeta_t)\}
  h(\x)\x\x\tp]}{4\Phi^2(\bbeta_t)}.\\
  &\ge\frac{\Exp\big\{\phi(\bbeta_t)h(\x)
  \Phi(\bbeta_t)\x\x\tp\big\}}{4\Phi^2(\bbeta_t)}
  =\frac{\Exp\big\{\phi(\bbeta_t)h(\x)\x\x\tp\big\}}{4\Phi(\bbeta_t)}
  =\bSigma_{\bbeta_t}^{-1},
\end{align*}
where the strict inequality holds if $\rho\phi(\bbeta_t)h(\x)\vee\Phi(\bbeta_t)\neq\rho\phi(\bbeta_t)h(\x)$ with positive probability, i.e., $\Pr\{\rho\phi(\bbeta_t)h(\x)>\Phi(\bbeta_t)\}>0$. 
\end{proof}

\subsubsection{Proof of Lemma \ref{lem7}}
\begin{proof}{\bf of Lemma~\ref{lem7}.}

  We first proof the case when the pilot estimates $\hbeta_0$ and $\hat\Psi_0$ depend on the data. 
For any $l\in\mathbb{R}^d$, denote
$\tau_{Ni}= \sqrt{N/n}\hat\Psi_0\delta_i^{\hbeta_0}\{n\pi_i^p(\hbeta_0)\vee1\}
\psi_i(\bbeta_t-\hbeta_0)\x_i\tp l$, $i=1, ..., N$, where $\delta_i^{\hbeta_0}=I\{u_i\le n\pi_i^p(\hbeta_0)\}$, and $u_i$ are i.i.d. standard uniform random variables. Note that $\tau_{Ni}$'s have the same distribution but they are not independent. {Again, since $n_0=o(\sqrt{N})$,
we can focus on $\tau_{Ni}$'s that $(\x_i,y_i)$'s are not included in the pilot subsample. We now exam the mean and variance of these $\tau_{Ni}$'s.} For the mean, based on calculation similar to that in~\eqref{eq:22}, we have, 
\begin{align*}
  &\Exp\big(\tau_{Ni}\big|\hbeta_0,\hat\Psi_0\big)
    =\sqrt{nN}\hat\Psi_0\Exp\big\{\pi_i^p(\hbeta_0)
    \psi_i(\bbeta_t-\hbeta_0)\x_i\tp l\big|\hbeta_0,\hat\Psi_0\big\}
    =\frac{\sqrt{n}\Exp(\eeta_i\big|\hbeta_0,\hat\Psi_0)}
    {\sqrt{N}}=0,
\end{align*}
which implies that
\begin{align*}
  &\Exp\tau_{Ni}=0.
\end{align*} 
For the variance, $\Var(\tau_{Ni})=\Exp(\tau_{Ni}^2)$, we start with the condition expectation,
\begin{align*}
  \Exp\big(\tau_{Ni}^2\big|\hbeta_0,\hat\Psi_0\big)\notag
  =&N\hat\Psi_0^2
     \Exp\Big[\pi_i^p(\hbeta_0)\{n\pi_i^p(\hbeta_0)\vee1\}
     \psi_i^2(\bbeta_t-\hbeta_0)(\x_i\tp l)^2
     \Big|\hbeta_0,\hat\Psi_0\Big]\notag\\
  =&\Exp\Big[|\psi_i(\hbeta_0)|\Big\{
    \frac{n}{N}|\psi_i(\hbeta_0)|h(\x_i)\vee \hat\Psi_0\Big\}
    \psi_i^2(\bbeta_t-\hbeta_0)h(\x_i)(\x_i\tp l)^2
    \Big|\hbeta_0,\hat\Psi_0\Big]\notag.
\end{align*}
If we let
\begin{align*}
  \Upsilon_{Ni}=|\psi_i(\hbeta_0)|\Big\{
    \frac{n}{N}|\psi_i(\hbeta_0)|h(\x_i)\vee \hat\Psi_0\Big\}
    \psi_i^2(\bbeta_t-\hbeta_0)h(\x_i)(\x_i\tp l)^2,
\end{align*}
then $\Var(\tau_{Ni})=\Exp(\Upsilon_{Ni})$.  
Note that
\begin{align*}
  \Upsilon_{Ni}
  &\rightarrow \Upsilon_i=0.25|\psi_i(\bbeta_t)|\{
    \rho|\psi_i(\bbeta_t)|h(\x_i)\vee\Psi(\bbeta_t)\}h(\x_i)(\x_i\tp l)^2,
\end{align*}
in probability. We now show that
\begin{align*}
  \Exp(\Upsilon_{Ni})\rightarrow\Exp(\Upsilon_i)
  &=0.25\Exp[|\psi(\bbeta_t)|\{\rho|\psi(\bbeta_t)|h(\x)
  \vee\Psi(\bbeta_t)\}h(\x)(\x\tp l)^2].
\end{align*}

Let $\Xi_i=|\Upsilon_{Ni}-\Upsilon_i|$. For any $\epsilon$, 
\begin{align*}
  |\Exp(\Upsilon_{Ni})-\Exp(\Upsilon_i)|
  &\le\Exp\{\Xi_i I(\Xi_i >\epsilon)\}+\Exp\{\Xi_i I(\Xi_i \le\epsilon)\}\\
  &\le\Exp\big[\{h^2(\x_i)(\x_i\tp l)^2+\Upsilon_i
    + \hat\Psi_0h(\x_i)(\x_i\tp l)^2\}I(\Xi_i >\epsilon)\big]+\epsilon
\end{align*}
We know that $\Exp\big[\{h^2(\x_i)(\x_i\tp l)^2 +\Upsilon_i\}
I(\Xi_i >\epsilon)\big]\rightarrow0$ since
$\Exp\{h^2(\x_i)(\x_i\tp l)^2 +\Upsilon_i\}<\infty$ for any $l\in\mathbb{R}^d$, and $I(\Xi_i >\epsilon)$ is bounded and is $\op$. Similarly, $\Exp\{\hat\Psi_0h(\x_i)(\x_i\tp l)^2I(\Xi_i >\epsilon)\}\le
\Exp\{h(\x)\}\Exp\{h(\x_i)(\x_i\tp l)^2I(\Xi_i >\epsilon)\}\rightarrow0$. 
Thus, $\Exp(\Upsilon_{Ni})-\Exp(\Upsilon_i)\rightarrow0$, 
and we have finished proving that
\begin{align}\label{eq:31}
  \Var(\tau_{Ni})\rightarrow\Exp(\Upsilon_i).
\end{align}

In the following, we exam the third moment of $\tau_{Ni}$ and prove that 
\begin{align}\label{eq:23}
  \Exp|\tau_{Ni}|^3=o(\sqrt{N}).
\end{align}
For the conditional expectation,
\begin{align*}
  &\Exp\big(|\tau_{Ni}|^3\big|\hbeta_0,\hat\Psi_0\big)\notag\\
  &=N\sqrt{N/n}\hat\Psi_0^3
     \Exp\Big[\pi_i^p(\hbeta_0)\{n\pi_i^p(\hbeta_0)\vee1\}^2
     \psi_i^3(\bbeta_t-\hbeta_0)(\x_i\tp l)^3
     \Big|\hbeta_0,\hat\Psi_0\Big]\notag\\
  &=\sqrt{N/n}\Exp\Big[|\psi_i(\hbeta_0)|\Big\{
    \frac{n}{N}|\psi_i(\hbeta_0)|h(\x_i)\vee \hat\Psi_0\Big\}^2
    \psi_i^3(\bbeta_t-\hbeta_0)h(\x_i)(\x_i\tp l)^3
    \Big|\hbeta_0,\hat\Psi_0\Big]\notag\\
  &\le2\|l\|^3\sqrt{N/n}\Exp\big[\{
    h^2(\x_i)+\hat\Psi_0^2\}h(\x_i)\|\x_i\|^3
    \big|\hbeta_0,\hat\Psi_0\big]\\
  &\le2\|l\|^3\sqrt{N/n}\big[\Exp\{h^3(\x_i)\|\x_i\|^3\}
    +\Exp\{\hat\Psi_0^2|\hbeta_0,\hat\Psi_0\}\Exp\{h(\x_i)\|\x_i\|^3\}\big].
\end{align*}
Since $\Exp\{h^3(\x_i)\|\x_i\|^3\}<\infty$ and $\Exp\{h(\x_i)\|\x_i\|^3\}<\infty$, \eqref{eq:23} follows if $\Exp\{\hat\Psi_0^2\}=O(1)$. 
This is true because
\begin{align}
  \Exp\{\hat\Psi_0^2\}\notag
  &=\Exp\bigg\{\oneN\sum_{k_1=1}^{N}
    \frac{\delta_{k_1}^{(1)}|y_{k_1}-p(\x_{k_1},\hbeta_0)|h(\x_{k_1})}{n_0/N}
    \oneN\sum_{k_2=1}^{N}
    \frac{\delta_{k_2}^{(1)}|y_{k_2}-p(\x_{k_2},\hbeta_0)|h(\x_{k_2})}{n_0/N}
    \bigg\}\notag\\
  &\le\frac{1}{N^2}\sum_{k_1\neq k_2}^N\Exp\bigg\{
    \frac{\delta_{k_1}^{(1)}h(\x_{k_1})}{n_0/N}
    \frac{\delta_{k_2}^{(1)}h(\x_{k_2})}{n_0/N}
    \bigg\}
  +\frac{1}{N^2}\sum_{k=1}^N\Exp\bigg\{
    \frac{\delta_{k}^{(1)}h^2(\x_{k})}{\{n_0/N\}^2}
    \bigg\}\notag\\
  &=\frac{1}{N^2}\sum_{k_1\neq k_2}^{N}\Exp\{
    h(\x_{k_1})h(\x_{k_2})\}
    +\frac{1}{N^2}\sum_{j=1}^N\Exp\bigg\{
    \frac{1}{n_0/N}
    h^2(\x_{k})\bigg\}\rightarrow\Exp\{h^2(\x)\}.\notag
\end{align}

Denote $\nu_{Ni}=\tau_{Ni}\{\Var(\tau_{Ni})\}^{-1/2}$. We know that {$\nu_{Ni}$'s, for which $(\x_i,y_i)$'s are not included in the pilot subsample, are i.i.d. conditional on $\hbeta_0$ and $\hat\Psi_0$. Thus, from Theorem 7.3.2 of \cite{ChowTeicher2003}, they are interchangeable. The fact that $\hbeta_0$ and $\hat\Psi_0$ are consistent estimators implies that they are a sequence of two estimators, and for each $\hbeta_0$ and $\hat\Psi_0$, $\tau_{Ni}$ are interchangeable and can be . For this setup, the central limit theorem in Theorem 2 of \cite{blum1958central} can be applied to prove the asymptotic normality.}

It is evident that $\nu_{Ni}$ have mean 0 and variance 1. It is also easy to verify that, for $i\neq j$,
\begin{align}\label{eq:28}
  &\Exp(\nu_{Ni}\nu_{Nj})
    =\Exp\{\Exp(\nu_{Ni}\nu_{Nj}|\hbeta_0,\hat\Psi_0)\}=0,
\end{align}
and
\begin{align}\label{eq:29}
  &\frac{1}{\sqrt{N}}\Exp\{|\nu_{Ni}|^3\}
    =\Exp|\tau_{Ni}|^3\{\Var(\tau_{Ni})\}^{-3/2}\rightarrow0
\end{align}
which follows from~\eqref{eq:23}. We now show that for $i\neq j$,
\begin{align}\label{eq:30}
  &\Exp\{\nu_{Ni}^2\nu_{Nj}^2\}
    \rightarrow1.
\end{align}

Since $\nu_{Ni}=\tau_{Ni}\{\Var(\tau_{Ni})\}^{-1/2}$, from~\eqref{eq:31}, to prove~\eqref{eq:30}, we only need to show that $\Exp(\tau_{Ni}^2\tau_{Nj}^2)\rightarrow\Exp(\Upsilon_i)\Exp(\Upsilon_j)=\Exp(\Upsilon_i\Upsilon_j)$, where the equality is because $\Upsilon_i$ and $\Upsilon_j$ are independent.
Noting that $\tau_{Ni}^2$ and $\tau_{Nj}^2$ are conditionally independent, we have $\Exp\big(\tau_{Ni}^2\tau_{Nj}^2\big|\hbeta_0,\hat\Psi_0\big)=\Exp\big(\tau_{Ni}^2\big|\hbeta_0,\hat\Psi_0\big)\Exp\big(\tau_{Ni}^2\big|\hbeta_0,\hat\Psi_0\big)=\Exp\big(\Upsilon_{Ni}\big|\hbeta_0,\hat\Psi_0\big)\Exp\big(\Upsilon_{Nj}\big|\hbeta_0,\hat\Psi_0\big)=\Exp\big(\Upsilon_{Ni}\Upsilon_{Nj}\big|\hbeta_0,\hat\Psi_0\big)$, so we know that $\Exp\big(\tau_{Ni}^2\tau_{Nj}^2\big)=\Exp\big(\Upsilon_{Ni}\Upsilon_{Nj}\big)$. 

Now we prove that $\Exp\big(\Upsilon_{Ni}\Upsilon_{Nj}\big)\rightarrow\Exp\big(\Upsilon_{i}\Upsilon_{j}\big)$. 
Let $\Xi_{ij}=|\Upsilon_{Ni}\Upsilon_{Nj}-\Upsilon_i\Upsilon_j|$. For any $\epsilon>0$,
\begin{align}
  |\Exp&(\Upsilon_{Ni}\Upsilon_{nj})-\Exp(\Upsilon_i\Upsilon_j)|\notag\\
  \le&\Exp\{\Xi_{ij}I(\Xi_{ij}>\epsilon)\}
     +\Exp\{\Xi_{ij}I(\Xi_{ij}\le\epsilon)\}\notag\\
  \le&\Exp\big[\{h^2(\x_i)(\x_i\tp l)^2h^2(\x_j)(\x_j\tp l)^2
       +\Upsilon_i\Upsilon_j\}
       I(\Xi_{ij}>\epsilon)\big]\notag\\
  &+\Exp\big[\hat\Psi_0^2h(\x_i)\|\x_i\|^2h(\x_j)\|\x_j\|^2
    I(\Xi_{ij}>\epsilon)\big]\notag\\
  &+\Exp(\hat\Psi_0)\Exp\big[\{h(\x_i)(\x_i\tp l)^2h^2(\x_j)(\x_j\tp l)^2
    +h(\x_j)(\x_j\tp l)^2h^2(\x_i)(\x_i\tp l)^2\}
    I(\Xi_{ij}>\epsilon)\big]+\epsilon\notag%
\end{align}
Since $i\neq j$, $\Exp\{h^2(\x)\|\x\|^2\}<\infty$, $\Exp\{h(\x)\|\x\|^2\}<\infty$, $\Exp(\hat\Psi_0)<\infty$, and $I(\Xi_{ij}>\epsilon)=\op$ is bounded, the above results indicates that \eqref{eq:30} holds. 

Since \eqref{eq:28}, \eqref{eq:29}, \eqref{eq:30} are satisfied, the central limit theorem in Theorem 2 of \cite{blum1958central} holds for $\nu_{Ni}$, which gives that
\begin{align*}
  \frac{1}{\sqrt{N}}\sumN\nu_{Ni}\rightarrow\Nor(0,1),
\end{align*}
in distribution. Note that
\begin{align*}
  \frac{1}{\sqrt{N}}\sumN\nu_i
  &=\frac{\hat\Psi_0}{\sqrt{n}\{\Var(\tau_{Ni})\}^{1/2}}
  \sumN\delta_i^{\hbeta_0}\{n\pi_i^p(\hbeta_0)\vee1\}
  \psi_i(\bbeta_t-\hbeta_0)\x_i\tp l\\
  &=\frac{\hat\Psi_0}{\sqrt{n}\{\Var(\tau_{Ni})\}^{1/2}}
  l\tp\dot\lambda_p(\bbeta_t)
  =\frac{\Psi}{\sqrt{n}\{\Var(\tau_{Ni})\}^{1/2}}
  l\tp\dot\lambda_p(\bbeta_t)+\op.
\end{align*}
Thus, from Slutsky's theorem, for any $l\in\mathbb{R}^d$, 
\begin{align}\label{eq:33}
  &\frac{1}{\sqrt{n}}l\tp\dot\lambda_p(\bbeta_t)
    \rightarrow\Nor(0,l\tp\bLambda_{u}l)
\end{align}
in distribution, where
\begin{align*}
  \bLambda_{u}
  =\frac{\Var(\tau_{Ni})}{\Psi^2(\bbeta_t)}
  &=\frac{\Exp[|\psi(\bbeta_t)|\{
    \rho|\psi(\bbeta_t)|h(\x)\vee\Psi(\bbeta_t)\}
  h(\x)\x\x\tp]}{4\Psi^2(\bbeta_t)}\\
  &\ge\frac{\Exp[|\psi(\bbeta_t)|
    h(\x)\x\x\tp]}{4\Psi(\bbeta_t)}
    =\frac{\Exp[\phi(\bbeta_t)
    h(\x)\x\x\tp]}{4\Phi(\bbeta_t)}
    =\bSigma_{\bbeta_t}^{-1},
\end{align*}
and the equality holds if $\rho=0$, i.e., $n/N\rightarrow0$. Based on \eqref{eq:33}, from the Cram\'{e}r-Wold theorem, we have that
\begin{align*}
  &\frac{1}{\sqrt{n}}\dot\lambda_p(\bbeta_t)
    \rightarrow\Nor(0,\bLambda_{u})
\end{align*}
in distribution.

  When the pilot estimates $\hbeta_0$ and $\hat\Psi_0$ are independent of the data, if we can prove the results in Lemma~\ref{lem7} under the conditional distribution given $\hbeta_0$ and $\hat\Psi_0$, then the result follows unconditionally. We provide the proof under the conditional distribution in the following. The proof is similar to the proof of Lemma~\ref{lem5} and thus we provide only the outline. The difference is we do not conditional on the full data $\Fn$ here.
  
Note that, given $\hbeta_0$ and $\hat\Psi_0$, $\dot\lambda_p(\bbeta_t)$ is a sum of $N$ independent random vectors. We now exam the mean and variance of $\dot\lambda_p(\bbeta_t)$ given $\hbeta_0$ and $\hat\Psi_0$. 
For the mean, 
\begin{align*}
  \frac{1}{\sqrt{n}}\Exp
  \big\{\dot\lambda_p(\bbeta_t)|\hbeta_0,\hat\Psi_0\big\}=\0.
\end{align*}
For the variance, 
\begin{align*}
  \onen\Var\big\{
    \dot\lambda_p(\bbeta_t)|\hbeta_0,\hat\Psi_0\big\}
  =\frac{\Exp\Big[
     |\psi_i(\hbeta_0)|\{n\pi_i^p(\hbeta_0)\vee1\}
     \psi_i^2(\bbeta_t-\hbeta_0)h(\x_i)\x_i\x_i\tp
     \Big|\hbeta_0,\hat\Psi_0\Big]}
    {\hat\Psi_0},
\end{align*}
which, under Assumptions~\ref{as:1} and \ref{as:2}, converges in probability to $\bLambda_{u}$. 

To check the Lindeberg-Feller condition \citep[Section $^*$2.8 of][]{Vaart:98} under the condition distribution, we note that for any $\epsilon>0$,
\begin{align*}
  &\onen\sumN\Exp\Big\{\big\|\dot\lambda_{pi}\big\|^2
    I(\big\|\dot\lambda_{pi}\big\|>\sqrt{n}\epsilon)
    \Big|\hbeta_0,\hat\Psi_0\Big\}\\
  &\le\frac{\Exp\Big[\{h(\x)\|\x\|^2+\hat\Psi_0\}h(\x)
    I(\{h(\x)/\hat\Psi_0+1\}\|\x\|
    >\sqrt{n}\epsilon)\Big|\hbeta_0,\hat\Psi_0\Big]}{\hat\Psi_0^2}
  =\op.
\end{align*}
Thus, applying the Lindeberg-Feller central limit theorem \citep[Section $^*$2.8 of][]{Vaart:98} finishes the proof.

\end{proof}

\subsection{Proofs for cases of misspecifications}
\subsubsection{Proofs with pilot misspecification}
\label{sec:proofs-pilot-missp}
\begin{proof}{\bf of Theorem~\ref{thm:R3}.}
By similar arguments used in the proof of Theorem~\ref{thm:1}, we know that $\sqrt{n}(\hbeta_{uw}-\bbeta_t)$ is the maximizer of
\begin{align*}
  \gamma(\s)
  &=\frac{1}{\sqrt{n}}\s\tp\dot\lambda_{uw}^*(\bbeta_t)
    +\frac{1}{2n}\sumn
    \phi_i^*(\bbeta_t-\hbeta_0+\Acute\s/\sqrt{n})
    (\s\tp\x_i^*)^2,
\end{align*}
where %
$\Acute\s$ lies between $\0$ and $\s$, and 
\begin{equation*}
  \dot\lambda_{uw}^*(\bbeta_t)
  =\sumn\big\{y_i^*-p_i^*(\bbeta_t-\hbeta_0)\big\}\x_i^*.
\end{equation*}

Given $\Fn$ and $\hbeta_0$, $\{y_i^*-p_i^*(\bbeta_t-\hbeta_0)\big\}\x_i^*$ are i.i.d. random vectors. We exam their mean and variance, and check the Lindeberg-Feller condition under the conditional distribution.
For the expectation, direct calculations give
\begin{align}\label{eq:24}
  &\Exp\big[\big\{y^*-p(\x_i^*,\bbeta_t-\hbeta_0)\big\}\x^*
    \big|\Fn,\hbeta_0\big]
  =\sumN\pi_i(\hbeta_0)\psi_i(\bbeta_t-\hbeta_0)\x_i
  =\frac{\sumN\eeta_i}{N\Psi_N(\hbeta_0)},
\end{align}
where $\eeta_i=|\psi_i(\hbeta_0)|\psi_i(\bbeta_t-\hbeta_0)h(\x_i)\x_i$. 
Conditional on $\hbeta_0$, $\eeta_i$'s are i.i.d., and we still have $\Exp(\eeta_i|\x_i,\hbeta_0)=\0$ and thus $\Exp(\eeta_i|\hbeta_0)=\0$ due to \eqref{eq:22}. Thus
\begin{align*}
  \Var(\eeta_i|\hbeta_0)
  =\Exp(\eeta_i\eeta_i\tp|\hbeta_0\big)
  &=\Exp\big\{\psi^2(\hbeta_0)\psi^2(\bbeta_t-\hbeta_0)h^2(\x)\x\x\tp
  \big|\hbeta_0\big\}\\
  &=\Exp\big\{\psi^2(\bbeta_0)\psi^2(\bbeta_t-\bbeta_0)h^2(\x)\x\x\tp
  \big\}+\op
  =\bvsigma_{b}+\op,
\end{align*}
where the third equality is from Lemma~\ref{lem1} and the facts that $\psi^2(\cdot)\le1$ and $\Exp\{h^2(\x)\|\x|^2\}<\infty$.

Similar to the proof of Lemma~\ref{lem2}, the Lindeberg-Feller central limit theorem applies conditional on $\hbeta_0$. Thus, we have that, conditional on $\hbeta_0$, 
\begin{align}\label{eq:25}
  \frac{\sumN\eeta_i}{\sqrt{N}}\longrightarrow \Nor(\0,\ \bvsigma_{b}),
\end{align}
in distribution. 

From \eqref{eq:24} and \eqref{eq:25}, we have
\begin{align}\label{eq:56}
  \Delta=&\Exp\big[\big\{y^*-p(\x_i^*,\bbeta_t-\hbeta_0)\big\}
  \x^*\big|\Fn,\hbeta_0\big]=O_P(1/\sqrt{N}).
\end{align}

For the conditional variance of $\{y_i^*-p_i^*(\bbeta_t-\hbeta_0)\big\}\x_i^*$, using similar approach to the proof of Lemma~\ref{lem3}, we have
\begin{align}
  &\Var\big[\big\{y^*-p(\x_i^*,\bbeta_t-\hbeta_0)\big\}\x^*
    |\Fn,\hbeta_0\big]\notag\\
  =&\sumN\pi_i(\hbeta_0)\psi_i^2(\bbeta_t-\hbeta_0)
     \x_i\x_i\tp-\Delta^2\notag\\
  =&\frac{\oneN\sumN|\psi_i(\hbeta_0)|
    \psi_i^2(\bbeta_t-\hbeta_0)h(\x_i)\x_i\x_i\tp}
    {\Psi_N(\hbeta_0)}-O_P(1/N)\notag\\
  =&\frac{\Exp\{|\psi(\bbeta_0)|
    \psi^2(\bbeta_t-\bbeta_0)h(\x)\x\x\tp\}}
     {\Psi(\bbeta_0)}+\op
  =\frac{\bvsigma_{a}}{\Psi(\bbeta_0)}+\op\label{eq:57}
\end{align}

The Lindeberg-Feller condition under the conditional distribution can be verified similarly to the proof of Lemma~\ref{lem3}.  Thus, we have
\begin{align}\label{eq:26}
  \frac{\dot\lambda_{uw}^*(\bbeta_t)}{\sqrt{n}}
  -\frac{\sqrt{n}\sumN\eeta_i}{N\Psi_N(\hbeta_0)}
  \longrightarrow \Nor\Big\{\0,\ \frac{\bvsigma_{a}}
    {\Psi(\bbeta_0)}\Big\},
\end{align}
in conditional distribution.

From Lemmas~\ref{lem1} and \ref{lem4}, and the law of large numbers, we have
\begin{align*}
  &\onen\sumn\phi_i^*(\bbeta_t-\hbeta_0+\Acute\s/\sqrt{n})
    \x_i^*(\x_i^*)\tp\\
  &=\sumN\pi_i(\hbeta_0)\phi_i(\bbeta_t-\bbeta_0)\x_i\x_i\tp+\op\\
  &=\frac{\oneN\sumN|\psi_i(\hbeta_0)|h(\x_i)
    \phi_i(\bbeta_t-\bbeta_0)\x_i\x_i\tp}{\Psi_N(\hbeta_0)}+\op\\
  &=\frac{\Exp\{|\psi(\bbeta_0)|\phi(\bbeta_t-\bbeta_0)
    h(\x)\x\x\tp\}}{\Psi(\bbeta_0)}+\op
  =\frac{\bvsigma_{a}}{\Psi(\bbeta_0)}+\op.
\end{align*}
Since $\bvsigma_{a}$ is a positive definite matrix, and combining \eqref{eq:24}, \eqref{eq:25}, \eqref{eq:56} and \eqref{eq:57} we know that $\dot\lambda_{uw}^*(\bbeta_t)/\sqrt{n}$ is stochastically bounded, from the Basic Corollary in page 2 of \cite{hjort2011asymptotics}, we have that
\begin{align}\label{eq:58}
  \sqrt{n}(\hat\bbeta_{uw}-\bbeta_t)
  =\Psi(\bbeta_0)\bvsigma_{a}^{-1}\frac{1}{\sqrt{n}}
  \dot\lambda_{uw}^*(\bbeta_t)+\op%
\end{align}
given $\Fn$ and $\hbeta_0$.

Now note that $\sqrt{N}(\hbeta_{\wmle}-\bbeta_t)$ is the maximizer of
  \begin{align*}
    \frac{1}{\sqrt{N}}\s\tp\sumN\eeta_i
    -\frac{1}{2N}\sumN|y_i-p(\x_i,\hbeta_0)|h(\x_i)
    \phi_i(\bbeta_t-\hbeta_0+\Acute\s/\sqrt{N})(\s\tp\x_i)^2
\end{align*}
with $\Acute\s$ between $\0$ and $\s$.
From Lemma~\ref{lem1}, conditional on $\hbeta_0$,  
\begin{align*}
 &\oneN\sumN|y_i-p(\x_i,\hbeta_0)|h(\x_i)
    \phi_i(\bbeta_t-\hbeta_0+\Acute\s/\sqrt{N})\x_i\x_i\tp\\
 &=\oneN\sumN|y_i-p(\x_i,\bbeta_0)|h(\x_i)
    \phi_i(\bbeta_t-\bbeta_0)\x_i\x_i\tp+\op\\
 &=\Exp\{|\psi(\bbeta_0)|
   \phi_i(\bbeta_t-\bbeta_0)h(\x_i)\x_i\x_i\tp\}+\op
  =\bvsigma_{a}+\op
\end{align*}
Thus, from the Basic Corollary in page 2 of \cite{hjort2011asymptotics}, we have
\begin{align}\label{eq:55}
  \sqrt{N}(\hbeta_{\wmle}-\bbeta_t)
  =&\bvsigma_{a}^{-1}
    \frac{1}{\sqrt{N}}\sumN\eeta_i+\op.
\end{align}
Combining this with \eqref{eq:58}, we have
\begin{align*}
  \sqrt{n}(\hat\bbeta_{uw}-\hbeta_{\wmle})
  =\Psi(\bbeta_0)\bvsigma_{a}^{-1}\Big\{\frac{1}{\sqrt{n}}
  \dot\lambda_{uw}^*(\bbeta_t)
  -\frac{\sqrt{n}}{N\Psi(\bbeta_0)}\sumN\eeta_i\Big\}+\op.
\end{align*}

Combining the above two equations with \eqref{eq:26}, Slutsky's theorem, and the fact that a conditional probability is bounded, Theorem \ref{thm:R3} follows.
\end{proof}

\begin{proof}{\bf of Remark~\ref{prop:R1}}
We first observe the following equations by direct calculations. 
\begin{align}
  \Psi(\bbeta_0)
  &=\Exp\big[\{p(\x,\bbeta_t)+p(\x,\bbeta_0)
  -2p(\x,\bbeta_0)p(\x,\bbeta_t)\}h(\x)\big],\notag\text{ and }\\
  p(\x,\bbeta_t-\bbeta_0)
  &=\frac{p(\x,\bbeta_t)\{1-p(\x,\bbeta_0)\}}
  {p(\x,\bbeta_t)\{1-p(\x,\bbeta_0)\}
  +p(\x,\bbeta_0)\{1-p(\x,\bbeta_t)\}}.\label{eq:59}
\end{align}
We need to verify that
\begin{align*}
  \frac{\Exp[\{1-p(\x,\bbeta_t)\}p(\x,\bbeta_0)
  p(\x,\bbeta_t-\bbeta_0)h(\x)\x\x\tp]}
  {\Exp\big[\{p(\x,\bbeta_t)+p(\x,\bbeta_0)
  -2p(\x,\bbeta_0)p(\x,\bbeta_t)\}h(\x)\big]}
  <\frac{\Exp\{\phi(\bbeta_t)h(\x)\x\x\tp\}}
    {4\Exp\{\phi(\bbeta_t)h(\x)\}}
\end{align*}
Note that from \eqref{eq:59}
\begin{align*}
  &2\{1-p(\x,\bbeta_t)\}p(\x,\bbeta_0)p(\x,\bbeta_t-\bbeta_0)
    <\phi(\bbeta_t)\notag\\
  \Leftrightarrow\ 
  &\{p(\x,\bbeta_t)-p(\x,\bbeta_0)\}
    \{1-2p(\x,\bbeta_0)\}>0,%
\end{align*}
and
\begin{align*}
  &p(\x,\bbeta_t)+p(\x,\bbeta_0)-2p(\x,\bbeta_0)p(\x,\bbeta_t)
  >2\phi(\bbeta_t)\notag\\
  \Leftrightarrow\ 
  &\{p(\x,\bbeta_t)-p(\x,\bbeta_0)\}
    \{2p(\x,\bbeta_t)-1\}>0.
\end{align*}
Thus the inequality holds if 
\begin{align*}
  &p(\x,\bbeta_t)>0.5>p(\x,\bbeta_0)
  \quad\text{ or }\quad
  p(\x,\bbeta_t)<0.5<p(\x,\bbeta_0)\notag\\
  \Leftrightarrow\
  & \x\tp\bbeta_t>0>\x\tp\bbeta_0
    \quad\text{ or }\quad
    \x\tp\bbeta_t<0<\x\tp\bbeta_0\notag\\
  \Leftrightarrow\
  & \x\tp\bbeta_t\x\tp\bbeta_0<0.
\end{align*}
This finishes the proof. 
\end{proof}

\begin{proof}{\bf of Theorem \ref{thm:R4}.} 
  Similarly to the proof of Theorem \ref{thm:2}, $\sqrt{n}(\hbeta_p-\bbeta_t)$ is the maximizer of
\begin{align*}  %
   \frac{1}{\sqrt{n}}\s\tp\dot\lambda_p(\bbeta_t)
    +\frac{1}{2n}\sumN\delta_i^{\hbeta_0}\{n\pi_i^p(\hbeta_0)\vee1\}
    \phi_i(\bbeta_t-\hbeta_0+\Acute\s/\sqrt{n})(\s\tp\x_i)^2,
\end{align*}
where $\Acute\s$ lies between $\0$ and $\s$, and
\begin{equation*}
  \dot\lambda_p(\bbeta_t)
  =\sumN\delta_i^{\hbeta_0}\{n\pi_i^p(\hbeta_0)\vee1\}
  \{y_i-p(\x_i,\bbeta_t-\hbeta_0)\}\x_i.
\end{equation*}
Similarly to the proof of Lemma~\ref{lem5}, we first notice that given $\Fn$, $\hbeta_0$, and $\hat\Psi_0$, $\dot\lambda_p(\bbeta_t)$ is a sum of $N$ independent random vectors. We now exam the mean and variance of $\dot\lambda_p(\bbeta_t)$. 
For the mean, we have, 
\begin{align*}
  &\frac{1}{\sqrt{n}}\Exp
  \big\{\dot\lambda_p(\bbeta_t)|\Fn,\hbeta_0,\hat\Psi_0\big\}
    =\frac{\sqrt{n}}{\sqrt{N}}
    \frac{\sumN\eeta_i}{\hat\Psi_0\sqrt{N}}=O_P(\sqrt{n/N}),
\end{align*}
where the last equality is due to \eqref{eq:25}. 

For the variance,
\begin{align*}
  \onen\Var\big\{
    \dot\lambda_p(\bbeta_t)|\Fn,\hbeta_0,\hat\Psi_0\big\}\notag
  =&\frac{\oneN\sumN
     |\psi_i(\hbeta_0)|\{n\pi_i^p(\hbeta_0)\vee1\}
    \psi_i^2(\bbeta_t-\hbeta_0)h(\x_i)\x_i\x_i\tp}
    {\hat\Psi_0}\notag\\
  &-\frac{n}{N}\frac{\oneN\sumN\psi_i^2(\hbeta_0)
    \psi_i^2(\bbeta_t-\hbeta_0)h^2(\x_i)\x_i\x_i\tp}
    {\hat\Psi_0^2}\notag\\
 \equiv&\Delta_{8}-\Delta_{9}
\end{align*}
From Lemma~\ref{lem1}, if $n/N\rightarrow\rho$, using a similar approach used in the proof of Lemma~\ref{lem5}, we have
\begin{align*}
  \Delta_{8}
  &=\frac{1}{\hat\Psi_0^2}\oneN\sumN|\psi_i(\hbeta_0)|
    \Big\{\frac{n|\psi_i(\hbeta_0)|h(\x_i)}{N}
    \vee\hat\Psi_0\Big\}
    \psi_i^2(\bbeta_t-\hbeta_0)h(\x_i)\x_i\x_i\tp\notag\\
  &=\frac{1}{\Psi_0^2}\Exp\Big[|\psi(\bbeta_0)|
    \big\{\rho|\psi(\bbeta_0)|h(\x)
    \vee\Psi(\bbeta_0)\big\}
    \psi^2(\bbeta_t-\bbeta_0)h(\x)\x\x\tp\Big]+\op
\end{align*}
and
\begin{align*}
 \Delta_{9}=\rho\frac{\Exp\{\psi^2(\bbeta_0)
  \psi^2(\bbeta_t-\bbeta_0)h^2(\x)\x\x\tp\}}{\Psi_0^2}
  =\rho\frac{\bvsigma_{b}}{\Psi_0^2}+\op.
\end{align*}
Thus,
\begin{align*}
  &\onen
  \Var\big\{\dot\lambda_p(\bbeta_t)|\Fn,\hbeta_0,\hat\Psi_0\big\}\\
  &=\frac{\Exp\Big[|\psi(\bbeta_0)|
    \big\{1-\Psi_0^{-1}\rho|\psi(\bbeta_0)|h(\x)\big\}_+
    \psi^2(\bbeta_t-\bbeta_0)h(\x)\x\x\tp\Big]}{\Psi_0}+\op
  =\frac{\bvsigma_{c}}{\Psi_0}+\op.
\end{align*}

Applying the Lindeberg-Feller central limit theorem, we have
\begin{align}\label{eq:32}
  \frac{\dot\lambda_p(\bbeta_t)}{\sqrt{n}}
  -\frac{\sqrt{n}\sumN\eeta_i}{N\Psi_N(\hbeta_0)}
  \longrightarrow \Nor\Big(\0,\ \frac{\bvsigma_{c}}{\Psi_0^2}\Big),
\end{align}
in conditional distribution. 

Using a similar approach used to prove Lemma~\ref{lem6}, we have
\begin{align*}
  &\onen\sumN\delta_i^{\hbeta_0}\{n\pi_i^p(\hbeta_0)\vee1\}
    \phi_i(\bbeta_t-\hbeta_0+\s_n)\x_i\x_i\tp\\
  &=\sumN\pi_i^p(\hbeta_0)\phi_i(\bbeta_t-\bbeta_0)
    \x_i\x_i\tp+\op\\
  &=\frac{\oneN\sumN|\psi_i(\hbeta_0)|h(\x_i)
    \phi_i(\bbeta_t-\bbeta_0)\x_i\x_i\tp}{\hat\Psi_0}+\op
    =\frac{\bvsigma_{a}}{\Psi_0}+\op.
\end{align*}
From \eqref{eq:25} and \eqref{eq:32}, $\dot\lambda_p(\bbeta_t)/\sqrt{n}$ is stochastically bounded. In addition, $\bvsigma_{a}$ is finite and positive-definite. Thus, from the Basic Corollary in page 2 of \cite{hjort2011asymptotics}, $\sqrt{n}(\hbeta_p-\bbeta_t)$ satisfies
\begin{align*}
  \sqrt{n}(\hat\bbeta_p-\bbeta_t)
  =\Psi_0\bvsigma_{a}^{-1}\frac{1}{\sqrt{n}}
  \dot\lambda_p(\bbeta_t)+\op,
\end{align*}
given $\Fn$, $\hbeta_0$, and $\hat\Psi_0$. Combining this with \eqref{eq:55}, \eqref{eq:32}, Slutsky's theorem, and the fact that a conditional probability is bounded by one, Theorem~\ref{thm:R4} follows. 
\end{proof}

\subsubsection{Proofs with model misspecification}

\begin{proof}{\bf of Theorem~\ref{thm:R1}.}
  By similar arguments used in the proof of Theorem~\ref{thm:1}, we know that $\sqrt{n}(\hbeta_{uw}-\bbeta_l)$ is the maximizer of 
\begin{align*}
    \frac{1}{\sqrt{n}}\s\tp\dot\lambda_{uw}^*(\bbeta_l)
    +\frac{1}{2n}\sumn
    \phi_i^*(\bbeta_l-\hbeta_0+\Acute\s/\sqrt{n})
    (\s\tp\x_i^*)^2,
\end{align*}
where $\Acute\s$ lies between $\0$ and $\s$, and
\begin{equation*}
  \dot\lambda_{uw}^*(\bbeta_l)
  =\sumn\big\{y_i^*-p(\x_i^*,\bbeta_l-\hbeta_0)\big\}\x_i^*.
\end{equation*}
We abuse the notation and redefine $\eeta_i=|\psi_i(\hbeta_0)|\psi_i(\bbeta_l-\hbeta_0)h(\x_i)\x_i$ in this proof. 
By similar arguments used in the proof of Lemma~\ref{lem3}, we have that 
\begin{align*}
  &\Exp\big[\big\{y^*-p(\x^*,\bbeta_l-\hbeta_0)\big\}\x^*
    \big|\Fn,\hbeta_0\big]
  =\sumN\pi_i(\hbeta_0)\psi_i(\bbeta_l-\hbeta_0)\x_i
  =\frac{\sumN\eeta_i}{N\Psi_N(\hbeta_0)},%
\end{align*}
and
\begin{align*}
  &\Var\big[\big\{y^*-p(\x^*,\bbeta_l-\hbeta_0)\big\}\x^*
    \big|\Fn,\hbeta_0\big]\\
  &=\sumN\pi_i(\hbeta_0)\{y_i-p(\x_i,\bbeta_l-\hbeta_0)\}^2
    \x_i\x_i\tp-\Delta^2\\
  &=\frac{\oneN\sumN|\psi_i(\bbeta_l)|
    (y_i-0.5)^2h(\x_i)\x_i\x_i\tp}{\Psi_N(\bbeta_l)}+\op
  =\frac{\bkappa_{a}}{\omega}+\op.
\end{align*}
The Lindeberg-Feller condition under the conditional distribution can be verified using a similar approached used in the proof of Lemma~\ref{lem3}. Thus, conditional on $\Fn$ and $\hbeta_0$, as $n_0$, $n$, and $N$ go to infinity, 
\begin{align}\label{eq:42}
  \frac{\dot\lambda_{uw}^*(\bbeta_t)}{\sqrt{n}}
  -\frac{\sqrt{n}\sumN\eeta_i}{N\Psi_N(\hbeta_0)}
  \longrightarrow \Nor\Big(\0,\ \frac{\bkappa_{a}}{\omega}\Big),
\end{align}
in conditional distribution. Now we exam $\eeta_i$. 
For the $j$-th element of $\eeta_i$,
\begin{align}
  \eta_{ij}=
  &|\psi_i(\hbeta_0)|\{y_i-p_i(\x_i,\bbeta_l-\hbeta_0)\}
    h(\x_i)x_{ij}\notag\\
  =&(2y_i-1)\{y_i-p_i(\x_i,\hbeta_0)\}
     \{y_i-p_i(\x_i,\bbeta_l-\hbeta_0)\}h(\x_i)x_{ij}\notag\\
  =&0.5(2y_i-1)^2\{y_i-p_i(\x_i,\bbeta_l)\}h(\x_i)x_{ij}
     +\Acute{\eta}_{ij}h(\x_i)x_{ij}\x_i\tp(\hbeta_0-\bbeta_l),
     \label{eq:43}
\end{align}
where
\begin{align*}
  \Acute{\eta}_{ij}=
  &(2y_i-1)\Big[\{y_i-p_i(\x_i,\Acute\bbeta)\}p_i(\x_i,\Acute\s)
    \{1-p_i(\x_i,\Acute\s)\}\notag\\
  &\hspace{3cm}-p_i(\x_i,\Acute\bbeta)\{1-p_i(\x_i,\Acute\bbeta)\}
    \{y_i-p_i(\x_i,\Acute\s)\}\Big],
\end{align*}
with $\Acute\s=\bbeta_l-\Acute\bbeta$, and $\Acute\bbeta$ being between $\bbeta_l$ and $\hbeta_0$. 
Note that $|\dot\eta_{ij}|\le2$ and
\begin{align*}
  \Acute\eta_{ij}\rightarrow
  &\frac{(2y_i-1)}{4}\big[\{y_i-p_i(\x_i,\bbeta_l)\}
  -2(2y_i-1)p_i(\x_i,\bbeta_l)
    \{1-p_i(\x_i,\bbeta_l)\}\big],
\end{align*}
in probability. %
Thus, from Lemma~\ref{lem1} and direct calculations, we have that
\begin{align}
  &\oneN\sumN h(\x_i)(\Acute{\eeta}_i\circ\x_i)\x_i\tp
  =\bkappa_{c}+\op, \label{eq:44}
\end{align}
where $\Acute{\eeta}_i=(\Acute{\eta}_{i1}, ..., \Acute{\eta}_{id})\tp$ and $\circ$ is the Hadamard product. From \eqref{eq:43} and \eqref{eq:44}, 
we have that
\begin{align}
  \sumN\eeta_i
  =&\frac{1}{2}\sumN(2y_i-1)^2\{y_i-p_i(\x_i,\bbeta_l)\}h(\x_i)\x_i
  +\bkappa_{c} N(\hbeta_0-\bbeta_l)\notag
  +o_P\{N(\hbeta_0-\bbeta_l)\}.
\end{align}
Thus, from the central limit theorem and Slutsky's theorem, and the fact that $\hbeta_0$ is independent of $\Fn$, 
\begin{align}\label{eq:45}
  \frac{\sumN\eeta_i}{\sqrt{N}}
  &\longrightarrow \Nor\Big(\0,\ \bkappa_{b}+
    \frac{\bkappa_{c}\bSigma_0\bkappa_{c}}{{\rho_0}}\Big).
\end{align}

From Lemma~\ref{lem1} and using a similar approach to prove Lemma~\ref{lem4}, we have
\begin{align*}
  &\onen\sumn\phi_i^*(\bbeta_l-\hbeta_0+\Acute\s/\sqrt{n})
    \x_i^*(\x_i^*)\tp\\
  &=\sumN\pi_i(\hbeta_0)\phi_i(\bbeta_l-\hbeta_0)\x_i\x_i\tp+\op\\
  &=\frac{\oneN\sumN|\psi_i(\hbeta_0)|h(\x_i)
    \phi_i(\bbeta_l-\hbeta_0)\x_i\x_i\tp}
    {\oneN\sumN|y_i-p(\x_i,\hbeta_0)|h(\x_i)}+\op
  =\frac{\bkappa_{a}}{\omega}+\op.
\end{align*}

Since $\bkappa_{a}$ is a positive definite matrix, and  $\dot\lambda_{uw}^*(\bbeta_l)/\sqrt{n}$ is stochastically bounded due to \eqref{eq:42} and \eqref{eq:45}, from the Basic Corollary in page 2 of \cite{hjort2011asymptotics}, $\sqrt{n}(\hbeta_{uw}-\bbeta_l)$ satisfies
\begin{align}\label{eq:48}
  \sqrt{n}(\hat\bbeta_{uw}-\bbeta_l)
  =\omega\bkappa_{a}^{-1}\frac{1}{\sqrt{n}}
  \dot\lambda_{uw}^*(\bbeta_l)+\op
\end{align}
given $\Fn$ and $\hbeta_0$.

Using similar arguments used in the proof of Lemma~\ref{lem2}, we know that $\sqrt{N}(\hbeta_{\wmle}-\bbeta_l)$ is the maximizer of
\begin{align*}
    \frac{1}{\sqrt{N}}\s\tp\sumN\eeta_i
    -\frac{1}{2N}\sumN|y_i-p(\x_i,\hbeta_0)|h(\x_i)
    \phi_i(\bbeta_l-\hbeta_0+\Acute\s/\sqrt{N})(\s\tp\x_i)^2,
\end{align*}
where $\Acute\s$ lies between $\0$ and $\s$. 
From Lemma~\ref{lem1}, 
\begin{align*}
  &\oneN\sumN|y_i-p(\x_i,\hbeta_0)|h(\x_i)
    \phi_i(\bbeta_l-\hbeta_0+\Acute\s/\sqrt{N})\x_i\x_i\tp
    =\bkappa_{a}+\op.
\end{align*}
Thus, from \eqref{eq:45} and the Basic Corollary in page 2 of \cite{hjort2011asymptotics}, we know that $\sqrt{N}(\hbeta_{\wmle}-\bbeta_l)$ satisfies
\begin{align}\label{eq:46}
  \sqrt{N}(\hbeta_{\wmle}-\bbeta_l)
  =&\bkappa_{a}^{-1}\frac{1}{\sqrt{N}}\sumN\eeta_i+\op.
\end{align}
From \eqref{eq:45}, Slutsky's theorem, and the fact that a conditional probability is bounded by one, the result in \eqref{eq:47} follows. 

From \eqref{eq:48} and \eqref{eq:46}, we have
\begin{align*}
  \sqrt{n}(\hat\bbeta_{uw}-\hbeta_{\wmle})
  &=\omega\bkappa_{a}^{-1}
  \bigg\{\frac{1}{\sqrt{n}}\dot\lambda_{uw}^*(\bbeta_l)
  -\frac{\sqrt{n}\sumN\eeta_i}{\omega N}\bigg\}+\op.
\end{align*}
Thus, from \eqref{eq:42}, Slutsky's theorem, and the fact that a conditional probability is bounded by one, \eqref{eq:49} of Theorem~\ref{thm:R1} follows. 
\end{proof}

\begin{proof}{\bf of Theorem \ref{thm:R2}.}
  Using similar arguments used in the proof Theorem \ref{thm:2}, we know that $\sqrt{n}(\hbeta_p-\bbeta_l)$ is the maximizer of
  \begin{align*}
    \frac{1}{\sqrt{n}}\s\tp\dot\lambda_p(\bbeta_l)
    +\frac{1}{2n}\sumN\delta_i^{\hbeta_0}\{n\pi_i^p(\hbeta_0)\vee1\}
    \phi_i(\bbeta_l-\hbeta_0+\Acute\s/\sqrt{n})(\s\tp\x_i)^2,
\end{align*}
where $\Acute\s$ lies between $\0$ and $\s$, and 
  \begin{equation*}
  \dot\lambda_p(\bbeta_l)
  =\sumN\delta_i^{\hbeta_0}\{n\pi_i^p(\hbeta_0)\vee1\}
  \{y_i-p(\x_i,\bbeta_l-\hbeta_0)\}\x_i.
\end{equation*}

Given $\Fn,\hbeta_0$ and $\hat\Psi_0$, $\dot\lambda_p(\bbeta_l)$ is a sum of independent variables, and the Lindeberg-Feller condition under the condition distribution can be verified similarly to the proof of Lemma~\ref{lem5}. Now we exam the conditional mean and variance of $\dot\lambda_p(\bbeta_l)$. 
For the mean, from \eqref{eq:45} we have, 
\begin{align*}
  &\frac{1}{\sqrt{n}}\Exp
    \big\{\dot\lambda_p(\bbeta_l)|\Fn,\hbeta_0,\hat\Psi_0\big\}
    =\frac{\sqrt{n}}{\sqrt{N}}
    \frac{\sumN\eeta_i}{\hat\Psi_0\sqrt{N}}=O_P(\sqrt{n/N}).
\end{align*}
For the variance,
\begin{align}
  \onen&\Var\big\{
    \dot\lambda_p(\bbeta_l)|\Fn,\hbeta_0,\hat\Psi_0\big\}\notag\\
  =&\frac{\oneN\sumN
     |\psi_i(\hbeta_0)|\{n\pi_i^p(\hbeta_0)\vee1\}
    \psi_i^2(\bbeta_l-\hbeta_0)h(\x_i)\x_i\x_i\tp}
    {\hat\Psi_0}\notag\\
  &-\frac{n}{N}\frac{\oneN\sumN\psi_i^2(\hbeta_0)
    \psi_i^2(\bbeta_l-\hbeta_0)h^2(\x_i)\x_i\x_i\tp}
    {\hat\Psi_0^2}
 \equiv\Delta_{10}+\Delta_{11}. \label{eq:52}
\end{align}
Note that $\Exp\{h(\x)\|\x\|^2\}<\infty$, $\Exp\{h^2(\x)\|\x\|^2\}<\infty$, and $|\psi_i(\cdot)|$ are bounded. Thus, from Lemma~\ref{lem1}, if $n/N\rightarrow\rho$, 
\begin{align}\label{eq:53}
 \Delta_{11}\rightarrow\rho\frac{\bkappa_{b}}{\omega^2},
\end{align}
in probability. 
For the term $\Delta_{10}$ in~\eqref{eq:52}, since $\Exp\{h(\x)\|\x\|^2\}<\infty$, $\Exp\{h^2(\x)\|\x\|^2\}<\infty$, and $|\psi_i(\cdot)|$ are bounded, from Lemma~\ref{lem1}, if $n/N\rightarrow\rho$, as $n_0$, $n$, and $N$ go to infinity, by a similar approach used in the proof of Lemma~\ref{lem5}, %
we have
\begin{align}
  \Delta_{10}
  &=\frac{1}{\hat\Psi_0^2}\frac{n}{N^2}\sumN\psi_i^2(\hbeta_0)
     \psi_i^2(\bbeta_l-\hbeta_0)h^2(\x_i)\x_i\x_i\tp
     I\Big\{\frac{n|\psi_i(\hbeta_0)|h(\x_i)}{N}
     >\hat\Psi_0\Big\}\notag\\
  &\quad+\frac{1}{\hat\Psi_0}\oneN\sumN|\psi_i(\hbeta_0)|
    \psi_i^2(\bbeta_l-\hbeta_0)h(\x_i)\x_i\x_i\tp
    I\Big\{\frac{n|\psi_i(\hbeta_0)|h(\x_i)}{N}
    \le\hat\Psi_0\Big\}\notag\\
  &=\frac{\Exp\Big(|\psi(\bbeta_l)|
    [\{\rho|\psi(\bbeta_l)|h(\x)\}\vee\omega]
    h(\x)\x\x\tp\Big)}{4\omega^2}+\op.\label{eq:54}
\end{align}
From, \eqref{eq:52}, \eqref{eq:53}, and \eqref{eq:54}, if $n/N\rightarrow\rho$,
\begin{align*}
  \onen
  &\Var\big\{\dot\lambda_p(\bbeta_l)|\Fn,\hbeta_0,\hat\Psi_0\big\}
  =\frac{\bkappa_{d}}{\omega}+\op.
\end{align*}

From the above results, conditional on $\Fn$, $\hbeta_0$, and $\hat\Psi_0$, we know that
\begin{align}\label{eq:50}
  \frac{\dot\lambda_p(\bbeta_l)}{\sqrt{n}}
  -\frac{\sqrt{n}\sumN\eeta_i}{N\hat\Psi_0}
  \longrightarrow \Nor\Big(\0,\ \frac{\bkappa_{d}}{\omega^2}\Big),
\end{align}
in distribution. 

In addition, from Lemma~\ref{lem1}, using an approach similar to the proof of Lemma~\ref{lem6}, we have
\begin{align}\label{eq:51}
 &\onen\sumN\delta_i^{\hbeta_0}\{n\pi_i^p(\hbeta_0)\vee1\}
    \phi_i(\bbeta_l-\hbeta_0+\Acute\s/\sqrt{n})\x_i\x_i\tp\notag\\
 &=\frac{1}{4n}\sumN\delta_i^{\bbeta_l}\{n\pi_i^p(\bbeta_l)\vee1\}
   \x_i\x_i\tp+\op\notag\\
 &=\frac{1}{4N\hat\Psi_0}\sumN|\psi_i(\bbeta_l)|h(\x_i)
   \x_i\x_i\tp+\op
 =\frac{\bkappa_{a}}{\omega}+\op.
\end{align}

Thus, based on \eqref{eq:45}, \eqref{eq:50}, and \eqref{eq:51}, from the Basic Corollary in page 2 of \cite{hjort2011asymptotics}, $\sqrt{n}(\hbeta_p-\bbeta_l)$, satisfies
\begin{align*}
  \sqrt{n}(\hat\bbeta_p-\bbeta_l)
  =\omega\bkappa_{a}\frac{1}{\sqrt{n}}\dot\lambda_p(\bbeta_l)+\op,
\end{align*}
given $\Fn$, $\hbeta_0$, and $\hat\Psi_0$. Combining this with \eqref{eq:46}, \eqref{eq:50}, Slutsky's theorem, and the fact that a conditional probability is bounded by one, Theorem~\ref{thm:R2} follows. 
\end{proof}

\color{black}

\vskip 0.2in
\bibliography{logisticref}

\end{document}